\documentclass[11pt]{article}

\usepackage{amssymb,amsmath,amsthm,amsfonts}

\usepackage[T1]{fontenc}
\usepackage[tt=false, type1=true]{libertine}
\usepackage[varqu]{zi4}
\usepackage[libertine]{newtxmath}
\usepackage[margin=1in]{geometry}
\usepackage{enumitem}
\setlist{nosep,topsep=0pt,leftmargin=*}

\usepackage{hyperref}
\hypersetup{
    colorlinks,
    allcolors=teal
}

\usepackage[ruled]{algorithm2e} 
    
    \SetAlFnt{\small}
    \SetAlCapFnt{\small}
    \SetAlCapNameFnt{\small}
    \SetAlCapHSkip{0pt} 
    \IncMargin{-\parindent}
    
\usepackage{tikz,xifthen,physics,xcolor,suffix,thm-restate,nicefrac,enumitem}
\usetikzlibrary{shapes, arrows, positioning}
\usepackage[bb=dsserif]{mathalpha}
\usepackage[capitalize]{cleveref}
   
\definecolor{darkgreen}{rgb}{0,0.5,0}

\renewcommand{\paragraph}[1]{\smallskip\noindent\textbf{#1.}}
\newcommand{\smallparagraph}[1]{\vspace{2pt plus 1pt minus 1pt}\noindent\textit{#1.}}

\DeclareMathOperator*{\E}{\mathbb{E}}
\DeclareMathOperator*{\PR}{\mathbb{P}}
\newcommand{\Ex}[2][]{\E_{\substack{#1}}\left[\mathchoice{\big.}{}{}{}#2\right]}
\newcommand{\ExC}[3][]{\E_{#1}\left[\mathchoice{\big.}{}{}{}#2\middle|#3\right]}
\renewcommand{\Pr}[2][]{\PR_{\substack{#1}}\left[\mathchoice{\big.}{}{}{}#2\right]}

\newcommand{\One}[1]{\mathop{\mathbb{1}}\left[\mathchoice{\big.}{}{}{}#1\right]}
\newcommand{\orderT}[1]{\tilde{\mathcal{O}}\qty(#1)}
\WithSuffix\newcommand\orderT*[1]{\tilde{\mathcal O}(#1)}

\newcommand{\floor}[1]{\left\lfloor{#1}\right\rfloor}
\WithSuffix\newcommand\floor*[1]{\lfloor{#1}\rfloor}
\newcommand{\ceil}[1]{\left\lceil{#1}\right\rceil}
\WithSuffix\newcommand\ceil*[1]{\lceil{#1}\rceil}

\let\a\alpha
\let\e\varepsilon
\let\l\ell
\let\d\delta
\let\sub\subseteq

\newcommand{\opt}{\mathtt{OPT}}
\newcommand{\optAlg}{\ensuremath{\opt^{\mathtt{ALG}}}}
\newcommand{\optInf}{\ensuremath{\opt^{\mathtt{inf}}}}
\newcommand{\rea}{\mathtt{Reach}}
\newcommand{\bb}{\mathbf{b}}
\newcommand{\bmu}{\boldsymbol{\mu}}

\newcommand{\rfunc}{r}
\newcommand{\R}{\mathbb{R}}
\newcommand{\N}{\mathbb{N}}
\newcommand{\calL}{\mathcal{L}}

\newcommand{\calX}{\mathcal{X}}
\newcommand{\calP}{\mathcal{P}}
\newcommand{\calH}{\mathcal{H}}

\newcommand{\SR}{R^{\mathtt{conv}}}
\newcommand{\SC}{C^{\mathtt{conv}}}
\newcommand{\rew}{\ensuremath{\mathtt{REW}}}
\newcommand{\pay}{\ensuremath{\mathtt{PAY}}}
\newcommand{\estbb}[1]{\bb^{T_{#1-1}}}
\newcommand{\vecestbb}[1]{\vec\bb^{T_{#1-1}}}
\newcommand{\poly}{\textrm{poly}}

\newcommand{\Line}[4]{%
    #1&%
    \ifthenelse{\isempty{#2}}{\phantom{=}}{#2}%
    #3%
    \ifthenelse{\isempty{#4}}{}{&&\qquad\left(\big.\substack{#4}\right)}%
}

\newcommand{\ignore}[1]{}
\newcommand{\Description}[1]{}
\newtheorem{theorem}{Theorem}[section]
\newtheorem*{theorem*}{Theorem} 

\newtheorem*{conjecture*}{Conjecture} 
\newtheorem{lemma}[theorem]{Lemma}
\newtheorem{corollary}[theorem]{Corollary}
\newtheorem{proposition}[theorem]{Proposition}

\newtheorem*{claim*}{Claim} 
\newtheorem{definition}[theorem]{Definition}
\newtheorem*{remark}{Remark}
\theoremstyle{definition}
\newtheorem*{definition*}{Definition}

\newtheorem*{example*}{Example}

\newcommand{\citet}[1]{\textcite{#1}}

\title{Learning in Budgeted Auctions with Spacing Objectives}

\author{
Giannis Fikioris%
\thanks{Supported in part by the Department of Defense (DoD) through the National Defense Science \& Engineering Graduate (NDSEG) Fellowship, the Onassis Foundation -- Scholarship ID: F ZS 068-1/2022-2023, the Google PhD Fellowship, and ONR MURI grant N000142412742.}\\
Cornell University\\
\texttt{gfikioris@cs.cornell.edu}
\and
Robert Kleinberg
\\
Cornell University\\
\texttt{rdk@cs.cornell.edu}
\and
Yoav Kolumbus\\
Cornell University\\
\texttt{yoav.kolumbus@cornell.edu}
\and
Raunak Kumar\\
Cornell University and Microsoft\\
\texttt{raunak@cs.cornell.edu}
\and
Yishay Mansour%
\thanks{Partially supported by funding from the European Research Council (ERC) under the European Union’s Horizon 2020 research and innovation program (grant agreement No. 882396), by the Israel Science Foundation, the Yandex Initiative for Machine Learning at Tel Aviv University and a grant from the Tel Aviv University Center for AI and Data Science (TAD).}\\
Tel Aviv University and Google Research\\
\texttt{mansour.yishay@gmail.com}
\and
\'Eva Tardos%
\thanks{Partially supported by AFOSR grant FA9550-23-1-0410, AFOSR grant FA9550-231-0068, and ONR MURI grant N000142412742}\\
Cornell University\\
\texttt{eva.tardos@cornell.edu}
}

\date{\vspace{-25pt}}

\begin{document}


\maketitle{}
\thispagestyle{empty}

\begin{abstract}
    In many repeated auction settings, participants care not only about how frequently they win but also about how their winnings are distributed over time. This problem arises in various practical domains where avoiding congested demand is crucial and spacing is important, such as advertising campaigns that require sustained visibility over time, as well as in online retail sales and compute services. We initiate the study of repeated auctions with preferences for spacing of wins over time.
We introduce a simple model of this phenomenon, modeling it as a budgeted auction where the value of a win is given by any concave function of the time since the last win.
This formulation captures our initial modeling goal: for a given number of wins, even spacing over time is optimal and adding further wins to any sequence is always beneficial.
We also extend our model and results to the case when not all wins result in ``conversions'' (realization of actual gains), and the probability of conversion depends on a context. The goal is then to maximize and evenly space conversions rather than just wins.

We study the optimal policies for this setting in repeated second-price auctions and offer learning algorithms for the bidders that achieve low regret against the optimal bidding policy in a Bayesian online setting. Our main result is
a computationally efficient online learning algorithm that achieves $\tilde O(\sqrt{T})$ regret. 
We achieve this by showing that an infinite-horizon Markov decision process (MDP) with the budget constraint in expectation is essentially equivalent to our problem, \emph{even when limiting that MDP to a very small number of states}.
The algorithm achieves low regret by learning a bidding policy that chooses bids as a function of the context and the state of the system, which will be the time elapsed since the last win (or conversion).
We show that state-independent strategies incur linear regret even without uncertainty of conversions. We complement this by showing that there are state-independent strategies that, while still having linear regret, achieve a $(1 - \frac{1}{e})$ approximation to the optimal reward as $T \rightarrow \infty$.

\vspace{8pt}
\begin{quote}
    {\em ``The only reason for time is so that everything doesn’t happen at once.''}
    \begin{flushright}
        — Albert Einstein
    \end{flushright}
\end{quote}

\end{abstract}

\clearpage
\setcounter{page}{1}
\section{Introduction} \label{sec:intro}

Auctions are a cornerstone of economic theory, illustrating one of the earliest structured forms of market interaction. 
Today, auctions play a central and ever-increasing role in the digital landscape, spanning domains such as online advertising \cite{choi2020online,edelman2007internet, lewis2014online, varian2009online}, retail markets \cite{chu2020position, sun2020empirical}, and blockchain fee markets \cite{basu2023stablefees, ferreira2021dynamic, nisan2023serial, roughgarden2024transaction}, and are studied extensively.
A prominent theme in online auction applications and the ensuing theoretical studies is that a bidder (e.g., a seller on an online marketplace like Amazon or an advertiser bidding for ad positions) participates in the auction multiple times.
Hence, a bidder needs to consider how to utilize their budget efficiently and how to learn from experience over time. 
Typically, the bidder engages in these markets by using some learning algorithm to set the bid in each auction period \cite{aggarwal2024auto, DBLP:journals/corr/AggarwalFZ24, DBLP:journals/mansci/BalseiroG19, banchio2022artificial, bichler2023convergence, borgs2007GFPdynamics, daskalakis2016learning, deng2022nash, fikioris2023liquid, kolumbus2022auctions, KolumbusN22, nekipelov2015econometrics, noti2021bid, perlroth2023auctions, weed2016online}.

The literature on repeated auctions generally assumes that at each step in a sequence of $T$ auctions, a bidder has some value for the item being auctioned (for example, an online ad impression). The bidder's total utility for participating in the sequence
of auctions is modeled as a function (often additive) of the \emph{set} of items they won. 
Bidders' preferences in this standard model are \emph{permutation-invariant}: the bidder is indifferent between
all permutations of a given (multi)set of auction outcomes. This assumption of permutation invariance underpins
existing work in this area, but it may be violated in many important cases, as we now discuss.

To see the flaw with this assumption, consider first the preferences of an advertiser running 
a campaign to build brand recognition. A metric of clear importance to the advertiser  
(and the one past models aim to maximize) is the number of ad impressions, i.e.~the number of times
they win in the repeated auction. However, the value of winning a given number of impressions depends
on how they are distributed over time, due to carry-over effects~\cite{weinberg1982econometric} that 
are widely documented in the marketing literature~\cite{broadbent1993advertising, broadbent1995adstock, broadbent2000advertisements, craig1976advertising, ha1997does}. For example, in studying how the effectiveness of 
advertising varies with the frequency of exposure, Naples~\cite{naples1997effective} observes that
``increasing frequency continues to build advertising effectiveness at a decreasing rate but with no evidence of a decline.''
The assumption of permutation-invariant preferences is also violated by
a retailer whose inventory is restocked
at a limited rate unless they pay extra charges to expedite restocking. 
If ad impressions and the resulting orders
are congested in time, the retailer may be unable to fulfill the orders at
the requested rate without incurring surcharges, and would thus 
prefer if the same sequence of auction outcomes were permuted to space 
their ad impressions more evenly over time.

In this paper, we initiate the study of learning in repeated auctions where bidders care not only about {\em how much they win} but also {\em when they win}.  
Specifically, we focus on scenarios in which bidders prefer that their winning times in the series of auctions---whose total number is limited by budget constraints---are not clustered together but are relatively evenly spaced over time. 
We propose a simple model to describe such temporal preferences, which allows us to analyze the optimal strategies for the bidders once they know the distribution of prices and contexts, as well as the algorithms they could use in a more natural scenario requiring them to learn to bid optimally online.
Next, we describe our model and continue with an illustrative example and an overview of our results and techniques.
We postpone a more extended discussion of related work to \cref{sec:related}. 

\paragraph{Model overview}
An informal summary of our model is the following. (For the formal model, see Section \ref{sec:model}.)
We first consider the standard model of budgeted auctions: a bidder participates in a sequence of $T$ second-price\footnote{One could consider the same spacing model for first-price or other auctions formats. We believe this is an interesting line of future work.} auctions with budget $B$; in each auction, the distribution of highest bid of the competitors (i.e., the price) is in $[0, 1]$ and unknown to the bidder.
We diverge from the standard model in the bidder's objective.
As outlined above, our bidder is interested, on the one hand, in increasing the number of winning events and, on the other hand, in spacing those wins over time.
Our proposed model combines these two goals by using an increasing concave reward function $\rfunc: \N \to \R_{\ge 0}$, e.g., $\rfunc(\l) = \sqrt\l$.
Specifically, the bidder's utility in the entire sequence of $T$ auctions is the sum of $\rfunc(\cdot)$ applied to the lengths of intervals between winning events.

We now comment on our choice of possible $\rfunc(\cdot)$ functions.
The three properties---non-negativity, concavity, and monotonicity---are motivated by the examples given above.
Specifically, all three properties are described in \cite{naples1997effective}: increasing frequency leads to increased benefits at a decreasing rate with no adverse effects. With the decreased effectiveness of additional advertisements too close together, evenly spaced advertising is of higher value.
In addition, we capture the heterogeneity of the influence of spacing in different markets and products, by considering the problem under any function $\rfunc(\cdot)$ that satisfies the above properties.
Finally, in the simplified setting where, instead of a standard budget constraint, the bidder is constrained to win at most $B$ times,\footnote{This simplified setting can be achieved in the more general one by assuming that the top competing bid is $1$ every round.} for any such function $\rfunc(\cdot)$, the utility-maximizing solution is to win every $\approx B/T$ rounds, i.e., maximizing the number of wins and equally spacing them.

We extend our model and results to a setting where auctions are not identical, but each auction has a context and not every winning event leads to a successful conversion event (realization of actual gains). Conversion rates and prices depend on the context and may be correlated. In this extended setting, the goal is to achieve many conversions and have them spaced relatively evenly. 
We note that this generalization reduces to the classical auction setting with stochastic values when the spacing of wins is not part of the objective (e.g., $\rfunc(\l) = 1$ for all $\l \ge 1$) and the probability of conversion serves as the ``value'' of winning.
  
\paragraph{Warm-up example}
To gain some intuition about our model, we consider the following simple scenario: the distribution of prices every round is uniform in $[0, 1]$, known to the bidder, and all wins lead to a conversion.
The bidder has a budget $B \leq T/4$, and the reward function is $\rfunc(\l) = \sqrt{\l}$.
We are interested in the long-term utility, where $T \gg 1$, and $\rho = B/T$ is constant. 

How should our bidder act in this auction? To illustrate the challenge of optimizing our bidder's objective, we begin by reviewing two simple strategies and the performance they can achieve.

\smallparagraph{Fixed intervals}
One possible approach is to enforce equal spacing: win with probability $1$ once every fixed number of steps to fully use the budget in expectation.
This means bidding $1$ every $\frac{T}{2B}$ steps and paying an average price of $1/2$ per win. 
In expectation,%
\footnote{The utility when there are $K$ wins $\nicefrac{T}{2B}$ apart from each other is $u = K \sqrt{\nicefrac{T}{2B}}$.
The total cost $C_k$ after $k$ wins follows a uniform-sum (Irwin–Hall) distribution.
For large $T$, with $B$ being a constant fraction of $T$, $C_k$ is concentrated around $k/2$, and so $C_{2 B} \approx B$ implying $K \approx 2B$ and $u \approx 2B \sqrt{\nicefrac{T}{2B}}$ with high probability.}\label{footnote-uniform-sum}
this leads to $2B(1 - o(1))$ wins, and the total utility for our bidder is approximately $\sqrt{\nicefrac{T}{2B}} \cdot 2B = \sqrt{2\rho} \, T$.

\smallparagraph{Fixed bid}
An alternative simple strategy would be to use fixed bidding that maximizes the number of wins and ignores spacing.
This is achieved by making the lowest possible payments while still using the full budget in expectation.
This, in fact, can be obtained by consistently using a fixed bid level%
\footnote{A fixed bid achieves this because the price distribution in our example has full support and no point masses. In general, an optimal bid distribution can be a randomization between multiple bids.}
$b$.
The probability of winning is then $b$, and the expected payment per win is $b/2$.
As in our previous example, the budget usage with a fixed bid $b$ follows a uniform-sum distribution (see footnote \ref{footnote-uniform-sum}), which is concentrated around $b/2$ per step.
Thus, the lowest bid that fully uses the budget is found by $\frac{b}{2} \cdot b T = B \Rightarrow b = \sqrt{2B/T} = \sqrt{2 \rho}$.
Therefore, the maximum expected number of wins is $\sqrt{2 \rho} \, T$.
Note that since $\rho < 1/2$,  the bidder wins more frequently than in the fixed-interval strategy. 
If all intervals were of equal length, i.e., spaced $\nicefrac{1}{\sqrt{2 \rho}}$ apart, the utility would be $(2\rho)^{1/4}T$, which is an upper bound on the maximum utility (i.e., maximum expected number of wins with perfect spacing). 
However, the actual resulting utility, due to the randomness of the intervals, ends up being about $\qty\big(1.05\rho^{1/4} + O(\rho^{3/4}) )T$ (the full calculation is deferred to \cref{appendix:warm-up}).

\smallparagraph{Comparison}
The above calculations show that fixed bidding gives higher utility than the ``fixed intervals'' strategy.
This shows that, in our proposed model with the $r(\ell)=\sqrt{\ell}$ reward function, more wins is of higher importance than ensuring exact spacing.
Specifically, fixed bidding achieves a constant-factor approximation to the optimal utility for any $\rho$, while the fixed interval strategy is increasingly suboptimal as $\rho$ decreases.
\cref{thm:state-independent} shows that this result holds more broadly:  for any concave reward function and any price distribution, there exists a static bidding policy that achieves at least $(1 - \frac{1}{e})$ of the optimal policy’s value as $T \to \infty$. We prove this in \cref{appendix:reverse-jensen}, using a `reverse Jensen inequality.'
However, because fixed bidding strategies, e.g., \cite{DBLP:journals/mansci/BalseiroG19}, completely ignore spacing, they are not optimal: in \cref{ssec:app:state_independent_optimality}, we prove they can be at least $10\%$ suboptimal.
Therefore, one can do better by combining the two ideas: bidding to ensure high probability of winning while also considering the time since the last win.

Next, we provide a brief summary of our results, followed by an overview of our main challenges, techniques, and proof outlines where we analyze the problem of bidding with spacing objectives and design algorithms that learn the above optimal policy in an online fashion.
Related work is discussed in Section \ref{sec:related}, and the formal presentation of our model appears in Section \ref{sec:model}. The subsequent sections include the full formal analysis.

\subsection*{Our Results}

Considering repeated auctions where the bidder cares about spacing wins relatively evenly raises important new questions.
In this paper, we offer a general model for this problem using second-price auctions where the value of winning is a concave function of the time since the last win, and we develop an effective learning algorithm to address it. 

We offer an efficient online algorithm for the bidder that achieves regret at most $\orderT*{\sqrt{T}}$ after $T$ time steps with high probability against the optimal bidding policy (\cref{cor:online:final_res}). 
Optimal bidding in this problem depends on the time since the last winning event, $\l$, as well as the context, $x$ and so one would assume that with a context space $\calX$, the learner would have to learn  $|\calX|\cdot T$ different bids, one for each pair of $(\l, x)$.
Therefore, it is surprising that $\orderT*{\sqrt T}$ regret, independent of $\calX$, is achievable since if we were to treat this pair as a joint context space of cardinality $d = |\calX|\cdot T$, one might expect $\Omega(\sqrt{d T})$ regret similar to contextual bandits.
While we solve a different (and harder) problem for the bidder, we still achieve the same $\orderT*{\sqrt T}$ regret rate as in the case of maximizing only the cumulative value of wins \cite{DBLP:journals/mansci/BalseiroG19}.

An interesting feature of our analysis is that we do not need to discretize the bidding space for the learning algorithm, avoiding issues related to discretization errors.

We also show that state-independent strategies (i.e., independent of the time since the last win) incur linear regret. See \cref{ssec:app:state_independent_optimality}. 
On the positive side, we show that such policies can achieve a $(1 - \frac{1}{e})$ approximation to the optimal reward as $T \rightarrow \infty$. (\cref{thm:state-independent}.)

\subsection*{Outline and Techniques for the Regret Bound}

In Section \ref{sec:bench}, we first notice that a large dynamic program can compute the optimal bidding policy under mild assumptions.
The dynamic program aims to select the best bid for each time\footnote{We denote $[n] = \{1, 2, \ldots, n\}$ for $n \in \N$.} $t\in [T]$, depending on the time since the last conversion and the remaining budget.
However, such a dynamic program with (loosely) $\Omega(T^3)$ ``states'' is not amenable as a starting point for the design or analysis of learning policies online, as it requires knowledge of the price/context distributions ahead of time.

\paragraph{An equivalent small MDP}
Our main result towards the online learning algorithm we eventually present in \cref{sec:online} is \cref{thm:bench:small_m}, offering a stronger, but much more compact benchmark, exponentially reducing the complexity of the learning problem at essentially no cost. We define an infinite-horizon Markov Decision Process (MDP) with only $m =  \order*{\log T}$ states where state $\l  \in [m]$ corresponds to the time elapsed since the last conversion, capped at $m$.
See Figure \ref{fig:mc} for a visual representation of the MDP.
We show that the average reward of this much smaller MDP is almost equal to the original problem with minimal $\order*{1/\poly(T)}$ error.

\smallparagraph{Proving the equivalence of the small MDP and the full problem}
To see how we define the approximately equivalent infinite-horizon MDP with a small state space, it is best to first consider a larger, $T$-state, infinite-horizon MDP whose states are recording the time since the last win. Here, we only require that the per-step budget is observed on average, as the time horizon approaches infinity (note that the time horizon is not $T$, but infinite, in this setting) and in expectation over the prices and contexts.
See Figure \ref{fig:mc} with $m = T$, where $W_\l$ is the winning probability of a particular bidding strategy in state $\l$.
Given the relaxation of the budget constraint, we show that the optimal reward of the constrained MDP is at least the value of the optimal strategy in the original auction. 

\begin{figure}[t!]
    \centering
    \begingroup
\newcommand{\BELOW}{.5cm}
\newcommand{\RIGHT}{1.7cm}
\begin{tikzpicture}[->, >=stealth', line width=.9pt]
    \footnotesize
    \node (s1) [draw] {$1$};
    \node [draw, right=\RIGHT of s1] (s2) {$2$};
    \node [draw, right=\RIGHT of s2] (s3) {$3$};
    \node [right=\RIGHT of s3] (dots) {$\dots$};
    \node [draw, right=\RIGHT of dots] (sm1) {$m-1$};
    \node [draw, right=\RIGHT of sm1] (sm) {$m$};

    \node [draw, dashed, below=\BELOW of s2] (b2) {$1$};
    \node [draw, dashed, below=\BELOW of dots] (bdots) {$1$};
    \node [draw, dashed, below=\BELOW of sm1] (bm1) {$1$};

    \path(s1) edge node [ above ] {$1 - W_1$} (s2);
    \path(s2) edge node [ above ] {$1 - W_2$} (s3);
    \path(s3) edge node [ above ] {$1 - W_3$} (dots);
    \path(dots) edge node [ above ] {$1 - W_{m-1}$} (sm1);
    \path(sm1) edge node [ above ] {$1 - W_m$} (sm);

    \path (sm) edge [ loop right ] node [ right ] {$1 - W_m$} (sm);

    \path (s1) edge [ loop left ] node [ left ] {$W_1$} (s1);
    \path (s2) edge [ bend left=60 ] node [ below ] {$W_2$} (s1);
    \path (s3) edge [ bend left=30 ] node [ below ] {$W_3$} (b2);
    \path (sm1) edge [ bend left=30 ] node [ below ] {$W_{m-1}$} (bdots);
    \path (sm) edge [ bend left=30 ] node [ below ] {$W_m$} (bm1);
    
\end{tikzpicture}
\endgroup
    \Description{The Markov Chain with states being the time since the last win.}
    \caption{The MDP of the infinite-horizon setting when the bidder wins with probability $W_\l$ when $\l \le m$ rounds have elapsed since her last win and with probability $W_m$ when more than $m$ have elapsed.}
    \label{fig:mc}
\end{figure}
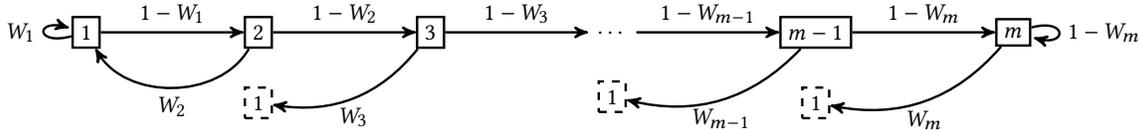

The main and most technically demanding result of \cref{sec:bench} is that the $T$-state MDP can be closely approximated by the small MDP with only $\order*{\log T}$ states%
\footnote{A different idea for compressing MDPs is forming a net to group different states and using the same action in them.
However, this proposal appears to have two issues.
First, this compression loses the MDP structure as transitions between states of the same group are not observed.
Second, when the budget per round is $\rho$, the first $\order*{1/\rho}$ states are each visited a constant fraction of the time, implying that suboptimal bidding in any of these would result in high error.
Therefore, each of these states (or potentially more) would have to be in their own group, leading to a solution close to our own where we just group states $\ge \order*{\frac{1}{\rho} \log T}$.}
(Theorem \ref{thm:bench:small_m}).
This approximation, in turn, has value at least as large as the value of the true optimal solution obtained by the dynamic program above (Lemma \ref{lem:bench:comparison} and \cref{thm:bench:comparison}).
To compress the MDP to have only $\order*{\log T}$ states, we assume that the average per-step budget $\rho$ is a constant independent of $T$. 
Under this assumption, an optimal policy is expected to win on average every constant number of steps.
The key idea to be able to show this is \cref{cor:bench:increasing} showing that the probability of winning increases as a function of the last win.
Using this, we show in \cref{lem:bench:win_LB} that after a constant number of losses, the probability of the optimal policy winning an auction is at least a constant.
This implies that there is a win every $\order*{\log T}$ rounds with high probability. 

The first step behind the proof of monotonicity is to focus on winning probabilities in each state rather than bidding policies.
Because bidding policies can be complicated, the winning probability is a simple one-parameter summary of the bidding policy.
This simplifies the analysis and allows us to prove the result in both contextual and non-contextual settings.
Given this simplification, our key result is \cref{cor:bench:increasing}: there exists an optimal policy (for any reward function $\rfunc(\cdot)$ and distribution of contexts and prices) which wins more frequently as the time since the last win increases ($W_{\l+1} \ge W_\l$ in Figure~\ref{fig:mc}).
This, together with the constant per-step budget, implies that after a constant number of losses, the probability of winning in each state is at least constant, as otherwise, the bidder would not be using her budget.
This monotonicity is necessary to rule out policies that use the budget while winning when the state $\l$ is either small or large and do not win in intermediate states, resulting in $\poly(T)$ states being important instead of $\log T$.
We prove this monotonicity by considering the Lagrangian relaxation of the budget-constrained problem (subtracting $\lambda$ times the budget spent from the reward with a parameter $\lambda \ge 0$).
With this relaxation of the budget constraint, we show that for any $\lambda$, the optimal winning probabilities are increasing in the time since the last win (\cref{lem:bench:incr}).
The claim for the optimal solution of the constrained problem then follows as it can be written as a linear program; therefore, its optimal solution is also optimal for the Lagrangian problem using the value of the variable associated with the budget constraint from the optimal solution of the dual problem.

\paragraph{Outline of the learning algorithm}
\cref{sec:online} offers our online algorithm for bidding: Follow the $k$-delayed Optimal Response Strategy (FKORS). (See \cref{algo:online}.)
Our algorithm relies on full-information feedback, i.e., after bidding it observes the price of that round.
Due to the model compression we made to our benchmark in \cref{sec:bench}, FKORS is fairly straight-forward.
We use a form of episodic learning over subsequences of the (single) $T$-length sequence to learn the best bidding in each state of the small MDP. 
We divide time into small epochs with stochastic length, with epochs ending either when conversions occur or after having no conversions for $k$ steps for a parameter $k \in \N$.
At the start of each epoch, we use data from previous rounds to estimate the optimal bidding strategy in the small MDP.
An interesting feature of this learning algorithm is its use of epochs with stochastic lengths.
If epochs' lengths were fixed (as is usually the case in this literature), most of them would not end immediately after a conversion, which would cause additional error.
It is important to use such variable-length epochs to avoid accumulating this error.
First, short epochs help the learning converge faster.
Second, we show that an epoch ending without a conversion will only occur with very low probability, leading to only a minor impact on the overall reward.
In the remark after \cref{cor:online:final_res} stating the final regret bound in \cref{sec:online}, we show that using fixed-length epochs would lead to $\Omega(T^{2/3})$ regret bounds.

An additional appealing feature of this algorithm is that we do not discretize the bidding space; instead, we can compute the optimal bidding policy (for the samples of prices and contexts seen so far), as explained below.
Importantly, this allows us to have no dependence on the size of the action space and avoid discretization errors altogether.

\paragraph{Optimal mapping from contexts to bids}
By showing the (near) equivalence of the small MDP that has $m = \order*{\log T}$ states to the larger problem, the number of states is not a hurdle for our learning algorithm.
However, we still have to learn a function that maps contexts to bids in each of the $m$ states.
We show that optimal bidding has a simple structure, with only one parameter per state that our online algorithm must learn; this is the key to getting a regret bound independent of the context space.
The key observation is to think of the main variable of the MDP as the probability of winning a conversion $W_\l$ in state $\l \in [m]$ (rather than the bids directly). 
To optimize the bidding policy, the bidder wants to use a bidding function from context to bids that achieves the desired winning probability with minimal expected cost.

The main observation here is that the optimal bidding policy can be expressed as a simple function of the conversion rate, regardless of other parts of the context, i.e., the potential information about the distribution of prices associated with that context. 
Depending on the state of the system (time since the last conversion), the optimal policy at that state when the conversion rate is $c$ is to bid $\min(c/\mu, 1)$ for some $\mu > 0$ (see \cref{lem:calc:simple_single} for the statement for price distributions with atoms).
We again prove this using a Lagrangian relaxation, this time for the problem of maximizing the conversion probability subject to the expected payment.
Finally, in \cref{lem:calc:sample_error}, we show that we can learn this mapping with low error using samples, i.e., we provide a Glivenko-Cantelli style bound for the error of learning the mappings from $\mu$ to the probability of conversion and expected payment.

\paragraph{Computation}
The optimal policy for our MDP with $m = \order*{\log T}$ states can be computed using a Linear Program.
We show in \cref{ssec:app:calc:lp} that, if the support of the joint distribution of prices and conversion rates is at most $K \in \N\cup\{\infty\}$, then the size of the LP is at most $\order*{m \min\{K, T\}}$.

\section{Further Related Work}
\label{sec:related}

While there is a vast body of work on online learning and on algorithmic aspects in auctions (see \cite{cesa2006prediction, DBLP:journals/corr/Slivkins19} and \cite{krishna2009auction,milgrom2021auction,MAL-077,roughgarden2010algorithmic}, respectively, for broad introductions), and in particular, budgeted auctions have been extensively explored (e.g., \cite{DBLP:journals/corr/AggarwalFZ24, balseiro2015repeated, DBLP:journals/mansci/BalseiroG19, borgs2005multi, dobzinski2012multi, dobzinski2014efficiency, feldman2012revenue,feldman2013limits, fikioris2023liquid, DBLP:conf/colt/LucierPSZ24}), to our knowledge, the problem of online learning with a temporal spacing objective has not been previously studied.

\paragraph{No-regret learning in auctions} This theme has attracted broad interest, with research exploring topics such as social welfare and the price of anarchy when bidders are regret-minimizers \cite{DBLP:journals/corr/AggarwalFZ24, DBLP:journals/mansci/BalseiroG19, blum2008regret, caragiannis2015bounding, daskalakis2016learning, fikioris2023liquid, roughgarden2017price, sundararajan2016prediction}; the learning of auction parameters, such as reserve prices \cite{cesa2014regret, mohri2014optimal, morgenstern2016learning, roughgarden2019minimizing}; the dynamics of no-regret algorithms in auction settings \cite{bichler2023convergence, daskalakis2016learning, deng2022nash,  kolumbus2022auctions,kolumbus2024paying}; the impact of learning algorithms on user incentives \cite{aggarwal2024randomized,alimohammadi2023incentive, feng2024strategic, KolumbusN22}; optimal strategies for the auctioneer when the bidders are learning \cite{braverman2018selling, cai2023selling, rubinstein2024strategizing}; and the estimation of user preferences and bid prediction under the assumption that have low regret \cite{gentry2018structural, nekipelov2015econometrics, noti2017empirical, nisan2017quantal, noti2021bid}.
Our work enhances this literature by 
introducing a new aspect not previously studied: the bidder's preference for spacing wins evenly over time.

\paragraph{Budget pacing}
The work most closely related to our setting is the concept of ``budget pacing,'' where bidders' strategies in a budgeted auction are limited to a scalar ``bid-shading factor'' that adjusts their bid in each auction by scaling their value.
Two primary lines of research in this area examine the equilibria in games with such strategy spaces for the bidders \cite{DBLP:conf/sigecom/ChenKK21, DBLP:conf/ec/ConitzerKPSSMW19, DBLP:conf/wine/ConitzerKSM18} and algorithms that learn the shading factor online and their dynamics \cite{aggarwal2024auto, balseiro2023field, DBLP:journals/mansci/BalseiroG19, DBLP:conf/innovations/GaitondeLLLS23, kumar2022optimal, DBLP:conf/colt/LucierPSZ24}.
The main concern in our setting, however, is not only to spend the budget at the right average rate but also to consider the spacing objective under this constraint, and we show that state-independent strategies which ignore the time elapsed since the last win are insufficient for this purpose and lead to high regret.
For example, one could use the budget pacing algorithms of \cite{DBLP:journals/mansci/BalseiroG19} in our problem, using conversion probability as the value of winning the auction.
However, these algorithms will converge to a fixed shading factor, i.e., a state-independent bid distribution, and hence will incur linear regret.

\paragraph{Online learning of Markov Decision Processes (MDPs)}
There is a rich literature on online learning of MDPs (see, e.g., \cite{BookRL}).
Our algorithm does not directly apply existing ones in that literature for several reasons.
First, a na\"ive model of the MDP on our setting would have $T$ states, making the regret bounds linear.
Second, in contrast to the \emph{episodic learning} model that predominates in research on learning of MDPs, we assume the learner participates in a single episode (a \emph{continuing task} in reinforcement learning terminology) and must simulate episodic learning by breaking the timeline into epochs of stochastic length, as described earlier.
Third, while we have a continuous action space (i.e., a bid interval), our regret bounds have no dependence on the size of the action space, in contrast to the MDP learning literature that assumes a discrete action space.
Using discretization would result in a dependence on the number of actions which is either square-root (in the bandit setting) or logarithmic (in the full information setting).
Fourth, due to the bidder's budget constraint, our setting is modeled by a \emph{constrained} MDP.
While there are some recently-discovered online learning algorithms for constrained MDPs in the infinite-horizon~\cite{DBLP:conf/icml/Chen0L22a}%
\footnote{
This work uses one long trajectory of length $T$ of a small MDP, but either (i) assumes that all policies are ergodic, which does not hold in our setting, or (ii) has $\order*{T^{2/3}}$ regret when the MDP is weakly communicating, or (iii) is not efficient.
Both (ii) and (iii) require knowing some parameters of the optimal policy and also require a discrete action set.
Therefore, even after reducing our problem to its small MDP version and discretizing the action space, we would need to know some properties of the optimal policy to use their algorithm in our setting.}
and episodic learning~\cite{DBLP:conf/icml/StradiGGC0024}
settings, those algorithms also share some of the drawbacks discussed above.

\paragraph{Bandits with Knapsacks (BwK)}
Another related thread of work, following \cite{DBLP:conf/focs/BadanidiyuruKS13}, is on the BwK model, which is a generalization of the classic multi-armed bandits model that incorporates resource consumption constraints~\cite{DBLP:conf/icml/CastiglioniCK22, DBLP:conf/colt/FikiorisT23, DBLP:conf/focs/ImmorlicaSSS19, DBLP:conf/nips/KumarK22, DBLP:journals/corr/Slivkins19, DBLP:conf/colt/SlivkinsSF23}.
These constraints require algorithms to consider not only the reward of each arm but also their long-term effect on the resource budgets.
However, an important distinction between the literature on BwK and our work is that the reward of an arm in BwK only depends on the current round.
In particular, as in the case of the budget-pacing algorithm discussed above, in a sequential auction setting, the reward only depends on the current win (or conversion) and not on the time elapsed since the last one.

\section{Model}\label{sec:model}

We consider a bidder facing an unknown price distribution in a sequence of $T$ second-price auctions.  
The bidder has a budget $B$, and we denote the average budget per step by $\rho = B/T$. 
We are interested in the long-term behavior and outcomes in the series of auctions and think about $\rho$ as constant as $T \rightarrow \infty$.
In each round $t \in [T]$, a price $p_t \in [0, 1]$ is sampled from an unknown distribution $\calP$.
The price $p_t$ can be interpreted as the highest bid among the pool of other competitors in the auction. 
The bidder chooses a bid $b_t \in [0, 1]$; if $b_t \ge p_t$, the bidder wins the auction and pays a price $p_t$ from her budget.
If $b_t < p_t$, the bidder loses the auction, pays nothing, and observes $p_t$.
We discuss the utility model for the bidder next.

When using bid $b$, we denote by $W(b)$ and $P(b)$ the probability of winning the auction and the expected payment, respectively.
Formally, $W(b) = \Pr[p\sim \calP]{b \ge p}$ and $P(b) = \Ex[p\sim\calP]{p \One{b \ge p}}$.
For a distribution of bids $\bb \in \Delta([0, 1])$, we overload our notation and define $W(\bb) = \Ex[b\sim\bb]{W(b)}$ and $P(\bb) = \Ex[b\sim\bb]{P(b)}$. 

\paragraph{Reward model}
The bidder is interested both in winning many auctions and spacing the winning events over time.
Two basic properties we require are that (i) adding any set of winning events to any sequence $t_1,\cdots,t_k$ of winning times can only improve utility, and (ii) for any fixed number $k$ of winning events, the sequences in which those events are most evenly spaced lead to highest utility.
While there is no single way to represent such preferences, a natural and tractable model that achieves these properties is maximizing the sum of some concave function of the lengths of intervals between winning times. This leads to the following reward model.

The bidder has a reward per winning event, which depends on the number of rounds since the last winning event.
Formally, let $t' < t$ be the last round before $t$ in which a winning event occurred and $\l_t = t - t'$ be the number of elapsed rounds since this last winning event. (If there was no such event before $t$, define $t' = 0$ and $\l_t = t$.)
Then, in round $t$, if the bidder wins the auction (i.e., if $b_t \ge p_t$), she receives a reward $r(\l_t)$, where $r: \N \cup \{0\} \to \mathbb{R}_+$ is an increasing and concave sequence with bounded differences.
That is, for every $\l \in \N \cup \{0\}$, it holds that 
\begin{equation*}
    0 \ 
    \le \ 
    r(\l + 2) - r(\l + 1) \ 
    \le \ 
    r(\l + 1) - r(\l)
    \le
     \ 1,
\end{equation*}
where we assume $r(0) = 0$.
An example of such a sequence is $r(\l) = \sqrt \l$ (as in the warm-up example in \cref{sec:intro}).
For any $m\in\N$, it will be useful to define $r_m(\l) = \min\{r(\l), r(m)\} = r(\min\{\l, m\})$.

\paragraph{Extension to contextual settings}
In the first part of the paper, we focus on the model defined above for simplicity of presentation. 
However, all our results extend to an important contextual generalization of our model.
In the contextual auction setting, in every round $t$, a context $x_t$ is sampled from some distribution over an arbitrary space $\calX$.
The bidder observes $x_t$ before bidding; naturally, the context can also affect the price.
While we continue to assume that prices in each iteration are independent, we allow the price and the context to be arbitrarily correlated.

Each context $x_t$ includes a conversion probability $c_t \in [0, 1]$, which is the probability that winning an auction in round $t$ (i.e., $b_t \ge p_t$) results in a conversion event, i.e., a reward.
Specifically, in round $t$ where the last conversion was $\l_t$ rounds ago, the bidder will collect a reward $\rfunc(\l_t)$ if $b_t \ge p_t$ and a conversion occurs, which, conditional on $b_t \ge p_t$, happens with probability $c_t$; in this case $\l_{t+1} = 1$. 
Only if these two conditions are satisfied does the bidder get any reward; otherwise, she gets zero reward in round $t$, and in the next round, the state changes to $\l_{t+1} = \l_t + 1$.

The context affects both the probability of winning a conversion and the expected payment of a bid.
For this case, we use a bidding function $\bb : \calX \to \Delta([0, 1])$ that maps a context $x$ to a randomized bid and define $W(\bb)$ and $P(\bb)$ analogously to denote the expected probability of winning a conversion and the expected price for a bidding function $\bb$.
For example, $W(\bb) = \Pr[x,p,b\sim\bb(x)]{b \ge p}$.
We denote $\bar c = W(1) = \Ex{c}$ the expected conversion rate, i.e., the probability that the bidder gets positive reward if she bids $1$, and assume that $\bar c$ is a constant with respect to $T$.

\section{State-Space Reduction for Near-Optimal Planning} \label{sec:bench}

Before we address the problem of learning to bid online in the next sections, in this section, we ask the following question.
{\em Suppose that the bidder knew the distribution of prices and contexts exactly. 
What is the maximum reward she could aim for?
And what policy would achieve this?}

The main result of this section is introducing an infinite-horizon optimization problem with only $\order*{\log T}$ parameters, whose solution can be used as a nearly optimal approximation for the budgeted auction problem with the spacing objective.

First, in \cref{ssec:bench:optimal_algo}, we define the optimal solution one would use knowing the distribution of prices: the optimal algorithm that respects the budget constraint and takes into account \textit{all} the information of past rounds.
Noticing that this solution is hard to compute online, we focus on a different solution.
In \cref{ssec:bench:infinite}, we introduce an infinite-horizon optimization problem over $m$ states, each state representing the time since the last win (such that $m, m+1, \ldots$ rounds since the last win all yield the same reward).
\cref{thm:bench:comparison} shows that when $m = T$ (the number of states equals the number of rounds in the original problem), the reward in the infinite-horizon setting is at least that of the best algorithm in the original problem.
While this may initially appear more complex, learning the optimal solution is simplified due to its stationary nature: with no time horizon, the only relevant variable is the bidder's current state among the $m$ possible states.
In \cref{ssec:bench:small_m}, we focus on simplifying this solution even further.
As mentioned above, with $m=T$ different states, we would have to learn a different bidding strategy for each state.
\cref{thm:bench:small_m} shows that we can, instead, consider an exponentially smaller MDP---with only $m = \order*{\log T}$ states---with only minimal error.
This new and nearly optimal solution is thus much easier to learn online, which will be our goal in \cref{sec:online}.
We start the presentation of these results in the simpler setting without contexts and conversion rates.
In \cref{ssec:bench:conversion}, after establishing the basics, we switch back to the more general contextual setting.

\subsection{Optimal Offline Algorithm} \label{ssec:bench:optimal_algo}

For the first two subsections, assume there are no contexts, and the conversion rate is always $1$ for simplicity of presentation.
When the bidder knows the distribution of prices, she can maximize her expected reward while satisfying the budget constraint without needing to learn this distribution.
In particular, in every round, she could calculate her optimal bid as a function of her current ``state'' $(t, B_t, \l_t)$, where $t$ is the current round, $B_t$ is her remaining budget, and $\l_t$ is the number of rounds since her last winning event.
We note that in every state where $B_t \ge 0$, the bidder can bid zero so that her budget is not depleted.
While the optimal strategy in this state space is complicated, it is possible to compute the optimal policy using a very large dynamic program under mild assumptions.
As long as the number of prices is finite, the possible states are also finite.
For example, if for some $K \in \N$ every price was of the form $p_t = \frac{i}{K}$ for some $i = 0, 1, \ldots, K$, then the number of possible states would be $\Theta(K T^3)$.
Therefore, we can calculate the optimal solution using dynamic programming.
We denote the value of this optimal solution by $T \cdot \optAlg$, where $\optAlg$ is the average-per-round value of the optimal solution.
This benchmark is referred to in the literature as the {\em best algorithm} (e.g., in Bandits with Knapsacks; see \cite[Chapter 10]{DBLP:journals/corr/Slivkins19} or \cite{DBLP:conf/colt/SlivkinsSF23}).

While the above solution is ``tractable,'' its state space is very big, and, most importantly, it is unclear how to design an online learning algorithm that approximates it since the same state is never seen twice. 
For this reason, we will focus on a different benchmark that is much simpler to define and, as we will see below, is also a stronger comparator.

\subsection{Infinite-Horizon Benchmark} \label{ssec:bench:infinite}

To develop our benchmark, we consider a setting with infinite rounds. 
To differentiate the notation from our finite-horizon setting, we use here $H$ to denote the time horizon and consider the limit $H \to \infty$.
Note that in the following analysis, both notations will be used when comparing the two settings: $T$ will refer to the time horizon of our original setting (as described above), and $H$ will be used for the infinite-horizon setting.
Like the finite-horizon setting, the bidder tries to maximize her reward while adhering to a budget constraint.
Since the time horizon is infinite, we need to define the notions of reward and budget for the bidder, which we explain next.

\paragraph{Budget}
In the infinite-horizon setting, we require the budget condition to hold only in expectation and only in the limit.
Specifically, the infinite-horizon time average of the expected spending (with expectation over the prices) must be less than $\rho$.
For example, if $b_h$ is the bid in round $h \in [H]$, then we want the bidder's expected average-per-round payment to be less than $\rho$ as $H \to \infty$, i.e.,
\begin{equation} \label{eq:bench:avg_pay}
    \lim_{H \to \infty} \frac{1}{H} \sum_{h = 1}^H \Ex{ P(b_h) }
    \le
    \rho,
\end{equation}
where, recall that $P(b)$ is the expected payment of bid $b$ when the price is sampled from distribution $\calP$.
We note that the above limit might not exist for some bid sequences; we will deal with this later (see the remark after Equation \eqref{eq:opt_inf}).

\paragraph{Reward}
Since in the original finite-horizon setting, we had that $r(\l_t) \le r(T)$, it is convenient to artificially make a similar restriction here with a parameter $m$.
Specifically, when the last winning event before round $h$ was $\l_h$ rounds ago, we limit the bidder's reward to be $r_m(\l_h)$ (recall $r_m(\l) = r(\min\{m, \l\})$).
This reward structure matches the one in \cref{ssec:bench:optimal_algo} when $m = T$.
Hence, for a sequence of bids $b_1, b_2, \ldots$ the expected average-per-round reward is
\begin{equation} \label{eq:bench:avg_rew}
    \lim_{H \to \infty} \frac{1}{H} \sum_{h = 1}^H \Ex{r_m(\l_h) W(b_h)}.
\end{equation}

Therefore, in this section, we aim to maximize the quantity in \cref{eq:bench:avg_rew} while satisfying the constraint in \cref{eq:bench:avg_pay}.
In addition, because for $\l_h \ge m$, the reward is the same as in the case $\l_h = m$, the optimal bidding can be the same for all $\l_h \ge m$.
In other words, we consider an infinite-horizon average-reward MDP with a budget constraint over $m$ states and the bid range $[0, 1]$ as the action set.
Since there are only $m$ states and no time horizon, we consider stationary policies that are distributions of bids for each of the $m$ states. 
The following optimization problem describes the optimal value for any $m \in \N$:
\begin{equation} \label{eq:opt_inf}
\begin{split}
    \optInf_m
    =
    \sup_{\bb_1,\ldots,\bb_m\in \Delta([0,1])} \;\;
    \lim_{H\to\infty}\frac{1}{H} \sum_{h=1}^H \Ex[\l_h]{ r_m(\l_h) W(\bb_{\l_h}) }
    \quad \text{s.t.} \quad
    \lim_{H\to\infty}\frac{1}{H} \sum_{h=1}^H \Ex[\l_h]{ P(\bb_{\l_h}) } \le \rho.
\end{split}
\end{equation}
A visual representation of the transitions of the MDP in Equation \eqref{eq:opt_inf} is shown in Figure~\ref{fig:mc}.
We remark that the limits now exist, according to \cite[Theorem 8.1.1]{DBLP:books/wi/Puterman94}, since the state space is finite and the policy is stationary.


Next, we prove that the optimal value of the infinite-horizon problem is at least the value of the best algorithm in the finite-horizon problem, when $m = T$.
This will allow us to limit our attention to learning the optimal solution to this problem instead of the optimal algorithm of \cref{ssec:bench:optimal_algo}.

\begin{restatable}{lemma}{benchcomparison}
\label{lem:bench:comparison}
    For every budget per round $\rho$, reward function $\rfunc(\cdot)$, distribution of prices $\calP$, and finite horizon $T$, it holds that $\optInf_T \ge \optAlg$.
\end{restatable}

The intuition behind the proof is a thought experiment where we simulate runs of the optimal algorithm corresponding to $\optAlg$ in the infinite-horizon setting.
Specifically, we partition the rounds into blocks of length $T$ and run a new instance of this algorithm in each block.
The expected reward in each block is at least $T \cdot \optAlg$, and the payment is at most $T \rho$.
This makes the time-average expected reward at least $\optAlg$ and the time-average payment at most $\rho$. See \cref{ssec:app:bench:comparison} for the full proof.

It will be simpler to work with a change of variables from bids to winning probabilities.
We used $W(b)$ as the probability of winning the auction with a bid $b$.
Instead, we will consider a winning probability $W \in [0, 1]$, and aim to bid so as to make the winning probability $W$.
For example, bidding $1$ with probability $W$ and $0$ with probability $1 - W$ wins with probability\footnote{
There is some subtlety here: if the price is zero with positive probability, then bidding zero does not result in winning with zero probability. We can solve this issue by assuming that the bidder also has the option to ``skip'' an auction, even if this would be avoided in any optimal solution.}
$W$. 
On the other hand, multiple bid distributions might win with probability $W \in [0, 1]$; in that case, we consider that the bidder uses the one with the lowest expected spending, which we define as $P(W)$.
We note a slight subtlety here: for some price distributions, there might not exist a bid distribution that achieves the minimum, but there exists one that gets arbitrarily close.
We discuss this in \cref{ssec:calc:single}.


\subsection{Adding Contexts}
\label{ssec:bench:conversion}

In this section, we extend the infinite-horizon constrained MDP defined in the previous section to the problem with contexts and conversion rates.
We focus on the final formulation, where the winning probabilities are the variables.
The important variables now are the probabilities of winning a {\em conversion}.
So, we redefine $W$ as the probability of winning a conversion (rather than just the probability of winning the auction).
Furthermore, the actual probability of winning a conversion also depends on the context and conversion rate.
Bidding $1$ would now result in a conversion probability of only $\bar c$ (recall $\bar c = \Ex{c}$ is the expected context per round).
Using this, we need $W \in [0, \bar c]$ (instead of the full range $[0,1]$).
Since the optimal bid that achieves probability $W$ of winning a conversion will depend on the context, we need a bidding function $\bb$ mapping contexts $\calX$ into bids such that $W(\bb) = W$.
As before, we define the payment for winning with probability $W \in [0, \bar c]$ as
$$
    P(W) = \inf_{\substack{\bb:\calX\to\Delta([0, 1])\\\text{s.t. }W(\bb)=W}} P(\bb).
$$
With this notation, the infinite-horizon constrained MDP we want to work with becomes
\begin{equation} \label{eq:opt_inf_W}
\begin{split}
    \optInf_m
    =
    \sup_{ \vec W \in [0, \bar c]^m } \;\;
    \lim_{H\to\infty}\frac{1}{H} \sum_{h=1}^H \Ex[\l_h]{ \rfunc_m(\l_h) W_{\l_h} }
    \quad \text{s.t.} \quad
    \lim_{H\to\infty}\frac{1}{H} \sum_{h=1}^H \Ex[\l_h]{ P(W_{\l_h}) } \le \rho.
\end{split}
\end{equation}

Analogous to \cref{lem:bench:comparison} we have the following theorem, whose proof is a direct extension of \cref{lem:bench:comparison}  by considering the above modified definitions of $W$, $P(W)$, and the bidding functions $\bb$.
\begin{theorem}
\label{thm:bench:comparison}
    For every finite-horizon setting (any context/price distribution, budget per round $\rho$, reward function $\rfunc(\cdot)$, and time horizon $T$) it holds that $\optInf_T \ge \optAlg$.
\end{theorem}

In \cref{ssec:app:bench:stationary} we introduce an alternative optimization problem for \eqref{eq:opt_inf_W}.
Specifically, we show that \eqref{eq:opt_inf_W} can be simplified using the stationary distribution on the states $[m]$ implied by winning probabilities $W_1, \ldots, W_m$.

\subsection{Reducing the Number of States in the Infinite-Horizon Problem to \texorpdfstring{$\order*{\log{T}}$}{logarithmic in T}} \label{ssec:bench:small_m}

As shown in \cref{thm:bench:comparison} we can use the infinite-horizon benchmark, $\optInf_T$, to approximate the expected reward of the best algorithm $\optAlg$.
While $\optInf_T$ is easier to describe than $\optAlg$, the number of states in the MDP is $T$.
This means that, during the $T$ steps of running our learning algorithm, some states will be visited at most once, which is insufficient for learning.
Furthermore, because regret bounds for learning MDPs usually depend polynomially on the number of states, the $T$-state MDP will not allow us to get our promised $\orderT*{\sqrt T}$ regret bound.
For example, standard regret bounds in this setting are $\order*{\sqrt{S T}}$ \cite{DBLP:conf/icml/AzarOM17} where $S$ is the number of states.

We solve this issue by showing that using $\optInf_m$ instead of $\optInf_T$ in fact, leads to minimal error even when $m \approx \log T$.
Specifically, the final result for this section is the following.

\begin{restatable}{theorem}{benchsmallm}
\label{thm:bench:small_m}
    Fix a constant $C > 0$ and an integer $m \ge \frac{2C}{\bar c \rho} \log T$.
    Then, for any integer $M \ge m$, it holds that $\optInf_m \ge \optInf_M - \order*{\frac{1}{\bar c \rho} T^{-C}}$.
    In addition, this can be achieved even if we constrain $W_m = \bar c$ (i.e., bid $1$ in state $m$).
\end{restatable}

The high-level idea of the proof is that with an average budget $\rho$, the bidder can win an auction on average every $\rho^{-1}$ steps and, with expected conversion probability $\bar c$, get a conversion every $(\bar c\rho)^{-1}$ rounds.
We will prove that in the optimal policy, the probability that no conversion occurred for a sufficiently long time is so low that dropping this part of the MDP leads to minimal error. 
The hard part of turning this into a proof is to argue that in an optimal strategy, the winning probabilities $W_\l$ are monotone increasing in $\l$ (the longer the time since the last conversion, the more eager the bidder is to win now).
Without this property, the following solution could win on average every $\rho^{-1}$ steps (consider $\bar c = 1$ for simplicity): $W_1 = 1 - \frac{1 - \rho}{\rho}\order*{T^{-1/2}}$, $W_2 = \ldots = W_{\sqrt T - 1} = 0$, and $W_{\sqrt T} = 1$.
Dropping the latter part of such a solution would result in a significant loss in value.
While not surprising, proving the monotonicity of the optimal $W_\l$ requires substantial effort.
We provide an extensive outline of this result; the complete proof can be found in \cref{ssec:app:increasing,ssec:app:bench:bias}.

\paragraph{Lagrangian problem}
The first step is to consider the Lagrangian of the problem.
In the Lagrangian version of the problem, we add the spending into the objective function with a Lagrange multiplier $\lambda \ge 0$.
The value of winning a conversion in state $\l$ with probability $W_\l$ becomes:
\begin{equation*}
    \calL_\l(W_\l, \lambda) = r_m(\l) W_\l - \lambda P(W_\l) .
\end{equation*}
Traditionally, one would write the Lagrangian objective as
$r_m(\l) W_\l + \lambda( \rho - P(W_\l))$, as (the infinite-horizon average of) this expression for any $\lambda \geq 0$ 
constitutes an upper bound on the value of the optimal solution to the constrained problem~\eqref{eq:opt_inf_W}.
Our formulation of the Lagrangian objective above omits the $\lambda \rho$ term as we focus on maximizing $\calL$ by picking $W$, and $\lambda \rho$ is a constant in this respect.

Fix a $\lambda \ge 0$.
In the Lagrangian MDP the reward in state $\l \in [m]$ is $\calL_\l(W_\l, \lambda)$ and the action space is $W_\l \in [0, \bar c]$.
Without constraints, we make use of standard machinery from unconstrained MDP optimization.
Following the notation of \cite[Section 5]{DBLP:books/wi/Puterman94}, let $g_\l^*(\lambda)$ and $h_\l^*(\lambda)$ be the gain and bias of the optimal policy $W_1^*(\lambda), \ldots, W_m^*(\lambda)$, which are defined as follows.
First, the gain $g_\l^*(\lambda)$ is the time-average expected reward of the optimal policy when starting at state $\l$.
Formally,
\begin{equation*}
    g_\l^*(\lambda)
    =
    \lim_{H\to\infty}\frac{1}{H} \sum_{h=1}^H
    \ExC[\l_h]{\calL_{\l_h}\qty\big( W_{\l_h}^*(\lambda), \lambda) }{\l_1 = \l}.
\end{equation*}
\cite[Theorem 8.3.2]{DBLP:books/wi/Puterman94} shows that for weakly communicating MDPs (MDPs where every state is reachable from any other state) there is always an optimal policy with constant gain, i.e., $g_\l^*(\lambda) = g_{\l'}^*(\lambda)$ for $\l \ne \l'$.
Therefore, $g_\l^*(\lambda) = g^*(\lambda)$.
The bias $h_\l^*(\lambda)$ at state $\l$ is defined as the difference in total expected reward if starting at state $\l$ instead of getting $g^*(\lambda)$ every round.
Formally,
\begin{equation*}
    h_\l^*(\lambda)
    =
    \lim_{H\to\infty} \sum_{h=1}^H \qty(
        \ExC[\l_h]{\calL_{\l_h}\qty\big(W_{\l_h}^*(\lambda^*), \lambda^*) }{\l_1 = \l}
        -
        g^*(\lambda)
    ).
\end{equation*}

We note that the above limit does not necessarily exist.
In that case, we can substitute the limit above with the Cesaro limit\footnote{\url{https://en.wikipedia.org/wiki/Cesaro_summation}} to solve this issue and make the bias function well-defined.

Using the gain and the bias functions, we use the Bellman optimality condition for infinite-horizon average reward MDPs (see \cite[Section 8.4.1]{DBLP:books/wi/Puterman94}).
Specifically, in our setting, this condition implies that for all $\l \in [m]$
\begin{equation*}
\begin{split}
    h_\l^*(\lambda) + g^*(\lambda)
    = &
    \max_{W}\qty\Big[
        \calL_\l(W, \lambda) + W h_1(\lambda) + (1 - W) h_{ \min\{m, \l+1\} }(\lambda)
    ]
\end{split}
\end{equation*}

We use two simplifying definitions.
First, we define $h_{m+1}(\lambda) = h_m(\lambda)$; this simplifies the minimum in the above notation.
Second, because the above equation is invariant to adding a constant to $h_\cdot^*(\lambda)$, we can assume w.l.o.g. that $h_1(\lambda) = 0$.
Using this notation, we get that for any optimal solution $W_1^*(\lambda),\ldots,W_m^*(\lambda)$ it holds that
\begin{equation} \label{eq:bell}
\begin{split}
    h_\l^*(\lambda) + g^*(\lambda)
    \;=\; &
    \max_{W}\qty\Big[
        W r(\l) - \lambda P(W) +  \qty\big(1 - W) h_{ \l+1 }^*(\lambda)
    ]
    \\
    \;=\; &
    W_\l^*(\lambda) r(\l) - \lambda P\qty\big(W_\l^*(\lambda)) +  \qty\big(1 - W_\l^*(\lambda)) h_{ \l+1 }^*(\lambda).
\end{split}
\end{equation}

\paragraph{Properties of the bias function}
We now examine the structure of the bias function since $W_\l(\lambda)$ depends on $h_{\l+1}^*(\lambda)$ (see the $\max$ term of \cref{eq:bell}).
The property that interests us is a bound on the increments of the bias function $h_{\l+1}^*(\lambda) - h_{\l+1}^*(\lambda)$.
We show that these increments cannot be more than the increments of the reward function $\rfunc(\cdot)$.
This will let us compare the two maximization problems that define $W_\l^*(\lambda)$ and $W_{\l+1}^*(\lambda)$.

\begin{restatable}{lemma}{benchbias}\label{lem:bench:bias}
    For every $\l \in [m]$ it holds that
    \begin{equation}\label{eq:42}
        h_{\l+1}^*( \lambda )
        \le
        h_\l^*(\lambda) + r_m(\l+1) - r_m(\l).
    \end{equation}
\end{restatable}

One idea to prove the above lemma is the following thought experiment.
Consider starting at state $\l$, but instead of following the optimal strategy, pretend to start at state $\l+1$ and follow that strategy. 
That is, pretend that we started one state ahead.
The discrepancy in reward between this and the optimal strategy is $r(\l'+1) - r(\l')$, where $\l'$ is the state where we won for the first time (after which the two trajectories coincide).
Because of concavity, this is at most $r(\l+1) - r(\l)$, proving the same bound for $h_{\l+1}^*(\l) - h_\l^*(\l)$, as needed.
The formal proof we present in \cref{ssec:app:bench:bias} is more technical and proves the lemma by induction using the optimality equation \eqref{eq:bell}.

\paragraph{Monotonicity of the optimal Lagrangian solution}
We now prove that the optimal solution for any Lagrange multiplier $\lambda$ must be non-decreasing.
This will imply the same property for the constrained problem by using the fact that the constrained problem \eqref{eq:opt_inf_W} can be rewritten as a Linear Program.
However, the last fact only proves that an optimal solution of the constrained problem is optimal for the Lagrangian problem of some $\lambda$.
Therefore, knowing that a monotone optimal solution exists for the Lagrange problem is insufficient.
We have to show that \textit{every} optimal solution is monotone.
The additional assumption we need for this result is that the reward function $r(\cdot)$ must satisfy strict concavity.
This can be artificially enforced in any reward function with negligible error by perturbing it by $\order*{1 / \text{poly}(T)}$.
The proof uses \cref{lem:bench:bias}, the Bellman optimality conditions, and the strict concavity of the reward function.
See the details in \cref{sec:app:bench}.

\begin{restatable}{lemma}{benchincr} \label{lem:bench:incr}
    Fix any optimal solution $\vec W^*(\lambda)$ and $\l \in [m-1]$.
    If $r(\l+1) - r(\l) > r(\l+2) - r(\l+1)$, then $W_\l^*(\lambda) \le W_{\l+1}^*(\lambda)$.
\end{restatable}

The lemma is proven by considering the use of $W_\l^*(\lambda)$ in state $\l+1$ and vice versa.
Due to optimality, this leads to (weakly) suboptimal behavior.
Comparing these two scenarios would lead to improvement in reward if $W_\l^*(\lambda) > W_{\l+1}^*(\lambda)$, proving the lemma.

\paragraph{Monotonicity of the optimal constrained solution}
Using the above lemma, we achieve the monotonicity of our original problem.
Because optimization problem (\ref{eq:opt_inf_W}) can be written as a Linear Program using the occupancy measure (see \cref{ssec:app:calc:lp} for a detailed presentation), its optimal solutions are also optimal for the Lagrangian problem, for some $\lambda$.
Therefore, \cref{lem:bench:incr} proves the monotonicity of those solutions: $W_1^* \le W_2^* \le \ldots \le W_m^*$.
However, when the context/price distribution has infinite support, the variables of this linear program are real-valued functions (specifically, probability-like measures over the space of actions).
This might make the optimality of $W_1^*, \ldots, W_m^*$ for the Lagrangian problem not trivial; we include a proof of this in \cref{ssec:app:increasing} for completeness.

\begin{restatable}{corollary}{benchincreasing} \label{cor:bench:increasing}
    Assume $r(\cdot)$ is strictly concave.
    Then there exists optimal solution for the optimization problem \eqref{eq:opt_inf_W} that is weakly increasing, i.e., $W_\l^* \le W_{\l+1}^*$ for all $\l\in[m-1]$.
\end{restatable}

\paragraph{Optimality of the small state MDP}
The key result to prove \cref{thm:bench:small_m} is that the optimal solution visits state $\l$ very infrequently for large $\l$.
Specifically, we show that for large enough $\l$, the probability of winning a conversion in state $\l$ is at least $\frac{\bar c \rho}{2}$; this implies that state $\l+1$ is visited $1 - \frac{\bar c \rho}{2}$ times less than state $\l$.
This is captured in the following lemma.

\begin{restatable}{lemma}{benchwinLB} \label{lem:bench:win_LB}
    There exists optimal solution $\vec W^*$ for the constrained problem where, 
    if $\l \ge \frac{2}{\bar c \rho}$ then $W_\l^* \ge \frac{\bar c \rho}{2}$.
\end{restatable}

We prove the lemma by first noticing that the average payment of the optimal solution must be exactly $\rho$ (unless the bidder can always afford to pay the price $p_t$ in expectation).
We then prove that unless the inequality in \cref{lem:bench:win_LB} is true, $W_\l^*$ must be small for all $\l \le \frac{\bar c \rho}{2}$ (by monotonicity in \cref{cor:bench:increasing}).
We end up proving that this implies that the bidder is paying less than $\rho$ per average, contradicting our earlier claim.
We defer the details of the proof to \cref{ssec:app:increasing}.

We proceed to sketch the proof of \cref{thm:bench:small_m} with the detailed proof in \cref{ssec:app:bench:small_m}.
First, let $R_m(\vec W)$ and $C(\vec W)$ be the expected average reward (when there are $m$ states) and payment of a vector $\vec W$.
The proof consists of two steps.
First, we define vector $\vec W'$ identical to an optimal vector $\vec W^*$ for states $\l < m$.
For state $\l \ge m$ we set $\vec W' = \bar c$.
We then examine the average-time reward of $\vec W'$: in the original MDP with $T$ states, it holds that $R_T(\vec W') \ge R_T(\vec W^*) = \optInf_T$, but this is not necessarily the case when there are $m$ states; \cref{cl:bench:reward_Wprime} proves that $R_T(\vec W^*) - R_T(\vec W')$ is very small.
Afterward, we examine the spending of $\vec W'$, which can be more than the spending of $\vec W^*$; \cref{cl:bench:cost_Wprime} proves that $C(\vec W^*) - \rho$ is very small.
The final step of the proof involves finding a feasible solution with reward close to $\vec W'$.
By defining a solution $\vec W''$ with very small spending, a combination of $\vec W'$ and $\vec W''$ yields the theorem.

\section{Online Learning Algorithm}
\label{sec:online}

This section presents our algorithm that achieves low regret against $\optInf_m$.
We assume the algorithm receives full-information feedback, i.e., after bidding in round $t$ it observes the price $p_t$.
We prove regret bounds for $\optInf_m$ for any $m \le T$.
Specifically, \cref{thm:online:regret} proves a $\orderT*{\textrm{poly}(m)\sqrt T}$ regret bound.
Then, setting $m = \Theta(\log T)$ as in \cref{thm:bench:small_m} will provide $\orderT*{\sqrt T}$ regret against our original benchmark, the best algorithm that knows the distributions.
The polynomial dependence on $m$ in the above regret bound shows why we proved in \cref{ssec:bench:small_m} that $\optInf_m \approx \optInf_T$ for $m = \Theta(\log T)$.

The algorithm works by partitioning the $T$ rounds into epochs of stochastic length.
Based on past data, the bidder calculates the optimal bidding policy of past data at the beginning of each epoch (in \cref{sec:calc} we show that $t$ samples result in $\order*{\frac{m^3}{\rho \sqrt t}}$ accuracy by solving a linear program).
During each epoch, we bid using this empirically optimal bidding without updating it.
Each epoch can end in one of two ways: when the bidder wins a conversion or if $k$ rounds have passed without a conversion.
$k$ is a parameter that ensures that the length of each epoch remains bounded with probability $1$.
To ensure that an epoch will end with winning a conversion with a high probability, we choose a high enough $k$ and also constrain the empirical bidding to bid $1$ in state $\l = m$ (\cref{thm:bench:small_m} shows that this results in negligible error if $m$ is large enough).

To keep epochs similar, we consider one more modification.
Specifically, at the start of each epoch, we bid as in state $\l = 1$ of the MDP.
This ensures that each epoch starts `anew,' even if there was no win with conversion in the previous one.
This leads to a mismatch between the time since the last conversion.
However, the real allocated reward can only be bigger than the reward of our `fake' start.
This is because if we win when the fake state is $\l$, then the bidder receives at least $\rfunc_m(\l + k)$, which cannot be less than $\rfunc_m(\l)$ by monotonicity.
The full algorithm is in \cref{algo:online}.

\begin{algorithm}[t]
\DontPrintSemicolon
\caption{Follow the $k$-delayed Optimal Response Strategy (FKORS)}
\label{algo:online}
\KwIn{Average budget $\rho$, reward function $\rfunc:\N \to \R_{> 0}$, parameters $m, k \in \N$}
Bid $0$ for the first $k$ rounds and observe price and conversion rates $\{ (p_t, x_t) \}_{t \in [K]}$\;
Set the length and end of epoch $0$: $L_0 = k$ and $T_0 = k$\;

\For{epoch $i = 1, 2, \ldots$}
{
    Calculate the optimal vector of bidding policies from contexts to bids $\vecestbb{i}$ on the empirical dataset $\{ (p_t, x_t) \}_{t \in [T_{i-1}]}$ using linear program \eqref{eq:calc:LP}, constrained that $\estbb{i}_m(\cdot) = 1$\;

    Restart the MDP: set $\tilde \l_{T_{i-1} + 1} = 1$\;

    \For{rounds $t = T_{i-1} + 1, T_{i-1} + 2, \ldots, T_{i-1} + k$}{
        Observe context $x_t$\;
    
        If the remaining budget is at least $1$, bid using $\bb^i_{\tilde\l_t}(x_t)$, otherwise bid $0$\;
        
        Observe price $p_t$\;
        
        \uIf{conversion happens}{
            Receive reward $\rfunc(\tilde \l_t)$\;
            Set the end of this epoch $T_i = t$ and its duration $L_i = T_i - T_{i-1}$\;
            Go to the next epoch\;
        }\uElse{
            Set $\tilde \l_{t+1} = \min\{m , \tilde\l_t + 1 \}$\;
        }
    }

    \uIf{no conversion happened in epoch $i$}{
        Set the end of this epoch $T_i = T_{i-1} + k$ and its duration $L_i = k$\;
    }
}
\end{algorithm}

The rest of the section focuses on proving the regret bound of \cref{algo:online}.
We make the simplifying assumption that we know $\bar c$.
This is needed to set $m$ and $k$ at least a function of $\frac{1}{\bar c}$.
This can be replaced by any lower bound for $\bar c$; as long as this lower bound is within a constant of $\bar c$, our regret bound becomes multiplicative bigger only by that constant.
Such a loose lower bound for $\bar c$ can be calculated by sampling the conversion rate a few times before running our algorithm.
We do not include this for simplicity.

In the following theorem, we prove a parametric high probability regret bound for \cref{algo:online}, assuming $k \ge m + \frac{1}{\bar c} \log T$.
One part of the regret bound is a $\orderT*{ k \sqrt{T} }$ term.
The other part depends on the error due to sampling, i.e., calculating the bidding of epoch $i$ with an empirical distribution.
To present this error, we first overload our previous notation and define $R(\vec\bb)$ and $C(\vec\bb)$ as the expected average reward and payment of using a vector of bidding policies $\vec\bb$ in the infinite horizon setting.
One part of this error is the sub-optimality of this bidding: $\e_R(t) = \qty\big(\optInf_m - R(\vec\bb^t))^+$ where $\vec\bb^t$ is the optimal bidding using the empirical distribution of the first $t$ rounds.
The other part of this error is due to the potential over-payment of this error: $\e_C(t) = (C(\vec\bb^t) - \rho)^+$.

\begin{restatable}{theorem}{ThmOnlineRegret} \label{thm:online:regret}
    Fix any $m \in \N$ and let $k \ge m + \frac{1}{\bar c} \log T$.
    Let $N$ be the number of epochs.
    Then for all $\d > 0$ with probability at least $1 - \d$ \cref{algo:online} achieves regret against $\optInf_m$ that is at most
    \begin{equation*}
        \sum_{j=1}^N L_j \qty\big( \e_R(T_{j-1}) + \e_C(T_{j-1}) )
        +
        \order{ k \sqrt{T \log\frac{T}{\d}} }
    \end{equation*}
    where $\e_R(T_{j-1})$ and $\e_C(T_{j-1})$ are error terms of the bidding of epoch $j$: $\e_R(T_{j-1}) = \qty\big( \optInf_m - R(\vecestbb{j}) )^+$ is the reward sub-optimality gap and $\e_C(T_{j-1}) = \qty\big( C(\vecestbb{j}) - \rho )^+$ is the expected average payment above $\rho$. 
\end{restatable}

Using \cref{thm:calc} that bounds $\e_R(t)$, $\e_C(t)$ and assuming $m$ is large enough we get the unconditional regret bound against $\optInf_T$.

\begin{restatable}{corollary}{CorOnlineFinal} \label{cor:online:final_res}
    Let $m = \ceil*{\frac{2}{\bar c \rho} \log T}$ and $k = \ceil*{ m + \frac{1}{\bar c} \log T }$.
    Then for all $\d > 0$ with probability at least $1 - \d$, \cref{algo:online} achieves total reward at least
    \begin{equation*}
        T \cdot \optInf_T
        -
        \order{
            \frac{1}{\bar c^3 \rho^4} \log^3 T
            \sqrt{T \log\frac{T}{\d}}
        }
    \end{equation*}
\end{restatable}

To prove \cref{thm:online:regret}, we face many challenges.
Many standard tools that analyze sequential decisions do not work in our setting since, depending on past rounds, the bidder is in a different state of the MDP.
This problem becomes even greater when the bidding policy is updated, and the underlying Markov Chain changes.
This is why we do not change the bidding policy during an epoch, fixing the Markov Chain over its duration.
While this is not entirely new in the Online Learning MDP literature, using epochs of unknown stochastic length is.

\begin{remark}
    We briefly go over why we need epochs of stochastic length, even under very simplistic error analysis.
    Assume we had epochs of fixed length $k$, i.e., commit to a bidding strategy for $k$ consecutive rounds.
    In that case, in every epoch, even if the bidding is close to optimal, we get $k \cdot \optInf_m - \Omega(1)$ value in expectation; the ``$- \Omega(1)$'' error is because of the change of policies between epochs.
    This makes the total reward over all epochs $\qty\big(k \cdot \optInf_m - \Omega(1)) \frac{T}{k} = T \cdot \optInf - \Omega(\frac{T}{k})$.
    However, this creates the following issue: since we need high probability bounds to deal with the event of running out of budget early, we have to use concentration inequalities.
    These concentration inequalities cannot be used on individual rounds as they are dependent, so we use them across epochs.
    However, this will lead to $\Omega(k\sqrt{T/k})$ error across all of the $\frac{T}{k}$ epochs, making the total error $\Omega(\frac{T}{k} + \sqrt{T k})$.
    Picking $k = \Theta(T^{1/3})$ yields the optimal error of $\Omega(T^{2/3})$, which is much higher than what we offer in \cref{cor:online:final_res}.
\end{remark}

\cref{thm:online:regret} follows by a series of lemmas.
First, we prove \cref{lem:online:total_rew}, a lower bound on the realized reward of the algorithm in every round, assuming it does not run out of budget.
Specifically, we prove that by any such round $\tau$, the realized reward is $\tau\optInf_\tau - \sum_j L_j \e_R(T_{j-1}) - \order*{k \sqrt T}$ with high probability.
This lemma proves \cref{thm:online:regret}, assuming that the budget is not depleted by some round $\tau$ close enough to $T$.
This is implied by \cref{lem:online:total_pay}, where we prove an upper bound on the realized spending of the algorithm by any round.
Specifically, \cref{lem:online:total_pay} proves that the total spending by any round $\tau$ is at most $\tau \rho + \sum_j L_j \e_C(T_{j-1}) + \order*{k\sqrt\tau}$ with high probability.
Combining these two lemmas, we get that by round $\tau = T - \sum_j L_j\e_C(T_{j-1}) - \order*{k\sqrt T}$ the budget is not depleted, and therefore the total reward is at least $\tau\optInf_\tau - \sum_i L_j \e_R(T_{j-1}) - \order*{k \sqrt T}$, with high probability.
The detailed proof of \cref{thm:online:regret} is in \cref{ssec:app:online:final}, after proving the aforementioned lemmas.

We now present the lower bound on the realized reward.
Specifically, we show that the total reward of the first $\tau$ rounds is at least $\tau \optInf_m$ with some error for any $\tau$.
The first part of the error comes from the sub-optimality of the bidding in epoch $i$, $\e_R(T_{i-1})$.
The second part of the error comes from the error in concentration.

\begin{restatable}{lemma}{LemOnlineTotalRew}
\label{lem:online:total_rew}
    Let $w_t \in \{0, 1\}$ indicate whether the bidder got a conversion in round $t$.
    Assume $k \ge m + \frac{1}{\bar c} \log T$.
    Fix round $\tau \le T - k$ and let $I_\tau$ be the epoch of round $\tau$.
    Then the total realized reward up to round $\tau$ is at least
    \begin{equation*}
        \tau \optInf_m
        -
        \sum_{j=1}^{I_\tau} L_j \e_R(T_{j-1})
        -
        \order{ k \sqrt{T \log\frac{T}{\d}} }
    \end{equation*}
    with probability at least $1 - \d$ for any $\d > 0$.
\end{restatable}

The proof of the lemma is quite convoluted to get the high probability guarantee.
The proof involves bounding quantities for which we do not have useful information.
For example, the expected reward of an epoch might not correlate with the expected reward between wins with conversions of an optimal solution.
First, we prove that the realized reward across epochs is close to the reward of their expectation.
Then, we use the fact that the total expected reward of epoch $i$ equals the expected average-time reward of $\vecestbb{i}$ (which is close to $\optInf_m$) times the expected return time to state $1$ of $\vecestbb{i}$.
This step results in $\optInf_m$ appearing.
However, now we have to handle the expected return time to state $1$ of $\vecestbb{i}$.
We show that the sum of these expected lengths is close to their realized values, which sum up to $\tau$.
Throughout the proof, we take advantage of the fact that an epoch ending early happens with high probability.
This makes quantities like an epoch's expected length very close to the return time to state $1$ in the infinite horizon setting.
We defer the full proof to \cref{ssec:app:online:totalRew}.

We now present the upper bound on the realized payment.
We show that by any round $\tau$, the total payment of the algorithm is at most $\tau \rho$ plus some error.
Similar to \cref{lem:online:total_rew}, this error comes from the overpayment due to suboptimal bidding in each epoch $i$, $\e_C(T_{i-1})$, and some error due to concentration.

\begin{restatable}{lemma}{LemOnlineTotalPay}
\label{lem:online:total_pay}
    Assume that $k \ge m + \frac{1}{\bar c} \log T$.
    Fix a round $\tau$ and let $I_\tau$ be its epoch.
    For any $\d > 0 $, with probability at least $1 - \d$ the total payment of the algorithm until round $\tau$ is at most
    \begin{align*}
        \tau \rho
        +
        \sum_{j = 1}^{I_\tau} L_j \e_C(T_{j-1})
        +
        \order{k \sqrt{T \log \frac{T}{\d}}}
    \end{align*}
\end{restatable}

The proof follows similar steps to the proof in \cref{lem:online:total_rew}.
First, we prove that with high probability, the total payment of the epochs is close to their expected spending between wins with conversions in the infinite horizon setting.
However, this is not useful directly.
We show that for each epoch, this is equal to the expected average spending of the bidding of that epoch, times the return time to state $1$.
Similar to \cref{lem:online:total_rew}, the expected average spending is close to $\rho$, but the return time is not directly useful.
We show that these return times are about the realized length of the epochs, which equals $\tau$, the round we examine in \cref{lem:online:total_pay}.
The full proof can be found in \cref{ssec:app:online:totalPay}.

\section{Approximating \texorpdfstring{$\optInf$}{the Infinite Horizon Problem} from Samples} \label{sec:calc}

In this section, we examine the calculation of the optimal solution of $\optInf_m$, the benchmark we use in \cref{algo:online}.
Specifically, we prove the following theorem, showing that using $t$ price/conversion rate samples, we can use a linear program of size $\order{t m}$ to calculate the optimal solution of $\optInf_m$ with only $\orderT*{m^3 / \sqrt{t}}$ error.

\begin{restatable}{theorem}{ThmCalc} \label{thm:calc}
    Fix $m \in \N$, $m \ge \frac{2}{\bar c \rho} \log T$.
    Using price/conversion rate samples $(p_1, c_1), \ldots, (p_t, c_t)$, we can construct an empirical vector of mappings from contexts to bids $\vec\bb^t$ with a linear program of size $\order{t m}$ constrained to bid $1$ in state $m$.
    For this vector, for any $\d > 0$, with probability at least $1 - \d$, the expected average payment and reward are
    \begin{align*}
        C(\vec\bb^t)
        \le
        \rho
        +
        \order{m^3 \sqrt{\frac{1}{t} \log\frac{1}{\d}}}
        \quad\text{ and }\quad
        R_m(\vec\bb^t)
        \ge
        \optInf_M
        -
        \order{ \frac{m^3}{\rho} \sqrt{\frac{1}{t} \log\frac{1}{\d}} }
    \end{align*}
    %
\end{restatable}

To prove this theorem, we will provide the following lemmas.
First, in \cref{ssec:calc:single}, \cref{lem:calc:simple_single} proves a crucial result: the optimal bidding does not fully consider the context.
In particular, any optimal solution uses bids of the form $\min(1,\frac{c}{\mu})$ for some $\mu \ge 0$ ($\mu = 0$ corresponds to bidding $1$).
This is the key simplifying step to \cref{thm:calc}.
First, it simplifies the calculation of the optimal solution since we only look for linear mappings from conversion rates to bids, allowing us to compute the optimal solutions with the linear program \eqref{eq:calc:LP} in \cref{ssec:app:calc:lp}.
Second, it simplifies the functions we want to learn: the probability of winning with a conversion $W(\cdot)$ and the expected payment $P(\cdot)$ become monotone mappings from the parameter $\mu$ to $[0, 1]$.
Utilizing this, \cref{lem:calc:sample_error} in \cref{ssec:calc:sample_error} proves a uniform convergence bound on $W(\cdot)$ and $P(\cdot)$.
The full proof can be found in \cref{ssec:app:calc:theorem}, after proving the above lemmas.

\subsection{Simplicity of Optimal Bidding} \label{ssec:calc:single}

This section examines the problem of maximizing the probability for a single conversion subject to an expected budget constraint.
Specifically, we assume that the bidder has per-round budget $\rho' \in (0, 1]$ (note this can depend on the state $\l$ and does not need to be $\rho$) and bids to maximize her probability of a conversion.
She picks a function $b(\cdot)$ that maps a context and a conversion rate to a (potentially randomized) bid.
This results in the following optimization problem:
\begin{equation} \label{eq:calc:single}
\begin{split}
    \sup_{\bb: \calX \to \Delta([0, 1])}
    \quad
    \Ex[x,p]{\Big. c \One{\bb(x) \ge p}}
    \quad\textrm{s.t. }\quad
    \Ex[x,p]{\Big. p \One{\bb(x) \ge p}} \le \rho'
\end{split}
\end{equation}

Analyzing the structure of the above optimization problem will give insight into how the bidder should bid in a single state $\l$ of the MDP of $\optInf$.
We proceed to show that the optimal solution takes a very simple form.

\begin{restatable}{lemma}{LemCalcSingle} \label{lem:calc:simple_single}
    The bidding function of \eqref{eq:calc:single} does not need to depend on the context $x$, only the conversion rate $c$.
    In addition, the optimal bidding takes the following form:
    \begin{itemize}
        \item If the price distribution for a given conversion rate has no atoms, the optimal solution is $b^\mu(c) = \min(1,\frac{c}{\mu})$ for some $\mu \ge 0$; note $\mu = 0$ implies bidding $1$ for any context. 
        \item If the price distribution for a given conversion rate has finite support, the optimal solution is supported on at most two bidding policies, $b^{\mu_1}(c)$ and $b^{\mu_2}(c)$, for some $\mu_1, \mu_2 \ge 0$.
        \item In all other cases, the maximum might not be obtained, but a solution supported on $b^{\mu_1}(c)$ and $b^{\mu_2}(c)$ can get arbitrarily close.
    \end{itemize}
\end{restatable}

The proof starts by defining the Lagrangian of \eqref{eq:calc:single}: $\Ex[x,p]{( c - \mu p ) \One{b(x) \ge p}} + \mu \rho'$ for some multiplier $\mu \ge 0$.
We next observe that the optimal solution to this is to bid $b(c) = \frac{c}{\mu}$.
The rest of the proof involves proving that the optimal solution for the Lagrangian can be translated to an optimal one for the constrained problem.
We defer the full proof in \cref{ssec:app:calc:single}.

The following subsections show how we can approximate $\optInf$ with the empirical distribution of $t$ samples we have collected by round $t$.
Using this empirical distribution with finite support and the fact that optimization problem \eqref{eq:opt_inf_W} can be formulated as a Linear Program, we can find the optimal empirical solution with a Linear Program of size at most $\order*{m t}$.
The Linear Program is written via the occupancy measure, which is the standard re-formulation of the decision variables to make the objectives and the constraint linear.
We present the full Linear Program in \cref{ssec:app:calc:lp}.

\subsection{Approximating \texorpdfstring{$\optInf_m$}{the Infinite Horizon problem} with an Approximate Distribution} \label{ssec:calc:approx}

In this section, we examine finding the optimal solution for the infinite horizon problem using an approximate distribution over prices and conversion rates.
Specifically, we model this as having functions $W'(\mu)$ and $P'(\mu)$ that are close to the real ones, $W(\mu)$ and $P(\mu)$.
We then prove that any vector $\vec\bmu$ has expected average reward and payment that are similar for both functions.
The proof can be found in \cref{ssec:app:calc:approximation}.

\begin{restatable}{lemma}{LemCalcApproximation} \label{lem:calc:approximation}
    Let $W(\mu)$ and $P(\mu)$ be the probability of winning and the expected payment of bid $\min\{1, \frac{c}{\mu} \}$.
    Assume that $W'(\mu), P'(\mu) \in [0, 1]$ are such that $|W(\mu) - W'(\mu)| \le \e$ and $|P(\mu) - P'(\mu)| \le \e$ for all $\mu \ge 0$.
    Fix a $\vec\bmu$ such that $\bmu_m = 0$ (i.e., bids $1$ at state $m$).
    Assume that $m \ge \frac{1}{\bar c}$.
    Then we have
    $
        \abs{ R'(\vec\bmu)  - R(\vec\bmu) }
        \le 
        36 m^2 \e
    $
    and
    $
        \abs{ C'(\vec\bmu) - C(\vec\bmu) }
        \le 
        39 m^3 \e
    $.
\end{restatable}

\subsection{Approximating \texorpdfstring{$W(\mu)$ and $P(\mu)$}{the probability to win and expected payment} with Samples} \label{ssec:calc:sample_error}

This section presents the bounded sample error needed for \cref{lem:calc:approximation}.
Specifically, we show the following lemma that shows that with $n$ samples, the error on $P(\cdot)$ and $W(\cdot)$ is $\orderT*{1/\sqrt{n}}$ with high probability.

\begin{restatable}{lemma}{LemCalcSampleError} \label{lem:calc:sample_error}
    Let $W(\mu)$ and $P(\mu)$ be the probability of winning a conversion and the expected payment when bidding $\min\{1, \frac{c}{\mu} \}$.
    Let $W_n(\mu)$ and $P_n(\mu)$ be the empirical estimates of these two functions using $n$ samples $\{(p_i, c_i)\}_{i \in [n]}$. 
    %
    Then, for all $\d \in (0, 1)$ with probability at least $1 - \d$ it holds that for all $\mu \ge 0$
    \begin{equation*}
        \abs{W_n(\mu) - W(\mu)}
        \le
        \order{\sqrt{\frac{1}{n} \log\frac{2}{\d}}}
        \quad\text{ and }\quad
        \abs{P_n(\mu) - P(\mu)}
        \le
        \order{\sqrt{\frac{1}{n} \log\frac{2}{\d}}}
    \end{equation*}
\end{restatable}

The proof follows by a standard result: the expected maximum error, $\Ex{\sup_\mu\abs{W(\mu) - W_n(\mu)}}$, is $\order{1/\sqrt n}$.
Using McDiarmid's inequality, the lemma follows, converting that bound to a high probability one.
The full proof can be found in \cref{ssec:app:calc:sample_error}



{\small
\bibliographystyle{acm}
\bibliography{ref}

\begin{thebibliography}{10}

\bibitem{BookRL}
{\sc Agarwal, A., Jiang, N., Kakade, S.~M., and Sun, W.}
\newblock Reinforcement learning: Theory and algorithms.
\newblock {\em CS Dept., UW Seattle, Seattle, WA, USA, Tech. Rep 32\/} (2019), 96.

\bibitem{aggarwal2024auto}
{\sc Aggarwal, G., Badanidiyuru, A., Balseiro, S.~R., Bhawalkar, K., Deng, Y., Feng, Z., Goel, G., Liaw, C., Lu, H., Mahdian, M., et~al.}
\newblock Auto-bidding and auctions in online advertising: A survey.
\newblock {\em ACM SIGecom Exchanges 22}, 1 (2024), 159--183.

\bibitem{DBLP:journals/corr/AggarwalFZ24}
{\sc Aggarwal, G., Fikioris, G., and Zhao, M.}
\newblock No-regret algorithms in non-truthful auctions with budget and {ROI} constraints.
\newblock In {\em Proceedings of the ACM Web Conference 2025\/} (New York, NY, USA, 2025), Association for Computing Machinery.

\bibitem{aggarwal2024randomized}
{\sc Aggarwal, G., Gupta, A., Perlroth, A., and Velegkas, G.}
\newblock Randomized truthful auctions with learning agents.
\newblock In {\em Advances in Neural Information Processing Systems 38: Annual Conference on Neural Information Processing Systems 2024, NeurIPS 2024, Vancouver, BC, Canada, December 10 - 15, 2024\/} (Red Hook, NY, USA, 2024), Curran Associates Inc.

\bibitem{alimohammadi2023incentive}
{\sc Alimohammadi, Y., Mehta, A., and Perlroth, A.}
\newblock Incentive compatibility in the auto-bidding world.
\newblock In {\em Proceedings of the 24th ACM Conference on Economics and Computation\/} (New York, NY, USA, 2023), EC '23, Association for Computing Machinery, p.~63.

\bibitem{DBLP:conf/icml/AzarOM17}
{\sc Azar, M.~G., Osband, I., and Munos, R.}
\newblock Minimax regret bounds for reinforcement learning.
\newblock In {\em Proceedings of the 34th International Conference on Machine Learning, {ICML} 2017, Sydney, NSW, Australia, 6-11 August 2017\/} (Sydney, NSW, Australia, 2017), vol.~70 of {\em Proceedings of Machine Learning Research}, {PMLR}, pp.~263--272.

\bibitem{DBLP:conf/focs/BadanidiyuruKS13}
{\sc Badanidiyuru, A., Kleinberg, R., and Slivkins, A.}
\newblock Bandits with knapsacks.
\newblock In {\em 54th Annual {IEEE} Symposium on Foundations of Computer Science, {FOCS} 2013, 26-29 October, 2013, Berkeley, CA, {USA}\/} (Los Alamitos, CA, USA, 2013), {IEEE} Computer Society, pp.~207--216.

\bibitem{balseiro2015repeated}
{\sc Balseiro, S.~R., Besbes, O., and Weintraub, G.~Y.}
\newblock Repeated auctions with budgets in ad exchanges: Approximations and design.
\newblock {\em Management Science 61}, 4 (2015), 864--884.

\bibitem{balseiro2023field}
{\sc Balseiro, S.~R., Bhawalkar, K., Feng, Z., Lu, H., Mirrokni, V., Sivan, B., and Wang, D.}
\newblock A field guide for pacing budget and ros constraints.
\newblock In {\em Proceedings of the 41st International Conference on Machine Learning\/} (Vienna, Austria, 2024), ICML'24, JMLR.org.

\bibitem{DBLP:journals/mansci/BalseiroG19}
{\sc Balseiro, S.~R., and Gur, Y.}
\newblock Learning in repeated auctions with budgets: Regret minimization and equilibrium.
\newblock {\em Manag. Sci. 65}, 9 (2019), 3952--3968.

\bibitem{banchio2022artificial}
{\sc Banchio, M., and Skrzypacz, A.}
\newblock Artificial intelligence and auction design.
\newblock In {\em Proceedings of the 23rd ACM Conference on Economics and Computation\/} (New York, NY, USA, 2022), Association for Computing Machinery, pp.~30--31.

\bibitem{basu2023stablefees}
{\sc Basu, S., Easley, D., O’Hara, M., and Sirer, E.~G.}
\newblock Stablefees: A predictable fee market for cryptocurrencies.
\newblock {\em Management Science 69}, 11 (2023), 6508--6524.

\bibitem{bichler2023convergence}
{\sc Bichler, M., Lunowa, S.~B., Oberlechner, M., Pieroth, F.~R., and Wohlmuth, B.}
\newblock On the convergence of learning algorithms in bayesian auction games.
\newblock {\em arXiv preprint arXiv:2311.15398 1\/} (2023).

\bibitem{blum2008regret}
{\sc Blum, A., Hajiaghayi, M., Ligett, K., and Roth, A.}
\newblock Regret minimization and the price of total anarchy.
\newblock In {\em Proceedings of the Fortieth Annual ACM Symposium on Theory of Computing\/} (New York, NY, USA, 2008), Association for Computing Machinery, pp.~373--382.

\bibitem{borgs2005multi}
{\sc Borgs, C., Chayes, J., Immorlica, N., Mahdian, M., and Saberi, A.}
\newblock Multi-unit auctions with budget-constrained bidders.
\newblock In {\em Proceedings of the 6th ACM Conference on Electronic Commerce\/} (New York, NY, USA, 2005), Association for Computing Machinery, pp.~44--51.

\bibitem{borgs2007GFPdynamics}
{\sc Borgs, C., Chayes, J., Immorlica, N., Mahdian, M., and Saberi, A.}
\newblock Dynamics for the gsp auction: Stability and efficiency.
\newblock {\em ACM Transactions on Economics and Computation (TEAC) 5}, 2 (2007), 1--29.

\bibitem{braverman2018selling}
{\sc Braverman, M., Mao, J., Schneider, J., and Weinberg, M.}
\newblock Selling to a no-regret buyer.
\newblock In {\em Proceedings of the 2018 ACM Conference on Economics and Computation\/} (New York, NY, USA, 2018), Association for Computing Machinery, pp.~523--538.

\bibitem{broadbent1993advertising}
{\sc Broadbent, S.}
\newblock Advertising effects: More than short term.
\newblock {\em Market Research Society. Journal. 35}, 1 (1993), 1--11.

\bibitem{broadbent2000advertisements}
{\sc Broadbent, S.}
\newblock What do advertisements really do for brands?
\newblock {\em International Journal of Advertising 19}, 2 (2000), 147--165.

\bibitem{broadbent1995adstock}
{\sc Broadbent, S., and Fry, T.}
\newblock Adstock modelling for the long term.
\newblock {\em Market Research Society. Journal. 37}, 4 (1995), 1--18.

\bibitem{cai2023selling}
{\sc Cai, L., Weinberg, S.~M., Wildenhain, E., and Zhang, S.}
\newblock Selling to multiple no-regret buyers.
\newblock In {\em International Conference on Web and Internet Economics\/} (Berlin, Heidelberg, 2023), Springer, Springer-Verlag, pp.~113--129.

\bibitem{caragiannis2015bounding}
{\sc Caragiannis, I., Kaklamanis, C., Kanellopoulos, P., Kyropoulou, M., Lucier, B., Leme, R.~P., and Tardos, E.}
\newblock Bounding the inefficiency of outcomes in generalized second price auctions.
\newblock {\em Journal of Economic Theory 156\/} (2015), 343--388.

\bibitem{DBLP:conf/icml/CastiglioniCK22}
{\sc Castiglioni, M., Celli, A., and Kroer, C.}
\newblock Online learning with knapsacks: the best of both worlds.
\newblock In {\em International Conference on Machine Learning, {ICML} 2022, 17-23 July 2022, Baltimore, Maryland, {USA}\/} (Baltimore, Maryland, USA, 2022), vol.~162 of {\em Proceedings of Machine Learning Research}, {PMLR}, pp.~2767--2783.

\bibitem{cesa2014regret}
{\sc Cesa-Bianchi, N., Gentile, C., and Mansour, Y.}
\newblock Regret minimization for reserve prices in second-price auctions.
\newblock {\em IEEE Transactions on Information Theory 61}, 1 (2014), 549--564.

\bibitem{cesa2006prediction}
{\sc Cesa-Bianchi, N., and Lugosi, G.}
\newblock {\em Prediction, learning, and games}.
\newblock Cambridge university press, USA, 2006.

\bibitem{DBLP:conf/icml/Chen0L22a}
{\sc Chen, L., Jain, R., and Luo, H.}
\newblock Learning infinite-horizon average-reward markov decision process with constraints.
\newblock In {\em International Conference on Machine Learning, {ICML} 2022, 17-23 July 2022, Baltimore, Maryland, {USA}\/} (Baltimore, Maryland, USA, 2022), vol.~162 of {\em Proceedings of Machine Learning Research}, {PMLR}, pp.~3246--3270.

\bibitem{DBLP:conf/sigecom/ChenKK21}
{\sc Chen, X., Kroer, C., and Kumar, R.}
\newblock The complexity of pacing for second-price auctions.
\newblock In {\em {EC} '21: The 22nd {ACM} Conference on Economics and Computation, Budapest, Hungary, July 18-23, 2021\/} (New York, NY, USA, 2021), {ACM}, p.~318.

\bibitem{choi2020online}
{\sc Choi, H., Mela, C.~F., Balseiro, S.~R., and Leary, A.}
\newblock Online display advertising markets: A literature review and future directions.
\newblock {\em Information Systems Research 31}, 2 (2020), 556--575.

\bibitem{chu2020position}
{\sc Chu, L.~Y., Nazerzadeh, H., and Zhang, H.}
\newblock Position ranking and auctions for online marketplaces.
\newblock {\em Management Science 66}, 8 (2020), 3617--3634.

\bibitem{DBLP:conf/ec/ConitzerKPSSMW19}
{\sc Conitzer, V., Kroer, C., Panigrahi, D., Schrijvers, O., Sodomka, E., Moses, N. E.~S., and Wilkens, C.}
\newblock Pacing equilibrium in first-price auction markets.
\newblock In {\em Proceedings of the 2019 {ACM} Conference on Economics and Computation, {EC} 2019, Phoenix, AZ, USA, June 24-28, 2019\/} (New York, NY, USA, 2019), Association for Computing Machinery, p.~587.

\bibitem{DBLP:conf/wine/ConitzerKSM18}
{\sc Conitzer, V., Kroer, C., Sodomka, E., and Moses, N. E.~S.}
\newblock Multiplicative pacing equilibria in auction markets.
\newblock In {\em Web and Internet Economics - 14th International Conference, {WINE} 2018, Oxford, UK, December 15-17, 2018, Proceedings\/} (Linthicum, MD, USA, 2018), G.~Christodoulou and T.~Harks, Eds., vol.~11316 of {\em Lecture Notes in Computer Science}, Springer, p.~443.

\bibitem{craig1976advertising}
{\sc Craig, C.~S., Sternthal, B., and Leavitt, C.}
\newblock Advertising wearout: An experimental analysis.
\newblock {\em Journal of Marketing Research 13}, 4 (1976), 365--372.

\bibitem{daskalakis2016learning}
{\sc Daskalakis, C., and Syrgkanis, V.}
\newblock Learning in auctions: Regret is hard, envy is easy.
\newblock In {\em IEEE Annual Symposium on Foundations of Computer Science, {FOCS}\/} (New Brunswick, New Jersey, USA, 2016), pp.~219--228.

\bibitem{deng2022nash}
{\sc Deng, X., Hu, X., Lin, T., and Zheng, W.}
\newblock Nash convergence of mean-based learning algorithms in first price auctions.
\newblock In {\em Proceedings of the ACM Web Conference 2022\/} (New York, NY, USA, 2022), Association for Computing Machinery, pp.~141--150.

\bibitem{dobzinski2012multi}
{\sc Dobzinski, S., Lavi, R., and Nisan, N.}
\newblock Multi-unit auctions with budget limits.
\newblock {\em Games and Economic Behavior 74}, 2 (2012), 486--503.

\bibitem{dobzinski2014efficiency}
{\sc Dobzinski, S., and Leme, R.~P.}
\newblock Efficiency guarantees in auctions with budgets.
\newblock In {\em International Colloquium on Automata, Languages, and Programming\/} (Berlin, Heidelberg, 2014), Springer, Springer Berlin Heidelberg, pp.~392--404.

\bibitem{edelman2007internet}
{\sc Edelman, B., Ostrovsky, M., and Schwarz, M.}
\newblock Internet advertising and the generalized second-price auction: Selling billions of dollars worth of keywords.
\newblock {\em American economic review 97}, 1 (2007), 242--259.

\bibitem{feldman2012revenue}
{\sc Feldman, M., Fiat, A., Leonardi, S., and Sankowski, P.}
\newblock Revenue maximizing envy-free multi-unit auctions with budgets.
\newblock In {\em Proceedings of the 13th ACM conference on electronic commerce\/} (New York, NY, USA, 2012), Association for Computing Machinery, pp.~532--549.

\bibitem{feldman2013limits}
{\sc Feldman, M., Lucier, B., and Syrgkanis, V.}
\newblock Limits of efficiency in sequential auctions.
\newblock In {\em International Conference on Web and Internet Economics\/} (Berlin, Heidelberg, 2013), Springer, Springer Berlin Heidelberg, pp.~160--173.

\bibitem{feng2024strategic}
{\sc Feng, Y., Lucier, B., and Slivkins, A.}
\newblock Strategic budget selection in a competitive autobidding world.
\newblock In {\em Proceedings of the 56th Annual ACM Symposium on Theory of Computing\/} (New York, NY, USA, 2024), Association for Computing Machinery, pp.~213--224.

\bibitem{ferreira2021dynamic}
{\sc Ferreira, M.~V., Moroz, D.~J., Parkes, D.~C., and Stern, M.}
\newblock Dynamic posted-price mechanisms for the blockchain transaction-fee market.
\newblock In {\em Proceedings of the 3rd ACM Conference on Advances in Financial Technologies\/} (New York, NY, USA, 2021), Association for Computing Machinery, pp.~86--99.

\bibitem{DBLP:conf/colt/FikiorisT23}
{\sc Fikioris, G., and Tardos, {\'{E}}.}
\newblock Approximately stationary bandits with knapsacks.
\newblock In {\em The Thirty Sixth Annual Conference on Learning Theory, {COLT} 2023, 12-15 July 2023, Bangalore, India\/} (Bangalore, India, 2023), G.~Neu and L.~Rosasco, Eds., vol.~195 of {\em Proceedings of Machine Learning Research}, {PMLR}, pp.~3758--3782.

\bibitem{fikioris2023liquid}
{\sc Fikioris, G., and Tardos, {\'E}.}
\newblock Liquid welfare guarantees for no-regret learning in sequential budgeted auctions.
\newblock In {\em Proceedings of the 24th ACM Conference on Economics and Computation\/} (New York, NY, USA, 2023), Association for Computing Machinery, pp.~678--698.

\bibitem{DBLP:conf/innovations/GaitondeLLLS23}
{\sc Gaitonde, J., Li, Y., Light, B., Lucier, B., and Slivkins, A.}
\newblock Budget pacing in repeated auctions: Regret and efficiency without convergence.
\newblock In {\em 14th Innovations in Theoretical Computer Science Conference, {ITCS} 2023, January 10-13, 2023, MIT, Cambridge, Massachusetts, {USA}\/} (Cambridge, Massachusetts, {USA}, 2023), vol.~251 of {\em LIPIcs}, Schloss Dagstuhl - Leibniz-Zentrum f{\"{u}}r Informatik, pp.~52:1--52:1.

\bibitem{gentry2018structural}
{\sc Gentry, M.~L., Hubbard, T.~P., Nekipelov, D., Paarsch, H.~J., et~al.}
\newblock {\em Structural Econometrics of Auctions: A Review}.
\newblock now publishers, USA, 2018.

\bibitem{ha1997does}
{\sc Ha, L., and Litman, B.~R.}
\newblock Does advertising clutter have diminishing and negative returns?
\newblock {\em Journal of Advertising 26}, 1 (1997), 31--42.

\bibitem{DBLP:conf/focs/ImmorlicaSSS19}
{\sc Immorlica, N., Sankararaman, K.~A., Schapire, R.~E., and Slivkins, A.}
\newblock Adversarial bandits with knapsacks.
\newblock In {\em 60th {IEEE} Annual Symposium on Foundations of Computer Science, {FOCS} 2019, Baltimore, Maryland, USA, November 9-12, 2019\/} (Baltimore, Maryland, USA, 2019), {IEEE} Computer Society, pp.~202--219.

\bibitem{DBLP:journals/corr/KleinbergLST23}
{\sc Kleinberg, B., Leme, R.~P., Schneider, J., and Teng, Y.}
\newblock U-calibration: Forecasting for an unknown agent.
\newblock In {\em Proceedings of Thirty Sixth Conference on Learning Theory\/} (Bangalore, India, 12--15 Jul 2023), vol.~195 of {\em Proceedings of Machine Learning Research}, PMLR, pp.~5143--5145.

\bibitem{kolumbus2024paying}
{\sc Kolumbus, Y., Halpern, J., and Tardos, {\'E}.}
\newblock Paying to do better: Games with payments between learning agents.
\newblock {\em arXiv preprint arXiv:2405.20880 1\/} (2024).

\bibitem{kolumbus2022auctions}
{\sc Kolumbus, Y., and Nisan, N.}
\newblock Auctions between regret-minimizing agents.
\newblock In {\em ACM Web Conference, {WebConf}\/} (New York, NY, USA, 2022), Association for Computing Machinery, pp.~100--111.

\bibitem{KolumbusN22}
{\sc Kolumbus, Y., and Nisan, N.}
\newblock How and why to manipulate your own agent: On the incentives of users of learning agents.
\newblock In {\em Annual Conference on Neural Information Processing Systems, {NeurIPS}\/} (Red Hook, NY, USA, 2022), Curran Associates Inc.

\bibitem{krishna2009auction}
{\sc Krishna, V.}
\newblock {\em Auction theory}.
\newblock Academic press, USA, 2009.

\bibitem{kumar2022optimal}
{\sc Kumar, B., Morgenstern, J., and Schrijvers, O.}
\newblock Optimal spend rate estimation and pacing for ad campaigns with budgets.
\newblock {\em arXiv preprint arXiv:2202.05881 1\/} (2022).

\bibitem{DBLP:conf/nips/KumarK22}
{\sc Kumar, R., and Kleinberg, R.}
\newblock Non-monotonic resource utilization in the bandits with knapsacks problem.
\newblock In {\em Advances in Neural Information Processing Systems 35: Annual Conference on Neural Information Processing Systems 2022, NeurIPS 2022, New Orleans, LA, USA, November 28 - December 9, 2022\/} (Red Hook, NY, USA, 2022), Curran Associates Inc.

\bibitem{lewis2014online}
{\sc Lewis, R.~A., and Reiley, D.~H.}
\newblock Online ads and offline sales: measuring the effect of retail advertising via a controlled experiment on yahoo!
\newblock {\em Quantitative Marketing and Economics 12\/} (2014), 235--266.

\bibitem{DBLP:conf/colt/LucierPSZ24}
{\sc Lucier, B., Pattathil, S., Slivkins, A., and Zhang, M.}
\newblock Autobidders with budget and {ROI} constraints: Efficiency, regret, and pacing dynamics.
\newblock In {\em The Thirty Seventh Annual Conference on Learning Theory, June 30 - July 3, 2023, Edmonton, Canada\/} (Edmonton, Canada, 2024), vol.~247 of {\em Proceedings of Machine Learning Research}, {PMLR}, pp.~3642--3643.

\bibitem{milgrom2021auction}
{\sc Milgrom, P.}
\newblock Auction research evolving: Theorems and market designs.
\newblock {\em American Economic Review 111}, 5 (2021), 1383--1405.

\bibitem{mohri2014optimal}
{\sc Mohri, M., and Munoz, A.}
\newblock Optimal regret minimization in posted-price auctions with strategic buyers.
\newblock In {\em Advances in Neural Information Processing Systems\/} (Red Hook, NY, USA, 2014), Curran Associates Inc., pp.~1871--1879.

\bibitem{morgenstern2016learning}
{\sc Morgenstern, J., and Roughgarden, T.}
\newblock Learning simple auctions.
\newblock In {\em 29th Annual Conference on Learning Theory\/} (Columbia University, New York, New York, USA, 23--26 Jun 2016), V.~Feldman, A.~Rakhlin, and O.~Shamir, Eds., vol.~49 of {\em Proceedings of Machine Learning Research}, PMLR, pp.~1298--1318.

\bibitem{naples1997effective}
{\sc Naples, M.~J.}
\newblock Effective frequency: Then and now.
\newblock {\em Journal of Advertising Research 37}, 4 (1997), 7--13.

\bibitem{MAL-077}
{\sc Nedelec, T., Calauzènes, C., Karoui, N.~E., and Perchet, V.}
\newblock Learning in repeated auctions.
\newblock {\em Foundations and Trends® in Machine Learning 15}, 3 (2022), 176--334.

\bibitem{nekipelov2015econometrics}
{\sc Nekipelov, D., Syrgkanis, V., and Tardos, E.}
\newblock Econometrics for learning agents.
\newblock In {\em Proceedings of the sixteenth ACM conference on economics and computation\/} (New York, NY, USA, 2015), Association for Computing Machinery, pp.~1--18.

\bibitem{nisan2023serial}
{\sc Nisan, N.}
\newblock Serial monopoly on blockchains.
\newblock {\em CoRR abs/2311.12731\/} (2023).

\bibitem{noti2017empirical}
{\sc Nisan, N., and Noti, G.}
\newblock An experimental evaluation of regret-based econometrics.
\newblock In {\em Proceedings of the 26th International Conference on World Wide Web\/} (Republic and Canton of Geneva, CHE, 2017), International World Wide Web Conferences Steering Committee, pp.~73--81.

\bibitem{nisan2017quantal}
{\sc Nisan, N., and Noti, G.}
\newblock A "quantal regret" method for structural econometrics in repeated games.
\newblock In {\em Proceedings of the 2017 ACM Conference on Economics and Computation\/} (New York, NY, USA, 2017), Association for Computing Machinery.

\bibitem{noti2021bid}
{\sc Noti, G., and Syrgkanis, V.}
\newblock Bid prediction in repeated auctions with learning.
\newblock In {\em Proceedings of the Web Conference 2021\/} (New York, NY, USA, 2021), Association for Computing Machinery, pp.~3953--3964.

\bibitem{perlroth2023auctions}
{\sc Perlroth, A., and Mehta, A.}
\newblock Auctions without commitment in the auto-bidding world.
\newblock In {\em Proceedings of the ACM Web Conference 2023\/} (New York, NY, USA, 2023), Association for Computing Machinery, pp.~3478--3488.

\bibitem{DBLP:books/wi/Puterman94}
{\sc Puterman, M.~L.}
\newblock {\em Markov Decision Processes: Discrete Stochastic Dynamic Programming}.
\newblock Wiley Series in Probability and Statistics. Wiley, USA, 1994.

\bibitem{roughgarden2010algorithmic}
{\sc Roughgarden, T.}
\newblock Algorithmic game theory.
\newblock {\em Communications of the ACM 53}, 7 (2010), 78--86.

\bibitem{roughgarden2024transaction}
{\sc Roughgarden, T.}
\newblock Transaction fee mechanism design.
\newblock {\em Journal of the ACM 71}, 4 (2024), 1--25.

\bibitem{roughgarden2017price}
{\sc Roughgarden, T., Syrgkanis, V., and Tardos, E.}
\newblock The price of anarchy in auctions.
\newblock {\em Journal of Artificial Intelligence Research 59\/} (2017), 59--101.

\bibitem{roughgarden2019minimizing}
{\sc Roughgarden, T., and Wang, J.~R.}
\newblock Minimizing regret with multiple reserves.
\newblock {\em ACM Transactions on Economics and Computation (TEAC) 7}, 3 (2019).

\bibitem{rubinstein2024strategizing}
{\sc Rubinstein, A., and Zhao, J.}
\newblock Strategizing against no-regret learners in first-price auctions.
\newblock In {\em Proceedings of the 25th ACM Conference on Economics and Computation\/} (New York, NY, USA, 2024), EC '24, Association for Computing Machinery, p.~894–921.

\bibitem{DBLP:journals/corr/Slivkins19}
{\sc Slivkins, A.}
\newblock Introduction to multi-armed bandits.
\newblock {\em CoRR abs/1904.07272\/} (2019).

\bibitem{DBLP:conf/colt/SlivkinsSF23}
{\sc Slivkins, A., Sankararaman, K.~A., and Foster, D.~J.}
\newblock Contextual bandits with packing and covering constraints: {A} modular lagrangian approach via regression.
\newblock In {\em The Thirty Sixth Annual Conference on Learning Theory, {COLT} 2023, 12-15 July 2023, Bangalore, India\/} (Bangalore, India, 2023), G.~Neu and L.~Rosasco, Eds., vol.~195 of {\em Proceedings of Machine Learning Research}, {PMLR}, pp.~4633--4656.

\bibitem{DBLP:conf/icml/StradiGGC0024}
{\sc Stradi, F.~E., Germano, J., Genalti, G., Castiglioni, M., Marchesi, A., and Gatti, N.}
\newblock Online learning in cmdps: Handling stochastic and adversarial constraints.
\newblock In {\em Forty-first International Conference on Machine Learning, {ICML} 2024, Vienna, Austria, July 21-27, 2024\/} (Vienna, Austria, 2024), OpenReview.net.

\bibitem{sun2020empirical}
{\sc Sun, H., Fan, M., and Tan, Y.}
\newblock An empirical analysis of seller advertising strategies in an online marketplace.
\newblock {\em Information Systems Research 31}, 1 (2020), 37--56.

\bibitem{sundararajan2016prediction}
{\sc Sundararajan, M., and Talgam-Cohen, I.}
\newblock Prediction and welfare in ad auctions.
\newblock {\em Theory of Computing Systems 59\/} (2016), 664--682.

\bibitem{varian2009online}
{\sc Varian, H.~R.}
\newblock Online ad auctions.
\newblock {\em American Economic Review 99}, 2 (2009), 430--434.

\bibitem{weed2016online}
{\sc Weed, J., Perchet, V., and Rigollet, P.}
\newblock Online learning in repeated auctions.
\newblock In {\em Conference on Learning Theory\/} (2016), PMLR, pp.~1562--1583.

\bibitem{weinberg1982econometric}
{\sc Weinberg, C.~B., and Weiss, D.~L.}
\newblock On the econometric measurement of the duration of advertising effect on sales.
\newblock {\em Journal of Marketing Research 19}, 4 (1982), 585--591.

\end{thebibliography}
}

\newpage
\appendix
\section*{Appendix}
\section{Warm-up Example: Further Details}\label{appendix:warm-up}

\begin{figure}[!ht]
\centering 
\includegraphics[width=0.5\linewidth]{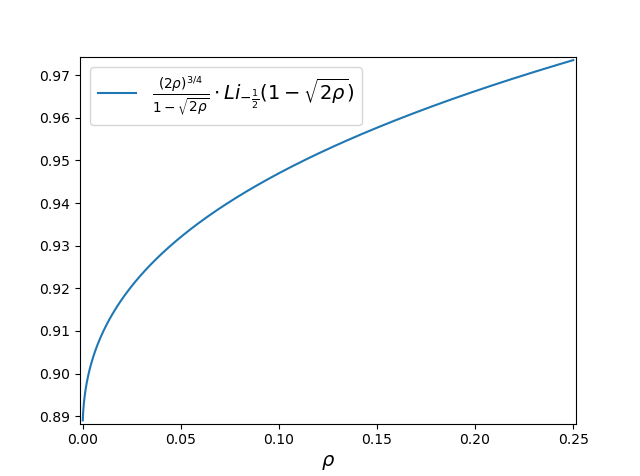}
\Description{A plot for every value $\rho \in [0, 1/4]$ starting from $0.89$ and increasing to $0.97$.}
\caption{Numerical evaluation of the competitive ratio of using a fixed bid distribution, shown as a function of the budget per step $\rho$ in the example of uniform prices and square-root objective.}
\label{fig:uniform-price-example}
\end{figure}

\paragraph{Calculation of the utility with a fixed bid:} An upper bound on our bidder's utility is the scenario where we have the same number of wins, and the winning times are equally spaced in intervals of $T/\sqrt{2BT}$, which gives an upper bound on the utility of $(2 \rho)^{\nicefrac{1}{4}} T$. 

While perfect spacing is not attainable, we would like to see how far our bidder's utility is from this value.
Suppose there are $k$ winning events, and the intervals between wins (including times zero and $T$ as interval ends) are $\l_1, \dots, \l_{k+1}$. The utility for our bidder is $u = \sum_{i=1}^{k+1} \sqrt{\l_i}$.
By Wald's equality, the expected utility equals the product of expectations:
$$
\E[u] = \sum_{i = 1}^{k+1} \sqrt{\l_i} = \E[k + 1]\E[\sqrt{\l_i}].
$$

The number of wins follows a binomial distribution with expectation $\E[k] = bT$ (to simplify, we henceforth omit one interval and look at $k$ rather than $k+1$). 
Except the first and last intervals, the interval lengths are geometrically distributed; the expectation is approximately (up to an $o(1)$ error as $T \rightarrow \infty$. See below). 
$$
\E[\sqrt{\l_i}] = 
b \sum_{\l=1}^\infty \sqrt{\l} (1-b)^{\l-1} = \frac{b}{1-b} Li_{-\frac{1}{2}}(1-b).
$$
where $Li_s(x)$ is the polylogarithm function, defined as $Li_s(x) = \sum_{n=1}^\infty \frac{x^n}{n^s}$.  Thus, the long-term expected utility is
$$
\E[u] = \frac{\sqrt{2\rho}}{1 - \sqrt{2\rho}} \cdot \sqrt{2\rho} T \cdot Li_{-\frac{1}{2}}(1-\sqrt{2\rho}) = 
\frac{2\rho T}{1 - \sqrt{2\rho}} \cdot  Li_{-\frac{1}{2}}(1-\sqrt{2\rho}) =
\qty( \frac{\sqrt\pi}{2^{3/4}} \rho^{1/4} + \order{\rho^{3/4}} ) T.
$$
The ratio between this expected utility and the optimal benchmark,
$$
\frac{\E[u]}{(2 \rho)^{\nicefrac{1}{4}} T} = 
\frac{(2\rho)^{\nicefrac{3}{4}}}{1 - \sqrt{2\rho}} \cdot  Li_{-\frac{1}{2}}(1-\sqrt{2\rho}) =
\frac{\sqrt\pi}{2} + \order{\sqrt\rho},
$$
is an increasing function of $\rho$ that is strictly less than one for all $\rho < 1/2$. 
By numerical evaluation (see Figure \ref{fig:uniform-price-example}), 
in our case where $\rho < 1/4$, the fixed bidding strategy obtains an approximation ratio between $\approx0.886$ and $\approx 0.973$ of the optimal utility.

\paragraph{Error due to infinite horizon summation} 
The error term arises from the fact that with a finite time horizon $T$, the values of $\l$ are not infinite. However, the error from not truncating the sum vanishes as $T$ becomes large:  
the expectation of $\sqrt{\l}$ is
$$
\E[\sqrt{\l_i}] = 
z \bigg(
b \sum_{\l=1}^\infty \sqrt{\l} (1-b)^{\l-1} - 
b \sum_{\l=T+1}^\infty \sqrt{\l} (1-b)^{\l-1}.
\bigg)
$$
where $z = \frac{1}{1 - (1-b)^T} = 1-o(1)$ is the normalization constant from the truncated geometric distribution. The last term is bounded by
$$
b \sum_{\l=T+1}^\infty \sqrt{\l} (1-b)^{\l-1} <
\frac{b}{1-b} \sum_{\l=T+1}^\infty \l (1-b)^{\l} = 
\frac{b}{1-b} \cdot \frac{(1-b)^{T+1}(bT + 1)}{b^2} = 
\frac{1}{b} \cdot (1-b)^{T}(bT + 1) \xrightarrow[T \to \infty]{} 0.
$$

\section{State-Independent Strategies}\label{appendix:state-independent}

\subsection{Utility Guarantee for State Independent Bidding}
\label{appendix:reverse-jensen}

In our example in the introduction (details in \cref{appendix:warm-up}), we saw that the static bidding policy of using a constant bid attained a constant factor approximation to the utility of an optimal policy. 
The following theorem shows that the example is, in fact, an instance of a more general result.

\begin{theorem} \label{thm:state-independent}
    There exists a state-independent static policy that depends only on the context that achieves an $(1-\frac{1}{e}$) fraction of the optimal policy as $T \rightarrow \infty$.
\end{theorem}

To prove this theorem, we need the following key Lemma.

\begin{lemma}[A Reverse Jensen Inequality for Geometric Random Variables] \label{lem:reverse-jensen}
    If $\rfunc : \N \to [0,\infty)$ is a non-decreasing concave function, $X$ is a geometric random variable supported on the positive integers, and $Y$ is an integer-valued random variable satisfying $\E[Y] = \E[X]$, then 
    \begin{equation} \label{eq:lem-reverse-jensen}
        \E[\rfunc(X)] \geq \left( 1 - \tfrac{1}{e} \right) \E[\rfunc(Y)] .
    \end{equation}
\end{lemma}\label{thm:reverse-jensen-lemma}

The lemma can be thought of as a reverse version of Jensen's inequality and it applies generally to geometric random variables.
In our context, this result implies that any concave reward function and any distribution of prices, there exists a static bidding policy that achieves at least $1 - \frac{1}{e} - o(1)$ times the value of the optimal policy.

\begin{proof}[Proof of Lemma \ref{thm:reverse-jensen-lemma}]
    Consider the sequence of non-decreasing concave functions
    $r_1, r_2, \ldots$ defined by 
    \[ r_m(n) = \min \{ m,n \} . \]
    Also, let $r_{\infty}(n) = n .$
    The proof of the lemma begins by showing that 
    $\rfunc$ is equal to a non-negative
    weighted sum of functions in the set
    ${\mathscr R} = \{r_1, r_2, \ldots, \} \cup \{r_{\infty}\}$.
    The proof concludes by showing that 
    each of the functions in ${\mathscr R}$ satisfies
    inequality~\eqref{eq:lem-reverse-jensen}.

    Let $w_1 = \rfunc(1)$ and, for each integer $n > 1$, 
    let $w_n = \rfunc(n) - \rfunc(n-1).$ 
    Since $\rfunc$ is concave and non-decreasing, the 
    sequence $w_1, w_2, \ldots$ is non-negative and 
    non-increasing, hence it converges 
    to its greatest lower bound
    \[ v_{\infty} = \inf \{ w_n \, : \, n > 0 \}
    = \lim_{n \to \infty} w_n \geq 0 \]
    For each positive integer $n$ let $v_n = w_n - w_{n+1} \geq 0$.
    We have $w_n = v_{\infty} + \sum_{m \geq n} v_m$
    and
    \begin{equation} \label{eq:rm-decomp}
        \rfunc(n) = \sum_{k=1}^n w_k 
        = \sum_{k=1}^n \left( v_\infty + \sum_{m \geq k} v_m \right) 
        = v_\infty \cdot n + \sum_{m = 1}^{\infty} v_m \cdot 
        \min \{ m, n \} 
        = \sum_{m \in \N \cup \{\infty\}} v_m r_m(n) .
    \end{equation}
    Since the coefficients $v_m \; (1 \leq m \leq \infty)$
    are all non-negative, 
    Equation~\eqref{eq:rm-decomp} shows that 
    $r$ is equal to a non-negative sum of functions
    in ${\mathscr R}$, as claimed. Now, by linearity of
    expectation,
    \begin{align*}
        \E[\rfunc(X)] & = \E \left[ \sum_{m \in \N \cup \{\infty\}} 
        v_m r_m(X) \right] = \sum_{m \in \N \cup \{\infty\}} v_m 
        \E[r_m(X)] \\
        \E[\rfunc(Y)] & = \E \left[ \sum_{m \in \N \cup \{\infty\}} 
        v_m r_m(Y) \right] = \sum_{m \in \N \cup \{\infty\}} v_m 
        \E[r_m(Y)] .
    \end{align*}
    If we can show that 
    \begin{equation} \label{eq:jensen-rm}
      \forall m \in \N \cup \{\infty\} \quad
      \E[r_m(X)] \geq \left( 1 - \tfrac1e \right) \E[r_m(Y)]
    \end{equation}    
    then the lemma follows by taking a weighted sum
    of inequalities of the form~\eqref{eq:jensen-rm},
    with the $m^{\mathrm{th}}$ inequality weighted by $v_m$.
    Let $\mu = \E[X] = \E[Y]$. By Jensen's Inequality, 
    we have $r_m(\mu) \geq \E[r_m(Y)]$, so inequality~\eqref{eq:jensen-rm}
    will follow if we can prove
    \begin{equation} \label{eq:jensen-rm-2}
        \forall m \in \N \cup \{\infty\} \quad
        \E[r_m(X)] \geq \left( 1 - \tfrac1e \right) r_m(\mu) .
    \end{equation}
    When $m=\infty$, the function $r_m$ is the
    identity function, so $\E[r_m(X)] = r_m(\E[X])$,
    establishing the $m=\infty$ case of~\eqref{eq:jensen-rm-2}.
    When $m < \infty$, we can calculate $\E[r_m(X)]$ by setting
    $p = 1 - \frac{1}{\mu}$ and observing that for all $n \geq 1$ 
    we have $\Pr{X \geq n} = p^{n-1}$. Hence,
    \[
        \E[r_m(X)] = \E[\min\{m,X\}]
        = \sum_{n=1}^{\infty} \Pr{\min\{m,X\} \geq n}
        = 1 + p + \cdots + p^{m-1} = \frac{1 - p^m}{1-p} = (1-p^m) \mu .
    \]  
    The function $(1 - p^x) \mu$ is an increasing concave function of $x$,
    so it satisfies 
    \begin{equation} \label{eq:1-px}
        (1 - p^x) \mu \geq \begin{cases}
            (1 - p^{\mu}) x & \mbox{if } 0 \leq x \leq \mu \\
            (1 - p^{\mu}) \mu & \mbox{if } x > \mu .
        \end{cases}
    \end{equation}
    Recalling that $p = 1 - \frac{1}{\mu} < e^{-1/\mu}$, 
    we see that $p^{\mu} < e^{-1}$ so $1 - p^{\mu} > 1 - \frac1e$.
    Now, when $0 \leq m \leq \mu$,
    \[
      \E[r_m(X)] = (1 - p^m) \mu \geq (1 - p^{\mu}) m 
      > \left( 1 - \tfrac1e \right) m = \left( 1 - \tfrac1e \right)
      r_m(\mu) . 
    \]
    When $\mu < m < \infty$, 
    \[
      \E[r_m(X)] = (1 - p^m) \mu \geq (1 - p^{\mu}) \mu 
      > \left( 1 - \tfrac1e \right) \mu 
      = \left( 1 - \tfrac1e \right) r_m(\mu) .
    \]
    We have established inequality~\eqref{eq:jensen-rm-2}
    for all $m \in \N \cup \{\infty\}$, which completes the proof.
\end{proof}

We now proceed to prove \cref{thm:state-independent}.

\begin{proof}[Proof of \cref{thm:state-independent}] 
\newcommand{\wtot}{w^{\mathrm{tot}}}
Now we apply the reverse Jensen inequality to show that 
there is always a static (potentially randomized) bidding
policy whose expected value is at least $1 - \frac1e$ times
the value of the optimal dynamic bidding policy. 

Consider the optimal dynamic policy for time horizon $T$. 
Let $w_{yt}$ denote the probability that 
this policy wins at time $t$ and that the preceding win is at 
time $t-y$. The overall probability of winning at time $t$ will
be denoted by $w_t = \sum_{y=1}^t w_{yt},$ and we will denote the 
policy's expected number of wins by $\wtot = \sum_{t=1}^T w_t$.

Let $Y$ be a random variable supported on the positive integers,
whose distribution is given by 
\[
    \Pr{Y=y} = \frac{\sum_{t=1}^T w_{yt}}{\wtot} .
\]
The expected value of running the optimal policy is 
\begin{equation} \label{eq:expectY}
    \sum_{t=1}^T w_{yt} \rfunc(y) = 
    {\E[\rfunc(Y)]} \cdot \wtot
\end{equation} 

A useful observation about the expected value of $Y$ is 
the following. Let $\tau$ denote the final time that the
policy wins in the time interval $[1,T]$, or $\tau=0$ for
a sample path on which the policy never wins. We have
\begin{align*}
    \E[\tau] &= \sum_{s=1}^{T} \Pr{\tau \geq s} \\
    &= \sum_{s=1}^T \sum_{t=s}^T \Pr{\mbox{first win in } [s,t] \mbox{ occurs at } t} \\
    &= \sum_{s=1}^T \sum_{t=s}^T \sum_{y = t-s+1}^t w_{yt} \\
    &= \sum_{t=1}^T \sum_{y=1}^t \sum_{s = t-y+1}^t w_{yt} \\
    &= \sum_{y=1}^T \sum_{t=y}^T y w_{yt} \\
    &= \E[Y] \cdot \wtot .
\end{align*}

Now consider a static bidding policy defined as follows.
Let $\bar{w} = \wtot/T$ and find the distribution over 
bids $b$ that minimizes $\E[P(b)]$ subject to the 
constraint $\E[W(b)] = \bar{w}.$ If $b'$ is the empirical
bid distribution of the optimal policy then $\E[W(b')] = \bar{w}$,
while $\E[P(b')] \leq \rho$ because the policy obeys the 
budget constraint. Hence, by the definition of $b$, we 
must have $\E[P(b)] \leq \rho$. If one runs the static
policy with bid distribution $b$, the expected number 
of wins\footnote{Actually, this overestimates the
expected number of wins, but only slightly. The
issue is that the static policy obeys the budget
constraint \emph{in expectation}, but it may run out of budget before $T$
and, for this reason, the probability of winning in the final
$o(T)$ rounds will be strictly less than $\bar{w}$.} 
is $\wtot$ and the distribution of the 
spacing between wins is a geometric random variable, $X$.
If $\tau'$ denotes the final time that the static 
bidding policy wins in the interval $[1,T]$ then,
as above, $\E[\tau'] = \E[X] \cdot \wtot$, and
since both $\tau$ and $\tau'$ should be
$T - o(T)$, this implies $\E[X] = (1 \pm o(1)) \cdot \E[Y].$
Applying \Cref{lem:reverse-jensen}, we have
$\E[\rfunc(X)] \geq (1 - \frac1e - o(1)) \E[\rfunc(Y)]$. Finally,
since the expected values of the static policy and the
optimal policy are given, respectively, by 
$\E[\rfunc(X)] \cdot \wtot$ and 
$\E[\rfunc(Y)] \cdot \wtot$, we 
conclude that the static policy obtains at 
least $1 - \frac1e - o(1)$ times the value
of the optimal policy.
\end{proof}

\subsection{State-independent Strategies occur linear regret} \label{ssec:app:state_independent_optimality}

We show an example demonstrating that a state bidding strategy cannot attain a perfect approximation ratio of $1$, and so must have linear regret.
We show that this can happen even in the simpler case without contexts and when the conversion rate is always $1$.
Consider a uniform price distribution $p_t \sim U[0, 1]$ and the concave reward function is $r(\l) = \min\{\l, 2\}$.
The optimal state-independent policy is (similar to what we saw in our warm-up example) to bid $\sqrt{2 \rho}$ every round since that depletes the budget in expectation. 
This leads to winning with probability $\sqrt{2 \rho}$ every round.

Let $L$ be the geometric random variable that denotes the time between wins with conversions when bidding $\sqrt{2 \rho}$ every round.
Specifically, we have with $\Pr{L = \l} = \sqrt{2 \rho} (1 - \sqrt{2 \rho})^{\l-1}$ for every $\l \in \N$.
The expected time-average reward of our state-independent policy is the ratio of the expected reward of a win divided by the expected time between wins (see \cref{def:bench:additional} for a formal statement of this result):
\begin{equation} \label{eq:app:50}
    \frac{ \Ex{r(L)} }{ \Ex{L} }
    =
    \frac{1\cdot\sqrt{ 2 \rho } + 2\cdot(1 - \sqrt{ 2 \rho }) }{ 1 / \sqrt{ 2 \rho } }
    =
    ( 2 - \sqrt{ 2 \rho } ) \sqrt{ 2 \rho }
\end{equation}

Since the reward function has only two possible values, $\rfunc(1)$, and $\rfunc(2)$, for every value of $\rho$ we can find the optimal pair of bids to use for $\l = 1$ and $\l \geq 2$, and programmatically search the value of $\rho$ that maximizes the utility gap between our fixed bidding strategy and the optimal strategy.
This turns out to be $\rho = \sqrt{3} - \frac{3}{2} \approx 0.23$, in which case the optimal state-dependent policy results in $\approx 1.11$ more utility than the utility of the state-independent policy.
We present the details next.

Fix $\rho = \sqrt{3} - \frac{3}{2}$ consider the strategy that bids $b_1 = 2 - \sqrt{3} \approx 0.27$ when $\l = 1$ and $b_2 = 1$ when $\l \ge 2 2$.
Now consider the random variable $L'$ that denotes the time between wins of this strategy.
Specifically, we have $\Pr{L' = 1} = b_1 = 2 - \sqrt{3}$ and $\Pr{L' = 2} = 1 - b_1 = \sqrt{3} - 1$.
First, note that $\Ex{L'} = \sqrt{3}$.

Now, we show that this policy does not run out of budget in expectation (this results in only $o(T)$ error with high probability when always adhering to the budget constraint, which does not affect the approximation ratio bound we want to show).
To calculate the expected time-average spending of this policy, we first calculate the expected payment of a conversion:
\begin{equation*}
    \frac{b_1}{2} \cdot b_1  + \frac{b_2}{2} \cdot (1 - b_1)
    =
    3 - \sqrt{3} \frac{3}{2}
\end{equation*}

The expected time-average payment of this policy is the ratio of the expected payment between conversions and the expected time between conversions (again see \cref{def:bench:additional}), which is $\frac{3 - \sqrt{3} \frac{3}{2}}{\sqrt{3}} = \sqrt{3} - \frac{3}{2} = \rho$.
The expected time-average reward of this policy is (similar to \eqref{eq:app:50}):
\begin{equation*}
    \frac{ \Ex{\rfunc(L')} }{ \Ex{L'} }
    =
    \frac{1 \cdot b_1 + 2 \cdot (1 - b_1)}{\sqrt{3}}
    =
    1
\end{equation*}

This reward is $\frac{1}{3 + 2 \sqrt{ 2\sqrt{3} - 3 } - 2 \sqrt{3}} \approx 1.11$ times higher than the utility of the state-independent strategy, as promised.
\section{Deferred Text and Proofs of Section \ref{sec:bench}} \label{sec:app:bench}

In this section, we complete the proofs of the results in \cref{sec:bench}.

\subsection{Reformulation of (\ref{eq:opt_inf_W}) using stationary distributions}
\label{ssec:app:bench:stationary}

To know the long-term average reward or cost, we need to consider the stationary distribution induced by a bidding policy.
Consider the stationary distribution defined by this vector of winning probabilities $\vec W \in [0, \bar c]^m$.
Specifically, we denote with $\pi_\l(\vec W)$ the probability mass of state $\l \in [m]$ in the stationary distribution defined by $\vec W$.
We now prove the following.

\begin{lemma}
    Let $\optInf_m$ be as defined in Equation \eqref{eq:opt_inf_W} for any $m \in \N$.
    Then it holds that
    \begin{equation} \label{eq:app:opt_inf_stat}
    \begin{aligned}
        \optInf_m = \;
        & \sup_{ \vec W \in [0, \bar c] } &&
        \sum_{\l=1}^m r(\l) W_\l \pi_\l(\vec W)
        \\
        & \textrm{\ \ \  s.t.} &&
        \sum_{\l=1}^m P(W_\l) \pi_\l(\vec W) \le \rho,
        \\
        & \textrm{where} &&
        \pi_1(\vec W) = \sum_{\l = 1}^m W_\l \pi_\l(\vec W)
        \\
        & &&
        \pi_\l(\vec W) = \qty\big( 1-W_{\l-1} ) \pi_{\l-1}(\vec W) \qquad\qquad \forall \l = 2, 3, \ldots, m-1
        \\
        & &&
        \pi_m(\vec W) = \sum_{\l=m-1}^m ( 1-W_\l ) \pi_{\l-1}(\vec W)
        \\
        & &&
        \sum_{\l = 1}^m \pi_\l(\vec W) = 1.
    \end{aligned}
    \end{equation}
\end{lemma}

We note that the equations of optimization problem \eqref{eq:app:opt_inf_stat} uniquely identify the stationary distribution unless $W_m = 0$.
We can avoid this issue by assuming that $W_m$ is at least some infinitesimally small positive constant.
In addition, \cref{lem:bench:incr} shows that there exists an optimal solution where $W_\l$ is non-decreasing in $\l$, in which case $W_m = 0$ implies zero reward.

The proof of the lemma shows that for any stationary policy, the average reward in the objective of optimization problem \ref{eq:app:opt_inf_stat} converges to the stationary distribution of that policy.
Proving the same for the average payment proves the lemma.

\begin{proof}
    Fix some $\vec W$.
    All we have to prove is that for the solution of the equations of $\big\{ \pi_\l(\vec W) \big\}_{\l \in [m]}$ in \eqref{eq:app:opt_inf_stat} it holds that 
    \begin{equation*}
        \lim_{H\to\infty}\frac{1}{H} \sum_{h=1}^H \Ex{ r_m(\l_h) W_{\l_h}}
        =
        \sum_{\l=1}^m \rfunc(\l) W_\l \pi_\l(\vec W)
        \quad \mbox{and} \quad 
        \lim_{H\to\infty}\frac{1}{H} \sum_{h=1}^H \Ex{P(W_{\l_h})}
        =
        \sum_{\l=1}^m P(W_\l) \pi_\l(\vec W)
        ,
    \end{equation*}
    where the expectation is taken over $\l_h$.
    Let $A$ be the (right-stochastic) transition matrix on the state space $[m]$ as defined by $\vec W$.
    Let $\vec f$ and $\vec q$ be the vectors with $m$ elements that have $r(\l) W_\l$ 
    and $P(W_\l)$, respectively, in the $\l$-th entry, and $e_1$ be the unit vector with $1$ in the first coordinate.
    Then we have,
    \begin{align*}
        \lim_{H\to\infty}\frac{1}{H} \sum_{h=1}^H \Ex{r_m(\l_h) W_{\l_h}}
        & =
        \lim_{H\to\infty}\frac{1}{H} \sum_{h=1}^H e_1^\top A^{h-1} \vec f
        =
        e_1^\top \bar A \vec f
        \\
        \lim_{H\to\infty}\frac{1}{H} \sum_{h=1}^H \Ex{P(W_{\l_h})}
        & =
        \lim_{H\to\infty}\frac{1}{H} \sum_{h=1}^H e_1^\top A^{h-1} \vec q
        =
        e_1^\top \bar A \vec q
        ,
    \end{align*}
    where $\bar A = \lim_H\frac{1}{H}\sum_{h \in [H]} A^t$ and the last equality on each line above follows from \cite{DBLP:books/wi/Puterman94} (see the discussion at the start of Section 8.2.1 there) which uses the fact that $\vec f$ and $\vec q$ are bounded and the state space is finite, meaning that $\bar A$ is stochastic.
    This proves the lemma, since $e_1^\top \bar A$ is the vector $\{\pi_\l(\vec W)\}_\l$, as described in \eqref{eq:app:opt_inf_stat} (also see Figure \ref{fig:mc}).
\end{proof}

\subsection{Proof of Lemma \ref{lem:bench:comparison}} \label{ssec:app:bench:comparison}

We first restate the lemma.

\benchcomparison*

\begin{proof}
    Fix the strategy of the optimal algorithm that achieves $\optAlg$; this depends on (i) the current round $t \in [T]$, (ii) how much is the total payment by that round, and (iii) the last win $\l_t$.
    Now consider the infinite-horizon setting, where we follow a strategy inspired by the optimal algorithm in the finite-horizon setting.
    Fix a round $h$ and let $n\in \N\cup\{0\}$ and $t \in [T]$ such that $h = T n + t$.
    In that round $h$, we consider the bid that would have been used in $\optAlg$ if that algorithm was started in round $T n + 1$.
    This means that in round $h = T n + t$, we consider that the algorithm's current round is $t$ and its total payment is the total payment in rounds $T n + 1, T n + 2, \ldots, T n + t-1$.
    What is less obvious to define is the number of rounds since the last winning event.
    If a winning event had happened in some round $T n + t'$ (where $t' \ge 1$), we consider $\l = t - t'$.
    However, if no winning event has happened in rounds $T n + 1, \ldots, T n + t-1$, we consider $\l = t$.
    Note that the actual number of rounds since the last winning event can only be bigger.

    We now calculate the reward and spending levels of the above bidding strategy in the infinite-horizon setting.
    By every round $h = T n + t$, the expected reward is at least $T n \cdot \optAlg$, since in rounds $[T n]$ we have collected at least this much reward in expectation; note here we use the fact that the actual time since the last winning event is at least the one considered by each run of $\optAlg$.
    Given that $T n \ge (h - T)$, the time-averaged expected reward as $h \to \infty$ is at least $\frac{h - T}{h} \optAlg \to \optAlg$.

    Now we verify that we satisfy the budget constraint.
    By round $h = T n + t$, with probability $1$ the total spending is at most $T (n+1) \rho$, since in each interval of $T$ rounds $\optAlg$ spends at most $T \rho$.
    Since $T (n+1) \le h + T$, we have that the average payment as $h \to \infty$ is at most $\frac{h + T}{h} \rho \to \rho$.
    This completes the proof.
\end{proof}

\subsection{Useful Equalities about (\ref{eq:opt_inf_W})} \label{ssec:app:facts}

We now make some very useful definitions.
We will use these throughout the proofs of our results.

\begin{definition} \label{def:bench:additional}
    For any $m \in \N$ and $\vec W \in [0, 1]^m$, we define the following quantities.
    For ease of notation, we define $W_\l = W_m$ for $\l > m$.
    The equalities are proven next.
    \begin{itemize}
        \item $L(\vec W)$ is the expected time between wins with conversions, i.e., the expected return time to state $1$.
        Formally,
        \begin{equation*}
            L(\vec W)
            =
            \sum_{\l = 1}^\infty \l W_\l \prod_{i = 1}^{\l - 1} (1 - W_i)
            =
            \sum_{\l = 1}^\infty \prod_{i = 1}^{\l - 1} (1 - W_i)
        \end{equation*}

        \item $\SR(\vec W)$ is the expected reward of a single win with a conversion, starting from state $1$.
        Formally,
        \begin{align*}
            \SR_m(\vec W)
            & =
            \sum_{\l = 1}^\infty \rfunc_m(\l) W_\l \prod_{i = 1}^{\l - 1} (1 - W_i)
            =
            \sum_{\l = 1}^m \qty( \rfunc_m(\l) - \rfunc_m(\l-1) ) \prod_{i = 1}^{\l - 1} (1 - W_i)
            \\
            & =
            \sum_{\l = 1}^{m-1} \rfunc(\l) W_\l \prod_{i = 1}^{\l - 1} (1 - W_i)
            +
            \rfunc(m) \prod_{i = 1}^{m} (1 - W_i)
        \end{align*}
        When $m$ is clear from the context, we omit the $m$ subscript.
        
        \item $R_m(\vec W)$ is the expected average reward, i.e., the optimization objective of Optimization Problem \eqref{eq:opt_inf_W}.
        Formally,
        \begin{equation*}
            R_m(\vec W)
            =
            \sum_{\l = 1}^m \rfunc(\l) W_\l \pi_\l(\vec W)
            =
            \frac{\SR_m(\vec W)}{L(\vec W)}
        \end{equation*}
        When $m$ is clear from the context, we omit the $m$ subscript.
    
        \item $\SC(\vec W)$ is the expected payment until the first win with conversion, starting from state $1$.
        Formally,
        \begin{equation*}
            \SC(\vec W)
            =
            \sum_{\l = 1}^\infty P(W_\l) \prod_{i = 1}^{\l - 1} (1 - W_i)
            =
            \sum_{\l = 1}^{m-1} P(W_\l) \prod_{i = 1}^{\l - 1} (1 - W_i)
            +
            \frac{P(W_m)}{W_m} \prod_{i = 1}^m (1 - W_i)
        \end{equation*}
    
        \item $C(\vec W)$ as the expected average payment, i.e., the term in the inequality of Optimization Problem \eqref{eq:opt_inf_W}.
        Formally,
        \begin{equation*}
            C(\vec W)
            =
            \sum_{\l = 1}^m P(W_\l) \pi_\l(\vec W)
            =
            \frac{\SC(\vec W)}{L(\vec W)}
        \end{equation*}

        \item $\rea_\l(\vec W)$ as the probability of not getting a win with conversion at least $\l$ times, starting from state $1$.
        Formally, for any $\l \in \N$
        \begin{equation*}
            \rea_\l(\vec W)
            =
            \prod_{i=1}^{\l-1} \qty\big( 1 - W_i )
        \end{equation*}
    \end{itemize}
\end{definition}

For the rest of this subsection we prove the above equalities.
We fix an $m\in\N$ and $\vec W$, and often drop the $(\vec W)$ notation from all the quantities.
Recall that we define $W_\l = W_m$ for $\l > m$.

\begin{proposition}[Return time to $1$] \label{cl:learn:L}
    It holds
    \begin{align*}
        L(\vec W)
        & =
        \sum_{\l = 1}^\infty \l W_{\min\{\l,m\}} \rea_\l(\vec W)
        \\
        & =
        \sum_{\l = 1}^{m-1} \l W_\l \rea_\l(\vec W)
        +
        \qty(
            m - 1 + \frac{1}{W_m}
        ) \rea_m
        \\
        & =
        \sum_{\l = 1}^{m-1} \rea_\l(\vec W)
        +
        \frac{1}{W_m} \rea_m(\vec W)
    \end{align*}
\end{proposition}

\begin{proof}
    The first equality holds by definition of $L$.
    The second equality follows from the first one by noticing that
    \begin{alignat*}{3}
        \Line{
            \sum_{\l=m}^\infty \l W_m \rea_\l
        }{=}{
            W_m \rea_m \sum_{\l=m}^\infty \l \qty\big( 1 - W_m )^{\l - m}
        }{}
        \\
        \Line{}{=}{
            W_m \rea_m \frac{1 - W_m + m W_m}{W_m^2}
        }{}
    \end{alignat*}

    The third equality follows from the second one by noticing that for every $\l \in [2,m]$
    \begin{align*}
        (\l - 1) \rea_\l
        = &
            (\l - 1)  (1- W_{\l-1} )\rea_{\l-1}
           \\
       = & \rea_{\l-1} + (\l - 2) \rea_{\l-1} - (\l - 1) W_{\l-1} \rea_{\l-1}
    \end{align*}
    which, if applied recursively proves that for every $\l \in [2, m]$
    \begin{equation*}
        (\l - 1) \rea_\l
        =
        \sum_{i = 1}^{\l-1} \rea_i - \sum_{i = 1}^{\l-1} i W_i \rea_i
    \end{equation*}
    Applying the above for $\l = m$ to the second equation of Claim~\ref{cl:learn:L} we get the third inequality.
\end{proof}

\begin{proposition}[Stationary distribution] \label{cl:learn:pi}
    It holds
    \begin{align*}
        \l < m: \qquad
        \pi_\l(\vec W)
        &=
        \frac{1}{L} \rea_\l(\vec W)
        \\
        \pi_m(\vec W)
        &= 
        \frac{1}{L(\vec W)} \sum_{i = m}^\infty \rea_i(\vec W)
        =
        \frac{1}{L(\vec W)} \frac{1}{W_m} \rea_m(\vec W)
    \end{align*}
\end{proposition}

\begin{proof}
    The second inequality of $\pi_m$ holds because for $i \ge m$, $\rea_i = \rea_m (1 - W_m)^{i-m}$.

    We need to prove that the above satisfies $m$ out of the $m+1$ equalities of \eqref{eq:app:opt_inf_stat} (since that system is over-defined).
    It holds that $\sum_\l \pi_\l = 1$, by the third equality of Claim~\ref{cl:learn:L}.
    Since $\rea_\l = \rea_{\l-1} (1 - W_{\l-1})$ the above satisfy $\pi_\l = (1 - W_{\l-1}) \pi_{\l-1}$ for $2 \le \l \le m-1$, as needed in \eqref{eq:app:opt_inf_stat}.
    The fact that $\rea_m = \rea_{m-1} (1 - W_{m-1})$ proves that
    \begin{equation*}
        \pi_m
        =
        (1 - W_{m-1}) \pi_{m-1}
        +
        (1 - W_m) \pi_m
        \iff
        W_m \pi_m
        =
        (1 - W_{m-1}) \pi_{m-1}
    \end{equation*}
\end{proof}

\begin{proposition} \label{cl:learn:rew_cos}
    It holds
    \begin{align*}
        R(\vec W)
        &=
        \frac{1}{L} \qty(
            \sum_{\l=1}^{m-1} r(\l) W_\l \rea_\l(\vec W)
            +
            r(m) \, \rea_m(\vec W)
        )
        \\
        &=
        \sum_{\l = 1}^{m-1} \qty\big(r(\l) - r(\l-1)) \pi_\l(\vec W) + \qty\big(r(m) - r(m-1)) W_m \pi_m(\vec W)
        \\
        C(\vec W)
        &=
        \frac{1}{L} \qty(
            \sum_{\l=1}^{m-1} P(W_\l) \rea_\l(\vec W)
            +
            \frac{P(b_m)}{W_m} \rea_m(\vec W)
        )
    \end{align*}
\end{proposition}

\begin{proof}
    The first equality of each quantity follow by the definitions of $R$ and $C$ and \cref{cl:learn:pi}.
    For the second equality of $R$ we have that
    \begin{alignat*}{3}
        \Line{
            R
        }{=}{
            \sum_{\l=1}^m r(\l) W_\l \pi_\l
        }{}
        \\
        \Line{}{=}{
            \sum_{\l=1}^{m-1} r(\l) \qty( \pi_\l - \pi_{\l+1} )
            +
            r(m-1) \qty( \pi_{m-1} - W_m \pi_m )
            +
            r(m) W_m \pi_m
            \hspace{-4pt}
        }{\l \le m-2 : \pi_{\l+1} = (1 - W_\l)\pi_\l\\\pi_m = (1-W_{m-1})\pi_{m-1} + (1-W_{m})\pi_{m}}
        \\
        \Line{}{=}{
            \sum_{\l = 1}^{m-1} \qty\big( r(\l) - r(\l-1) ) \pi_\l
            +
            \qty\big( r(m) - r(m-1) ) W_m \pi_m
        }{r(0) = 0}
    \end{alignat*}
\end{proof}

\begin{corollary}
    The equalities about $\SR$ and $\SC$ in \cref{def:bench:additional} follow by combining their definition with Claims \ref{cl:learn:pi} and \ref{cl:learn:rew_cos}.
\end{corollary}

\subsection{Proof of Lemma \ref{lem:bench:bias}} \label{ssec:app:bench:bias}

We first restate the lemma.

\benchbias*

\begin{proof}
    We fix a $\lambda$ and drop $(\lambda)$ from all notation.
    We will prove the lemma using induction on the state $\l$, starting from state $m$.

    For $\l = m$, the claim is trivial because $h_{m+1}^* = h_m^*$: \eqref{eq:42} becomes $r(m) \ge r(m)$.

    Now fix $ \l \in [m-1]$.
    Using the Bellman optimality condition \eqref{eq:bell} we have,
    \begin{alignat*}{3}
        \Line{
            h^*_{\l+1} + g^*
        }{=}{
            \max_W\qty\Big[
                W r(\l+1) - \lambda P(W) + (1 - W) h_{\l+2}^*
            ]
        }{}
        \\
        \Line{}{\le}{
            \max_W\qty\Big[
                W r(\l+1) - \lambda P(W) + (1 - W) \qty\big( h^*_{\l+1} + r(\l+2) - r(\l+1) )
            ]
        }{\text{ind. hypothesis, \eqref{eq:42}}}
        \\
        \Line{}{\le}{
            \max_W\qty\Big[
                W r(\l+1) - \lambda P(W) + (1 - W) \qty\big( h^*_{\l+1} + r(\l+1) - r(\l) )
            ]
        }{r(\l+2) - r(\l+1) \ge\\ r(\l+1) - r(\l)}
        \\
        \Line{}{=}{
            \max_W\qty\Big[
                W r(\l) - \lambda P(W) + (1 - W) h^*_{\l+1}
            ]
            + r(\l+1) - r(\l)
        }{}
        \\
        \Line{}{=}{
            h^*_{\l} + g^*
            +
            r(\l + 1) - r(\l)
        }{\text{using \eqref{eq:bell}}}
    \end{alignat*}
    which proves the lemma.
\end{proof}

\subsection{Proving that the Optimal Constrained Winning Probabilities are Increasing} \label{ssec:app:increasing}

First, we complete the proof of \cref{lem:bench:incr}.

\benchincr*

\begin{proof}
    We fix a $\lambda$ and drop $(\lambda)$ from all notation.
    Fix $\l \in [m-1]$.
    By the Bellman optimality, we have that
    \begin{align*}
        W_\l^* r(\l) - \lambda P(W_\l^*) + (1 - W_\l^*) h_{\l+1}^*
        \ge
        W_{\l+1}^*r(\l) - \lambda P(W_{\l+1}^*) + (1 - W_{\l+1}^*) h_{\l+1}^*
    \end{align*}
    and that
    \begin{align*}
        W_{\l+1}^* r(\l+1) - \lambda P(W_{\l+1}^*) + (1 - W_{\l+1}^*) h_{\l+2}^*
        \ge
        W_\l^* r(\l+1) - \lambda P(W_\l^*) + (1 - W_\l^*) h_{\l+2}^*
    \end{align*}

    Adding the above two inequalities and rearranging, we get
    \begin{equation*}
        \qty\Big(
            h_{\l+1}^* - h_{\l+2}^* + r(\l+1) - r(\l)
        )
        \qty\Big(
            W_{\l+1}^* - W_\l^*
        )
        \ge 0.
    \end{equation*}
    We proceed to prove that the term in the left parentheses is strictly positive, which does imply $W_{\l+1}^* \ge W_\l^*$.
    Using Lemma~\ref{lem:bench:bias} we get
    \begin{equation*}
        h_{\l+1}^* - h_{\l+2}^* + r(\l+1) - r(\l)
        \ge
        r(\l+1) - r(\l+2) + r(\l+1) - r(\l)
        >
        0,
    \end{equation*}
    where the strict equality holds by strict concavity of $r$: $r(\l+1) - r(\l) > r(\l+2) - r(\l+1)$.
\end{proof}

Given that any optimal solution $\vec W^*(\lambda)$ is non-decreasing for any $\lambda$, we want to prove that the optimal solution of the constrained problem, $\vec W^*$, is also non-decreasing.
We get this as a corollary of \cref{lem:bench:incr}, since the Optimization problem \eqref{eq:opt_inf_W} can be re-written as a Linear Program, making the $\vec W^*$ optimal for the Lagrangian problem of some multiplier $\lambda$.
We re-state this result and prove it for completeness.

\benchincreasing*

The proof will follow from the following lemmas, which we present next.
Recall the notation of \cref{def:bench:additional}: $R(\vec W)$ and $C(\vec W)$ are the expected average-time reward and payment, respectively, under the vector of winning probabilities $\vec W$.
This means that for any $\vec W$, the objective of $\vec W$ for the Lagrange problem with multiplier $\lambda$ is $R(\vec W) - \lambda C(\vec W)$.

We start with a simple and intuitive lemma: the average payment of solution $\vec W^*(\lambda)$ cannot increase as $\lambda$ gets larger.

\begin{lemma}
    The function $C(\vec W^*(\lambda))$ is non-increasing in $\lambda$.
\end{lemma}

\begin{proof}
    Fix $\lambda \ge 0$ and $\e > 0$.
    We will prove that $C(\vec W^*(\lambda)) \le C(\vec W^*(\lambda + \e))$.
    First we use the fact that $\vec W^*(\lambda)$ is optimal for the Lagrangian problem with multiplier $\lambda$.
    Specifically, we use that its objective value is larger than the one of $\vec W^*(\lambda + \e)$:
    \begin{equation*}
        R(\vec W^*(\lambda)) - \lambda C(\vec W^*(\lambda))
        \ge
        R(\vec W^*(\lambda + \e)) - \lambda C(\vec W^*(\lambda + \e))
    \end{equation*}

    Similarly, we use the optimality of $\vec W^*(\lambda + \e)$ when the Lagrange multiplier is $\lambda+\e$:
    \begin{equation*}
        R(\vec W^*(\lambda + \e)) - (\lambda+\e) C(\vec W^*(\lambda + \e))
        \ge
        R(\vec W^*(\lambda)) - (\lambda+\e) C(\vec W^*(\lambda))
    \end{equation*}

    Adding the above two inequalities, we get
    \begin{equation*}
        \e C(\vec W^*(\lambda)) \ge \e C(\vec W^*(\lambda + \e))
    \end{equation*}
    which proves the lemma since $\e > 0$.
\end{proof}

Given the above lemma we can define $\lambda_0$ such that:
\begin{equation*}
    \forall \lambda > \lambda_0 : C\qty(\vec W^*(\lambda)) \le \rho
    \quad \text{ and } \quad
    \forall \lambda < \lambda_0 : C\qty(\vec W^*(\lambda)) \ge \rho
\end{equation*}

Given the above we have the following simple observation.

\begin{lemma} \label{lem:bench:case1}
    If $C(\vec W^*_{\lambda_0}) = \rho$ then $\vec W^*$ is an optimal solution for the Lagrangian problem with $\lambda = \lambda_0$.
\end{lemma}

\begin{proof}
    This is easy to prove since if $R(\vec W^*) < R(\vec W^*_{\lambda_0})$, $\vec W^*$ wouldn't be optimal for the constrained problem.
\end{proof}

\begin{lemma} \label{lem:bench:case2}
    Assume that $C\qty(\vec W^*(\lambda_0)) \ne \rho$, implying that
    \begin{equation*}
        \forall \lambda > \lambda_0 : C\qty(\vec W^*(\lambda)) < \rho
        \quad \text{ and } \quad
        \forall \lambda < \lambda_0 : C\qty(\vec W(\lambda)) > \rho
    \end{equation*}
    Then the $\vec W^*$ is an optimal solution for the Lagrangian problem with $\lambda = \lambda_0$.
\end{lemma}

\begin{proof}
    Fix $\e > 0$ such that $\lambda_0 - \e \ge 0$.
    Consider the vectors $\vec W^*(\lambda_0 - \e)$ and $\vec W^*(\lambda_0 + \e)$.
    We create the vector $\vec W'$ by mixing these two solutions.
    Specifically, let
    \begin{equation*}
        \a = \frac{\rho - C(\vec W^*(\lambda_0 + \e))}{C(\vec W^*(\lambda_0 - \e)) - C(\vec W^*(\lambda_0 + \e))}
    \end{equation*}

    Note that $\a \in (0, 1)$ and is well defined since $C(\vec W^*(\lambda_0 - \e)) > \rho >  C(\vec W^*(\lambda_0 + \e))$.
    We define $W'$ by mixing the occupancy measures of $\vec W^*(\lambda_0-\e)$ and $\vec W^*(\lambda_0+\e)$.
    The occupancy measure of a vector $\vec W$ is a probability distribution over $[m]\times[0, 1]$.
    Specifically, if $q$ is the occupancy measure of $\vec W$ then
    \begin{equation*}
        q(\l, w)
        =
        \pi_\l(\vec W) \delta(w - W_\l)
    \end{equation*}
    where $\delta(x)$ is the Dirac delta function.
    In this case we can write $R(\vec W) = \sum_{\l = 1}^m r(\l) \int_0^1 w q(\l, w) dw$ and $C(\vec W) = \sum_{\l = 1}^m \int_0^1 P(w) q(\l, w) dw$.
    This makes $R(\cdot)$ and $C(\cdot)$ linear in the occupancy measure.
    
    We know define $W'$ by defining its occupancy measure using the occupancy measures of the other two solutions: $\a$ times the one of $\vec W^*(\lambda_0-\e)$ and $1 - \a$ times the one of $\vec W^*(\lambda_0+\e)$ (we note that this could be suboptimal in the sense that a different vector achieves lower payment).
    Since the functions $R$ and $C$ are linear in the occupancy measures, we get that
    \begin{equation*}
        C(\vec W')
        =
        \a C(\vec W^*(\lambda_0-\e))
        +
        (1- \a) C(\vec W^*(\lambda_0+\e))
        =
        \rho
    \end{equation*}

    For the reward $R(\vec W')$ we have
    \begin{alignat*}{3}
        \Line{
            R(\vec W')
        }{=}{
            \a R(\vec W^*(\lambda_0-\e))
            +
            (1- \a) R(\vec W^*(\lambda_0+\e))
        }{}
        \\
        \Line{}{\ge}{
            \a \qty(
                (\lambda_0 - \e) C(\vec W^*(\lambda_0-\e))
                +
                R(\vec W^*(\lambda_0)) - (\lambda_0 - \e) C(\vec W^*_{\lambda_0})
            )
        }{\text{optimality of}\\\vec W^*(\lambda_0-\e)}
        \\
        \Line{}{}{
            \quad +
            (1- \a) \qty(
                (\lambda_0 + \e) C(\vec W^*(\lambda_0+\e))
                +
                R(\vec W^*(\lambda_0)) - (\lambda_0 + \e) C(\vec W^*_{\lambda_0})
            )
        }{\text{optimality of}\\\vec W^*(\lambda_0+\e)}
        \\
        \Line{}{=}{
            R(\vec W^*(\lambda_0))
            + \lambda_0 \qty\big( \rho - C(\vec W^*(\lambda_0)) )
            -O(\e)
        }{}
    \end{alignat*}
    where in the inequality we use the fact that for every $\lambda$, $W^*(\lambda)$ is the optimal solution for the Lagrangian problem with multiplier $\lambda$.
    The above implies that
    \begin{equation*}
        R(\vec W') - \lambda_0 C(\vec W')
        \ge
        R(\vec W^*(\lambda_0)) - \lambda_0 C(\vec W^*(\lambda_0)) 
        - O(\e)
    \end{equation*}

    The above implies that as we take $\e \to 0$, the solution $\vec W'$ becomes optimal for the Lagrangian problem with $\lambda_0$.
    Since $C(\vec W') = C(\vec W^*) = \rho$ ($\vec W'$ is feasible for the constrained problem) we have that $R(\vec W^*) \ge R(\vec W')$.
    This implies that as $\e \to 0$, the vector $\vec W^*$ also becomes optimal for $\lambda_0$.
    This proves what we want.
\end{proof}

Finally, the proof of \cref{cor:bench:increasing} follows easily.

\begin{proof}[Proof of \cref{cor:bench:increasing}]
    The proof follows by \cref{lem:bench:case1,lem:bench:case2} and \cref{lem:bench:incr}.
\end{proof}

\subsection{Proof of Theorem \ref{thm:bench:small_m}} \label{ssec:app:bench:small_m}

We start by proving \cref{lem:bench:win_LB}.

\benchwinLB*

We now present the formal proof of \cref{lem:bench:win_LB}.

\begin{proof}[Proof of \cref{lem:bench:win_LB}]
    We assume the average optimal payment satisfies $C(\vec W^*) = \rho$.
    Otherwise, if $C(\vec W^*) < \rho$, the optimal solution bids $1$ every round (implying $W_\l^* = \bar c$ that implies the lemma) or by increasing the payment, we could get a higher reward.
    We now note that for any $W \in [0, \bar c]$, $P(W) \le \frac{W}{\bar c}$, since a conversion probability of at least $W$ can always be achieved by bidding $1$ with probability $\frac{W}{\bar c}$.
    Using $C(\vec W^*) = \rho$ and $P(W) \le \frac{W}{\bar c}$ we first make the following simple observation:
    \begin{equation} \label{eq:bench:52}
        \rho
        =
        C(\vec W^*)
        =
        \sum_{\l \in [m]} P(W_\l^*) \pi_\l(\vec W^*)
        \le
        \frac{1}{\bar c} \sum_{\l \in [m]} W_\l^* \pi_\l(\vec W^*)
        =
        \frac{1}{\bar c} \pi_1(\vec W^*)
    \end{equation}
    where in the equality we use the condition of $\pi_1(\vec W^*)$ in Optimization Problem \ref{eq:opt_inf_W}.

    Let $\l' = \lceil \frac{2}{\bar c \rho} \rceil$.
    We will prove that $W_{\l'}^* \ge \bar c \rho / 2$, which using \cref{cor:bench:increasing} would prove the desired lemma.
    Towards a contradiction, assume that $W^*_{\l'} < \bar c \rho / 2$ which implies $W^*_{\l} < \bar c \rho / 2$ for $\l \le \l'$.
    We have that
    \begin{alignat*}{3}
        \Line{
            C(\vec W^*)
        }{=}{
            \sum_{\l \in [m]} P(W_\l^*) \pi_\l(\vec W^*)
        }{}
        \\
        \Line{}{\le}{
            \sum_{\l \le \l'} P\qty(\frac{\bar c \rho}{2}) \pi_\l(\vec W^*)
            +
            \sum_{\l > \l'} P(W_\l^*) \pi_\l(\vec W^*)
        }{W_1^* \le W_2^* \le \ldots \le W^*_{\l'} < \frac{\bar c \rho}{2}}
        \\
        \Line{}{\le}{
            \frac{1}{\bar c} \sum_{\l \le \l'} \frac{\bar c \rho}{2} \pi_\l(\vec W^*)
            +
            \frac{1}{\bar c} \sum_{\l > \l'} \pi_\l(\vec W^*)
        }{P(W) \le \frac{W}{\bar c} \le \frac{1}{\bar c}}
        \\
        \Line{}{=}{
            \frac{1}{\bar c} \qty ( 
                1 - \qty( 1 - \frac{\bar c \rho}{2} ) \sum_{\l \le \l'} \pi_\l(\vec W^*)
            )
        }{ \sum_\l \pi_\l(\vec W^*) = 1}
    \end{alignat*}

    We proceed to lower bound $\sum_{\l \le \l'} \pi_\l(\vec W^*)$:
    \begin{alignat*}{3}
        \Line{
            \sum_{\l \le \l'} \pi_\l(\vec W^*)
        }{=}{
            \sum_{\l \le \l'} \pi_1(\vec W^*) \prod_{i = 1}^{\l-1}(1 - W_i^*)
        }{\text{using \eqref{eq:opt_inf_W}}}
        \\
        \Line{}{>}{
            \bar c \rho \sum_{\l \le \l'} \prod_{i = 1}^{\l-1}\qty(1 - \frac{\bar c \rho}{2})
        }{W_i^* < \frac{\bar c \rho}{2} \text{ and \eqref{eq:bench:52}: } \pi_1(\vec W^*) \ge \bar c \rho}
        \\
        \Line{}{=}{
            \bar c \rho \sum_{\l \le \l'} \qty(1 - \frac{\bar c \rho}{2})^{\l-1}
            =
            2 \qty( 1 - \qty(1 - \frac{\bar c\rho}{2})^{\l'} )
        }{}
        \\
        \Line{}{\ge}{
            2 \qty( 1 - \qty(1 - \frac{\bar c\rho}{2})^{\frac{2}{\bar c \rho}} )
            \ge
            2\qty( 1 - \frac{1}{e} )
        }{\l' \ge \frac{2}{\bar c \rho}}
    \end{alignat*}
    Plugging this back in the previous inequality we get
    \begin{equation*}
        \rho
        =
        C(\vec W^*)
        \le
        \frac{1}{\bar c}\qty(
            1 - 2\qty( 1 - \frac{\bar c \rho}{2}) \qty(1 - \frac{1}{e})
        )
        \iff
        e + \bar c \rho \le 2
    \end{equation*}
    where the last inequality is a contradiction.
    This completes the proof of the lemma.
\end{proof}

Now we can proceed to prove \cref{thm:bench:small_m}.

\benchsmallm*

\begin{proof}
    Fix $m = \lceil \frac{2C}{\bar c \rho} \log T \rceil$.
    We prove the theorem for this $m$, which would imply the theorem for larger values as well.
    The proof will revolve around a slight modification of the optimal vector $\vec W^*$.
    Specifically, we consider the strategy $\vec W'$ such that $W_\l' = W_\l^*$ for $\l < m$ and $W_\l' = \bar c$ for $\l \ge m$.
    We have to resolve two problems that $\vec W'$ has.
    First, while its expected time-average reward is greater than the one of $\vec W^*$ when the reward in state $\l$ is $\rfunc_M(\cdot)$ (i.e., $R_M(\vec W') \ge R_M(\vec W^*) = \optInf_M$), we cannot directly claim something about $R_m(\vec W')$, its average reward when the per-round reward function is $\rfunc_m(\cdot)$.
    Second, the expected time-average spending of $\vec W'$ might be larger than $\rho$, since it is as aggressive as possible in states $\l \ge m$.
    We solve each of these problems by proving two claims.

    \begin{proposition} \label{cl:bench:reward_Wprime}
        The expected time-average reward of $\vec W'$ with reward $\rfunc_m(\cdot)$ is $R_m(\vec W') \ge \optInf_M - \frac{4}{\bar c \rho} T^{-C}$.
    \end{proposition}

    The claim above takes care of the first problem.
    The infeasibility due to the payment is solved by the following claim.

    \begin{proposition} \label{cl:bench:cost_Wprime}
        The expected time-average payment of $\vec W'$ is at most $\rho \qty(1 + \frac{16}{\bar c \rho}T^{-C})$.
    \end{proposition}

    These two claims show that $\vec W'$ is only a little suboptimal compared to $\optInf_T$ and only a little infeasible.
    We will combine $\vec W'$ with another solution that has very small spending to get the desired solution.
    We prove both claims after completing the proof of \cref{thm:bench:small_m}.

    Recall that $C(\vec W)$ is the expected average payment of solution $\vec W$ and $\SC(\vec W)$ is the expected payment between wins with conversions.
    Also for any $\vec W$ that has $m$ coordinates we define $W_\l = W_m$ for $\l > m$.
    
    We now define a strategy $\vec W''$ that we will mix with $\vec W'$ to get a solution that is slightly suboptimal but feasible.
    Specifically, we consider $W_\l'' = 0$ for $\l < m$ and $W_\l'' = \bar c$ for $\l \ge m$.
    The expected average payment of $\vec W''$ can be found by calculating
    \begin{equation*}
        L(\vec W'')
        =
        \sum_{\l = 1}^{m-1} 1 + \sum_{\l = m}^\infty (1 - \bar c)^{\l-1}
        =
        m + \frac{(1 - \bar c)^m}{\bar c}
    \end{equation*}
    and
    \begin{equation*}
        \SC(\vec W'')
        =
        \sum_{\l = m}^\infty P(\bar c) (1 - \bar c)^{\l-1}
        \le
        \frac{(1 - \bar c)^{m}}{\bar c}
    \end{equation*}
    making $C(\vec W'') \le \frac{1}{1 + m \bar c (1 - \bar c)^{-m}} \le \frac{1}{1 + m \bar c} \le \frac{\rho}{2 C \log T}$.
    We now define our feasible solution whose reward will be a lower bound for $\optInf_m$.
    This solution's occupancy measure is defined by taking a convex combination of the occupancy measures of $\vec W'$ and $\vec W''$; this makes its expected average reward and payment the same convex combination of the rewards and payments of $\vec W'$ and $\vec W''$.
    Specifically, we consider $1 - \a$ times the occupancy measure of $\vec W'$ and $\a$ times the occupancy measure of $\vec W''$ where
    \begin{align*}
        \a
        =
        \frac{
            32 C T^{-C} \log T
        }{
            2 C \bar c \rho \log T + 32 C T^{-C} \log T - \bar c \rho
        }
        \le
        \order{\frac{1}{\bar c \rho} T^{-C}}
    \end{align*}
    This results in expected average payment:
    \begin{align*}
        (1 - \a)C(\vec W')
        +
        \a C(\vec W'')
        \le
        (1 - \a)\rho \qty(1 + \frac{16}{\bar c \rho}T^{-C})
        +
        \a \frac{\rho}{2 C \log T}
        =
        \rho
    \end{align*}
    making this solution feasible.
    Its expected reward is
    \begin{equation*}
        (1- \a) R_m(\vec W')
        +
        \a R_m(\vec W'')
        \ge
        (1 - \a)\qty( \optInf_M - \frac{4}{\bar c \rho} T^{-C} )
        \ge
        \optInf_M - \order{\frac{1}{\bar c \rho} T^{-C}}
    \end{equation*}
    where in the first inequality we used \cref{cl:bench:reward_Wprime}.
    This completes the theorem's proof.
\end{proof}

We now prove the two claims.
\begin{proof}[Proof of \cref{cl:bench:reward_Wprime}]
    Recall that $R_m(\vec W)$ and $\SR_m(\vec W)$ are the expected time-average reward of a solution $\vec W$ and its expected reward between wins with conversions when the reward on state $\l$ is $\rfunc_m(\l)$.

    Using that $R_m(\vec W) = \frac{\SR_m(\vec W)}{L(\vec W)}$ (recall \cref{def:bench:additional}), we have
    \begin{alignat}{3} \label{eq:bench:reward}
        \Line{
            \SR_M(\vec W^*)
            -
            \SR_m(\vec W')
        }{=}{
            \sum_{\l = 1}^M \qty\big( \rfunc_M(\l) - \rfunc_M(\l-1) ) \prod_{i=1}^{\l-1} ( 1 - W_i^* )
        }{}
        \nonumber\\
        \Line{}{}{
            \;
            -
            \sum_{\l = 1}^m \qty\big( \rfunc_{m}(\l) - \rfunc_{m}(\l-1) ) \prod_{i=1}^{\l-1} ( 1 - W_i^* )
        }{}
        \nonumber\\
        \Line{}{=}{
            \sum_{\l = m}^M \qty\big( \rfunc_M(\l) - \rfunc_M(\l-1) ) \prod_{i=1}^{\l-1} ( 1 - W_i^* )
        }{\l < m \implies\\W_\l^* = W_\l', \rfunc_M(\l) = \rfunc_m(\l)}
        \nonumber\\
        \Line{}{\le}{
            \sum_{\l = m}^\infty \prod_{i=\lceil \frac{2}{\bar c \rho} \rceil}^{\l-1} ( 1 - W_i^* )
        }{\rfunc_M(\l) - \rfunc_M(\l-1)\le 1\\1 - W_i^* \le 1}
        \nonumber\\
        \Line{}{\le}{
            \sum_{\l = m}^\infty \prod_{i=\lceil \frac{2}{\bar c \rho} \rceil}^{\l-1} \qty( 1 - \frac{\bar c \rho}{2} )
        }{\text{\cref{lem:bench:win_LB}}}
        \nonumber\\
        \Line{}{\le}{
            \sum_{\l = m}^\infty \qty( 1 - \frac{\bar c \rho}{2} )^{\l + \frac{2}{\bar c \rho} - 1}
            =
            \frac{1}{\frac{\bar c \rho}{2}} \qty( 1 - \frac{\bar c \rho}{2} )^{m + \frac{2}{\bar c \rho} - 1}
        }{}
        \nonumber\\
        \Line{}{\le}{
            \frac{2}{\bar c \rho} \qty( 1 - \frac{\bar c \rho}{2} )^{\frac{2 C}{\bar c \rho} \log T - 1}
        }{m \ge \frac{2 C}{\bar c \rho} \log T}
        \nonumber\\
        \Line{}{\le}{
            \frac{4}{\bar c \rho} \qty( 1 - \frac{\bar c \rho}{2} )^{\frac{2 C}{\bar c \rho} \log T}
        }{\bar c \rho \le 1}
        \nonumber\\
        \Line{}{\le}{
            \frac{4}{\bar c \rho} \frac{1}{\exp(C \log T)}
            =
            \frac{4}{\bar c \rho} T^{-C}
        }{}
    \end{alignat}
    It is not hard to observe that $L(\vec W') \le L(\vec W^*)$, since $W_\l^* \le W_\l'$ for all $\l$.
    Using the above, we have
    \begin{align*}
        R_{m}(\vec W')
        =
        \frac{\SR_m(\vec W')}{L(\vec W')}
        \ge
        \frac{\SR_M(\vec W^*) - \frac{4}{\bar c \rho} T^{-C}}{L(\vec W^*)}
        \ge
        R_M(\vec W^*) - \frac{4}{\bar c \rho} T^{-C}
    \end{align*}
    where the last inequality holds because $L(\vec W^*) \ge 1$.
    This proves the claim.
\end{proof}

We now prove \cref{cl:bench:cost_Wprime}.

\begin{proof}[Proof of \cref{cl:bench:cost_Wprime}]
    Using $C(\vec W) = \frac{\SC(\vec W)}{L(\vec W)}$, we upper bound $C(\vec W')$ using the fact that it is close to $C(\vec W^*)$.
    We assume that $C(\vec W^*) = \rho$ since otherwise, because of optimality of $\vec W^*$, it must hold that $\vec W^*$ bids $1$ every round, implying $W_\l^* = \bar c$ for all $\l$; this would imply $\vec W' = \vec W^*$ making the claim trivial.
    First, we compare $L(\vec W^*)$ and $L(\vec W')$, using the equations of \cref{def:bench:additional}
    \begin{alignat}{3} \label{eq:length}
        \Line{
            L(\vec W^*)
            -
            L(\vec W')
        }{=}{
            \sum_{\l = m}^\infty \qty(
                \prod_{i=1}^{\l-1} ( 1 - W_i^* )
                -
                \prod_{i=1}^{\l-1} ( 1 - W_i' )
            )
        }{\l < m \implies\\W_\l^* = W_\l'}
        \nonumber\\
        \Line{}{\le}{
            \sum_{\l = m}^\infty \prod_{i=\lceil \frac{2}{\bar c \rho} \rceil}^{\l-1} ( 1 - W_i^* )
        }{}
        \nonumber\\
        \Line{}{\le}{
            \frac{4}{\bar c \rho} T^{-C}
            \le
            \frac{4}{\bar c \rho} T^{-C} L(\vec W')
        }{}
    \end{alignat}
    where in the second to last inequality we used the same calculation as in \eqref{eq:bench:reward} and in the last inequality we used the fact that $L(\vec W') \ge 1$.

    Now, using \cref{def:bench:additional}, for the average payment per epoch of $\vec W'$ we have
    \begin{alignat}{3}\label{eq:cost}
        \Line{
            \SC(\vec W')
            -
            \SC(\vec W^*)
        }{\le}{
            \sum_{\l = m}^\infty P(W_\l') \prod_{i = 1}^{\l - 1} ( 1 - W_i' )
        }{\l < m \implies\\W_\l^* = W_\l'}
        \nonumber\\
        \Line{}{\le}{
            P(\bar c) \sum_{\l = m}^\infty \prod_{i = \lceil \frac{2}{\bar c \rho} \rceil}^{\l - 1} ( 1 - W_i' )
        }{}
        \nonumber\\
        \Line{}{\le}{
            P(\bar c) \sum_{\l = m}^\infty \qty( 1 - \frac{\bar c \rho}{2} )^{\l + \frac{2}{\bar c \rho} - 1}
        }{\text{\cref{lem:bench:win_LB}}}
        \nonumber\\
        \Line{}{\le}{
            P(\bar c) \frac{4}{\bar c \rho} T^{-C}
            \le
            \frac{4}{\rho} T^{-C} \SC(\vec W^*)
        }{}
    \end{alignat}
    where in the second to last inequality, we used the same analysis as in \eqref{eq:bench:reward} and in the final inequality we used that for any vector $\vec W$ it holds that
    \begin{equation*}
        \SC(\vec W)
        =
        \sum_{\l = 1}^\infty P(W_\l) \prod_{i = 1}^{\l - 1} ( 1 - W_i )
        \le
        \frac{P(\bar c)}{\bar c} \sum_{\l = 1}^\infty W_\l \prod_{i = 1}^{\l - 1} ( 1 - W_i )
        =
        \frac{P(\bar c)}{\bar c}
    \end{equation*}
    where the inequality holds because $P(W) \le \frac{W}{\bar c} P(\bar c)$: the bidder can win with probability $W$ by bidding $1$ with probability $\frac{W}{\bar c}$ and $0$ otherwise and bidding $1$ results in paying $P(\bar c)$ in expectation, i.e. the expected price.

    Combining \eqref{eq:length} and \eqref{eq:cost} we get
    \begin{equation*}
        C(\vec W')
        =
        \frac{\SC(\vec W')}{L(\vec W')}
        \le
        \frac{\SC(\vec W^*)}{L(\vec W^*)} \frac{1 + \frac{4}{\rho} T^{-C}}{\frac{1}{1 + \frac{4}{\bar c\rho} T^{-C}}}
        \le
        C(\vec W^*) \qty(1 + \frac{4}{\bar c \rho} T^{-C})^2
        \le
        \rho \qty(1 + \frac{16}{\bar c \rho}T^{-C})
    \end{equation*}
    where in the last inequality we used $C(\vec W^*) \le \rho$ and $\frac{2}{\bar c \rho} T^{-C} \le 1$.
    This proves the claim.
\end{proof}
\section{Deferred Proofs of Section \ref{sec:online}}\label{sec:app:online}

In this section, we complete the proofs of the results in \cref{sec:online}.

\subsection{Lower Bound on the Realized Reward}\label{ssec:app:online:totalRew}

Before proving the main lemma of this part, \cref{lem:online:total_rew}, we prove a simple lemma on the expected reward of a single epoch.
Fix an epoch $i$.
We want to compare the reward of epoch $i$ with $\SR(\vecestbb{i})$, the expected reward between wins with conversions of using $\vecestbb{i}$ in the infinite horizon setting (we overload the notation of \cref{def:bench:additional}).
These would be the same in expectation if we let $k = \infty$.
However, because epoch $i$ might end early, there is some reward loss.
This loss depends on the probability of an epoch ending early.
This probability becomes small because of bidding $1$ in state $m$, $\estbb{i}_m = 1$, and by setting $k - m \approx \log T$.

\begin{lemma} \label{lem:online:singe_rew}
    Let $\calH_t$ be the history up to round $t$ (prices, contexts, conversions).
    Fix an epoch $i$ that starts at least $k$ rounds before $T$: $T_{i-1} \le T - k$.
    Assume $k \ge m + \frac{1}{\bar c} C\log T$ for some $C \ge 1$.
    Let $\rew_i$ be the reward of epoch $i$ if the budget is not violated.
    Then,
    \begin{equation*}
        \ExC{\rew_i}{\calH_{T_{i-1}}}
        \ge
        \SR(\vecestbb{i})
        -
        m T^{-C}
    \end{equation*}
\end{lemma}

Since $\estbb{i}_m = 1$ we know that after not getting a win with impression $m$ times, the probability that an epoch ends is $\bar c$ per round.
By the lower bound on $k$, the probability of not winning with an impression is less than $T^{-C}$, which would lead to losing at most $m$ reward.

\begin{proof}
    We first lower bound the probability of epoch $i$ ending without a conversion.
    We use the fact that, by bidding $\bb_m^i = 1$ when it has been $m$ rounds or more since the last conversation, it holds that $W(1) = \bar c$:
    \begin{alignat*}{3}
        \prod_{\l = 1}^k \qty( 1 - W(\bb_\l^i) )
        \le
        \prod_{\l = m}^k \qty(1 - W(\bb_\l^i))
        =
        \prod_{\l = m}^k \qty(1 - \bar c)
        =
        \qty(1 - \bar c)^{k - m + 1}
        \le
        \qty(1 - \bar c)^{\frac{1}{\bar c} C \log T}
        \le
        T^{-C}
    \end{alignat*}
    where in the second to last inequality we used $k \ge m + \frac{1}{\bar c} C\log T$.
    This proves that
    \begin{equation*}
        \ExC{\rew_i}{\calH_{T_{i-1}}}
        \ge
        \ExC{ \SR(\vecestbb{i}) }{\calH_{T_{i-1}}}  \qty( 1 - T^{-C} )
    \end{equation*}
    
    The above proves the lemma since $\SR(\vecestbb{i}) \le m$ and $\SR(\vecestbb{i})$ does not depend on the randomness of rounds after $T_{i-1}$.
\end{proof}

We now prove the main lemma of this part, which we restate here for completeness.

\LemOnlineTotalRew*

\begin{proof}[Proof of \cref{lem:online:total_rew}]
    Let $w_t \in \{0, 1\}$ indicate whether the bidder won a conversion in round $t$.
    We first notice that 
    \begin{equation} \label{eq:online:1}
        \sum_{t = 1}^\tau w_t \rfunc(\tilde \l_t)
        \ge
        \sum_{j = 1}^{I_\tau} \rew_j
        -
        m
    \end{equation}
    where the first inequality holds because by round $\tau$, the algorithm has been allocated the rewards of the first $I_\tau-1$ epochs but may not have been allocated the reward of epoch $I_\tau$, which is at most $\rfunc( m ) \le m$.

    We shoe next that $\sum_{j = 1}^{I_\tau} \rew_j \approx \sum_{j = 1}^{I_\tau} \ExC{\rew_j}{\calH_{T_{j-1}}}$.
    For $i = 0, 1, \ldots,$ we define the sequence
    \begin{equation*}
        Z_i
        =
        \sum_{j = 1}^i \rew_j - \sum_{j = 1}^i \ExC{\rew_j}{\calH_{T_{j-1}}}
    \end{equation*}
    We notice that $Z_i$ is a martingale (since $\ExC{Z_i - Z_{i-1} }{ \calH_{T_{i-1}} } = 0$) that has differences $|Z_i - Z_{i-1}| \le m$.
    Using Azuma's inequality, we get that for all $\d > 0$ with probability at least $1-\d$
    \begin{align*}
        \sum_{j = 1}^i \rew_j
        \ge
        \sum_{j = 1}^i \ExC{\rew_j}{\calH_{T_{j-1}}}
        -
        \order{ m \sqrt{i \log\frac{1}{\d}} }
    \end{align*}

    Using the union bound over all epochs $i$ (note there are at most $T$ epochs), we get the above inequality for all $i$, implying that all $\d > 0$ with probability at least $1-\d$
    \begin{align*}
        \sum_{j = 1}^{I_\tau} \rew_j
        \ge
        \sum_{j = 1}^{I_\tau} \ExC{\rew_j}{\calH_{T_{j-1}}}
        -
        \order{ m \sqrt{T \log\frac{T}{\d}} }
    \end{align*}

    We now use \cref{lem:online:singe_rew} and get
    \begin{align} \label{eq:online:3}
        \sum_{j = 1}^{I_\tau} \rew_j
        \ge
        \sum_{j = 1}^{I_\tau} \SR(\vecestbb{j})
        -
        \order{ m \sqrt{T \log\frac{T}{\d}} }
    \end{align}
    where we include the $\sum_j m T^{-1} \le m$ term in the $\order{\cdot}$ term.
    We proceed to analyze the expected reward between wins with conversions $\SR(\vecestbb{j}) = R(\vecestbb{j}) L(\vecestbb{j})$ (recall \cref{def:bench:additional}, by overloading the notation: $R(\vec\bb)$ is the time average reward of bidding $\vec\bb$ and $L(\vec\bb)$ is the expected time between conversions):
    \begin{alignat}{3} \label{eq:online:4}
        \SR(\vecestbb{j})
        & =
        R(\vecestbb{j}) L(\vecestbb{j})
        \ge
        \qty( \optInf_m - \e_R(T_{j-1}) )^+
        L(\vecestbb{j})
        \nonumber\\
        & \ge
        \qty( \optInf_m - \e_R(T_{j-1}) )^+
        \ExC{ L_j }{\calH_{T_{j-i}}}
    \end{alignat}
    where the first inequality holds by the sub-optimality assumption on $\vecestbb{i}$ and the second holds because epoch $j$ may be stopped early, so its expected time until a conversion can only be bigger than the expected length of the epoch.
    Relying on the fact that $\e_R(T_{j-1})$ is a deterministic function of $\calH_{T_{j-i}}$, we now define another martingale:
    \begin{equation*}
        Y_i
        =
        \sum_{j=1}^i \qty( \optInf_m - \e_R(T_{j-1}) )^+ \ExC{ L_j }{\calH_{T_{j-i}}}
        -
        \sum_{j=1}^i \qty( \optInf_m - \e_R(T_{j-1}) )^+ L_j
    \end{equation*}
    which has differences bounded by $k$ since $L_j \le k$ and $\optInf_m - \e_R(T_{j-1}) \le 1$.
    Using Azuma's inequality, we get that with probability at least $1-\d$ for all $\d > 0$ it holds
    \begin{alignat*}{3}
        \Line{
            \sum_{j=1}^i \qty( \optInf_m - \e_R(T_{j-1}) )^+ \! \ExC{ L_j }{\calH_{T_{j-i}}}
        }{\ge}{
            \sum_{j=1}^i \qty( \optInf_m - \e_R(T_{j-1}) )^+ \! L_j
            -
            \order{k \sqrt{ i \log \frac{1}{\d}}} 
        }{}
        \\
        \Line{}{\ge}{
            \optInf_m (T_i - k)
            -
            \sum_{j=1}^i L_j \e_R(T_{j-1})
            -
            \order{k \sqrt{ i \log \frac{1}{\d}}}\!
        }{}
    \end{alignat*}
    where the last inequality holds from $\sum_{j=1}^i L_j = T_i-k$.
    Using the union bound on the above for all epochs $i$ (there are at most $T$) and plugging it into \eqref{eq:online:4} and then \eqref{eq:online:3} we get
    \begin{equation*}
        \sum_{j = 1}^{I_\tau} \rew_j
        \ge
        \optInf_m (T_{I_\tau} - k)
        -
        \sum_{j=1}^{I_\tau} L_j \e_R(T_{j-1})
        -
        \order{ (m + k) \sqrt{T \log\frac{T}{\d}} }
    \end{equation*}

    The proof is finalized by plugging the above into \eqref{eq:online:1} and using the following inequalities: $k \optInf_m \le k$, $T_{I_\tau} \ge \tau$, $m + k \le 2k$.
\end{proof}

\subsection{Upper Bound on the Realized Payment}\label{ssec:app:online:totalPay}

Before proving the main lemma of this part, \cref{lem:online:total_pay}, we prove a simple lemma on the expected length of a single epoch.
This lemma proves that the expected length of an epoch $i$ is not too much smaller than the expected time between wins with conversions in the infinite horizon setting, $L(\vecestbb{i})$.
While this looks simple, we have to look into the exact difference between the two quantities.
The length of an epoch is bounded almost surely but the time between wins with conversions is not.

\begin{lemma} \label{lem:online:singe_pay}
    Fix an epoch $i$ and assume that $k \ge m + \frac{1}{\bar c} C \log T$ and $C \ge 1$.
    Then, it holds that
    \begin{equation*}
        L(\vecestbb{i})
        \le
        \ExC{L_i}{\calH_{T_{i-1}}}
        +
        \frac{1}{\bar c} T^{-C}
    \end{equation*}
\end{lemma}

\begin{proof}
    The only difference between $L_i$ and $L(\vecestbb{i})$ is that in $L_i$ we might stop early if there is no conversion after $k$ rounds.
    This means that the difference of their expectations is
    \begin{alignat*}{3}
        \Line{
            L(\vecestbb{j}) - \ExC{L_j}{\calH_{T_{j-1}}}
        }{=}{
            \sum_{\l = 1}^\infty \l W(\bb^j_\l) \prod_{\l' = 1}^{\l - 1} \qty( 1 - W(\bb^j_{\l'}) )
        }{}
        \\
        \Line{}{}{
            \quad
            -
            \sum_{\l = 1}^\infty \min\{\l, k\} W(\bb^j_\l) \prod_{\l' = 1}^{\l - 1} \qty( 1 - W(\bb^j_{\l'}) )
        }{}
        \\
        \Line{}{=}{
            \sum_{\l = k+1}^\infty (\l - k) W(\bb^j_\l) \prod_{\l' = 1}^{\l - 1} \qty( 1 - W(\bb^j_{\l'}) )
        }{}
        \\
        \Line{}{\le}{
            \sum_{\l = k+1}^\infty (\l - k) \bar c \prod_{\l' = m}^{\l - 1} \qty( 1 - \bar c )
        }{\bb_m^i = 1\\k \ge m}
        \\
        \Line{}{=}{
            \sum_{\l = k+1}^\infty (\l - k) \bar c \qty( 1 - \bar c )^{\l - m}
            =
            \frac{(1 - \bar c)^{k - m + 1}}{\bar c}
            \le
            \frac{1}{\bar c} T^{-C}
        }{}
    \end{alignat*}
    where in the last inequality we used $k \ge m + \frac{1}{\bar c} C \log T$.
\end{proof}

We now restate \cref{lem:online:total_pay} for completeness and then prove it.

\LemOnlineTotalPay*

\begin{proof}
    Let $\pay_i$ denote the total payment during epoch $i$.
    It is not hard to see that the total payment up to round $\tau$ is at most $\sum_{i=1}^{I_\tau} \pay_i$ since epoch $I_\tau$ ends at round $\tau$ or later.

    For $i=0,1,\ldots$ we define the following martingale with respect to the history $\calH_{T_{j-1}}$ up to each epoch $j-1$.
    \begin{equation*}
        Z_i
        =
        \sum_{j = 1}^i \pay_j
        -
        \sum_{j = 1}^i \ExC{\pay_j}{\calH_{T_{j-1}}}
    \end{equation*}
    which has bounded differences $\abs{Z_i - Z_{i-1}} \le k$, since the maximum payment at any epoch can be at most $k$.
    Using Azuma's inequality we get that for all $\d > 0$ with probability at least $1 - \d$
    \begin{equation*}
        \sum_{j = 1}^i \pay_j
        \le
        \sum_{j = 1}^i \ExC{\pay_j}{\calH_{T_{j-1}}}
        +
        \order{k \sqrt{i \log\frac{1}{\d}}}
    \end{equation*}

    We now notice that $\ExC{\pay_j}{\calH_{T_{j-1}}} \le \SC(\vecestbb{j})$: the expected payment of epoch $j$ is less than the expected payment between wins with conversions in the infinite horizon setting, since an epoch might stop earlier.
    Using the union bound over all epochs $i$ (there are at most $T$ of them) we get the above inequality for all $i$ and therefore also $i = I_\tau$: with probability at least $1-\d$ it holds
    \begin{equation} \label{eq:online:11}
        \sum_{j = 1}^{I_\tau} \pay_j
        \le
        \sum_{j = 1}^{I_\tau} \SC(\vecestbb{j})
        +
        \order{k \sqrt{T \log\frac{T}{\d}}}
    \end{equation}
    
    We now examine $\SC(\vecestbb{j})$.
    Using \cref{def:bench:additional},
    \begin{equation*}
        \SC(\vecestbb{j})
        =
        C(\vecestbb{j}) L(\vecestbb{j})
        \le
        \qty( \rho + \e_C(T_{j-i}) ) L(\vecestbb{j})
        \le
        \qty( \rho + \e_C(T_{j-i}) ) \ExC{L_i}{\calH_{T_{i-1}}}
        +
        \frac{2}{\bar c} T^{-1}
    \end{equation*}
    where in the first inequality we used that the payment of $\vecestbb{i}$ is $\e_C(T_{j-i})$ approximate and that it is feasible for the empirical distribution; for the second equality we used \cref{lem:online:singe_pay} and that $\rho + \e_C(T_{j-i}) \le 2$.
    Plugging into \eqref{eq:online:11} we get
    \begin{alignat}{3} \label{eq:online:13}
        \Line{
            \sum_{j = 1}^{I_\tau} \pay_j
        }{\le}{
            \sum_{j = 1}^{I_\tau} \qty( \rho + \e_C(T_{j-i}) ) \ExC{L_j}{\calH_{T_{j-1}}}
            +
            \order{k \sqrt{T \log\frac{T}{\d}}}
            +
            \sum_{j = 1}^{I_\tau}\frac{2}{\bar c} T^{-1}
        }{}
        \nonumber\\
        \Line{}{\le}{
            \sum_{j = 1}^{I_\tau} \qty( \rho + \e_C(T_{j-i}) ) \ExC{L_j}{\calH_{T_{j-1}}}
            +
            \order{k \sqrt{T \log\frac{T}{\d}}}
        }{}
    \end{alignat}
    where in the last inequality we move the last sum term into the $\order{\cdot}$ term by using $I_\tau \le T$, and that $k \ge \frac{1}{\bar c}$.

    The rest of the proof is similar to the one of \cref{lem:online:total_rew}: we define a martingale
    \begin{equation*}
        Y_i
        =
        \sum_{j=1}^i \qty( \rho + \e_C(T_{j-i}) )\ExC{L_j}{\calH_{T_{j-1}}}
        -
        \sum_{j=1}^i \qty( \rho + \e_C(T_{j-i}) ) L_j
    \end{equation*}
    which has bounded differences $\abs{Y_i - Y_{i-1}} \le 2 k$ since $1 \le L_j \le k$ and $\rho + \e_C(T_{j-i}) \le 2$.
    Using Azuma's inequality we get that with probability at least $1 - \d$
    \begin{alignat*}{3}
        \Line{
            \sum_{j=1}^i \qty( \rho + \e_C(T_{j-i}) )\ExC{L_j}{\calH_{T_{j-1}}}
        }{\le}{
            \sum_{j=1}^i \qty( \rho + \e_C(T_{j-i}) ) L_j
            +
            \order{k \sqrt{T \log\frac{1}{\d}}}
        }{}
        \\
        \Line{}{=}{
            (T_i - k) \rho
            +
            \sum_{j=1}^i L_j \e_C(T_{j-i})
            +
            \order{k \sqrt{T \log\frac{1}{\d}}}
        }{\sum_{j=1}^i L_j = T_i-k}
    \end{alignat*}

    Getting the above inequality for $i = I_\tau$ (by using the union bound) and combining it with \cref{eq:online:13} we get
    \begin{equation*}
        \sum_{j = 1}^{I_\tau} \pay_j
        \le
        (T_{I_\tau} - k) \rho
        +
        \sum_{j=1}^{I_\tau} L_j \e_C(T_{j-i})
        +
        \order{k \sqrt{T \log\frac{T}{\d}}}
    \end{equation*}

    By using $\tau \ge T_{I_\tau} - k$, we get the lemma.
\end{proof}

\subsection{Deferred Proofs of Theorem \ref{thm:online:regret} and Corollary \ref{cor:online:final_res}} \label{ssec:app:online:final}

Using \cref{lem:online:total_rew,lem:online:total_pay} it is not hard to prove \cref{thm:online:regret}.
By picking $\tau = T - \orderT*{\sqrt T}$ so that the algorithm does not run out of budget in round $\tau$ with high probability, we get the promised high reward of \cref{lem:online:total_rew}.

\ThmOnlineRegret*

\begin{proof}[Proof of \cref{thm:online:regret}]
    Fix a $\d > 0$.
    Using the union bound, assume that \cref{lem:online:total_rew,lem:online:total_pay} hold for all $\tau \in [T]$ with probability at least $1 - \d$.
    Fix a round $\tau$ such that
    \begin{equation} \label{eq:online:21}
        T - \tau
        \ge
        \sum_{j=1}^{I_\tau} L_j \e_C(T_{j-1})
        +
        \order{k \sqrt{T \log\frac{T}{\d}}}
    \end{equation}

    \cref{lem:online:total_pay} implies that the algorithm has run out of budget by that round $\tau$.
    This means that $\sum_{t=1}^\tau w_t \rfunc(\tilde \l_t)$ is a lower bound for the algorithm's reward by that round (lower bound because the actual lengths between wins with conversions can only be bigger than $\tilde \l_t$).
    \cref{lem:online:total_rew} now implies that the total reward by round $\tau$ is at least
    \begin{align*}
        & \tau \optInf_m
        -
        \sum_{j=1}^{I_\tau} L_j \e_R(T_{j-1})
        -
        \order{ (m + k) \sqrt{T \log\frac{T}{\d}} }
        \\
        \ge\; &
        T \optInf_m
        -
        \sum_{j=1}^{I_\tau} L_j \qty\big( \e_R(T_{j-1}) + \e_C(T_{j-1}))
        -
        \order{ (m + k) \sqrt{T \log\frac{T}{\d}} }
    \end{align*}
    where in the inequality we used \eqref{eq:online:21} and $\optInf_m \le 1$.
    This proves the theorem
\end{proof}

Finally, we prove \cref{cor:online:final_res}.
The proof is quite simple using \cref{thm:calc}, but we include it for completeness.

\CorOnlineFinal*

\begin{proof}[Proof of \cref{cor:online:final_res}]
    Using \cref{thm:calc} and the union bound, we get that with probability at least $1 - \d$ for all rounds $t$ it holds that
    \begin{equation*}
        \e_R(t) + \e_C(t)
        \le
        \order{
            \frac{m^3}{\rho} \frac{1}{\sqrt{t}} \sqrt{ \log\frac{T}{\d} }
        }
    \end{equation*}
    The above implies that for any round $\tau$
    \begin{align*}
        \sum_{j = 1}^{I_\tau} L_j \qty\big( \e_R(T_{j-1}) + \e_C(T_{j-1}) )
        \le
        \order{
            \frac{m^3}{\rho} \sqrt{ \log\frac{T}{\d} }
            \sum_{j = 1}^{I_\tau} \frac{L_j}{\sqrt{T_{j-1}}}
        }
    \end{align*}

    The corollary follows by proving that $\sum_{j = 1}^{I_\tau} \frac{L_j}{\sqrt{T_{j-1}}} = \order*{\sqrt T}$, which makes the above term dominate the other error term in \cref{thm:online:regret}.
    \begin{alignat*}{3}
        \Line{
            \sum_{j = 1}^{I_\tau} \frac{L_j}{\sqrt{T_{j-1}}}
        }{\le}{
            \sum_{j = 1}^{I_\tau} \sum_{t = T_{j-1} - L_j + 1}^{T_{j-1}} \frac{1}{\sqrt{t}}
            \le
            2 \sum_{j = 1}^{I_\tau} \qty( \sqrt{T_{j-1} - 1} - \sqrt{T_{j-1} - L_j} )
        }{\sum_{t=a}^b\frac{1}{\sqrt{t}} \le 2\qty( \sqrt{b-1} - \sqrt{a-1} )}
        \\
        \Line{}{\le}{
            2 \sum_{j = 1}^{I_\tau} \qty( \sqrt{T_j - k - 1} - \sqrt{T_j - k - L_j} )
        }{T_{j-1} \ge T_j - k}
        \\
        \Line{}{=}{
            2 \sum_{j = 1}^{I_\tau} \qty( \sqrt{T_j - k} - \sqrt{T_{j-1} - k} )
            =
            2 \sqrt{T_{I_\tau}}
            \le
            2 \sqrt{T}
        }{T_j = T_{j-1} + L_j}
    \end{alignat*}
\end{proof}
\section{Deferred Proofs of Section \ref{sec:calc}} \label{sec:app:calc}

In this section, we present the deferred proofs of \cref{sec:calc}.

\subsection{Deferred Proof of Lemma \ref{lem:calc:simple_single}} \label{ssec:app:calc:single}

We first restate the lemma.

\LemCalcSingle*

\begin{proof}
    We first write the Lagrangian of \eqref{eq:calc:single} for some Lagrange multiplier $\mu \ge 0$.
    \begin{align*}
        \sup_{\bb}
        \quad
        \Ex[x,p]{\Big.\qty\big( c - \mu p ) \One{b(x) \ge p}} + \mu \rho'
    \end{align*}
    where recall the conversion rate $c$ is part of the context $x$.

    The above is maximized when $b(x) = \min(\frac{c}{\mu}, 1)$ independent of other parts of the context, since the objective when $b(x) = \frac{c}{\mu}$ becomes $\Ex{\big( c - \mu p )^+} + \mu \rho'$, which is the supremum value of the above optimization problem.
    When $\frac{c}{\mu} > 1$, then winning the auction for sure by bidding $1$ maximizes the value.
    Because the only part of the context that matters is now the conversion rate, we used $b(c)$ instead of $b(x)$ and focus on the distribution over the pair $(p, c)$ instead of $(p, x)$.

    We overload the previous notation of the expected payment of a bid $b$, $P(b)$ and define the expected payment of using bid $c / \mu$ as
    \begin{equation*}
        P(\mu)
        =
        \Ex[x, p, c]{\big. p \One{c \ge \mu p} }
        .
    \end{equation*}

    The above function is non-decreasing in $\mu$ and lower semi-continuous.
    Let $\mu^*$ be the greatest non-negative number such that $P(\mu^*) \ge \rho'$.
    If no such $\mu^*$ exists, always bidding $1$ (i.e., using $\mu = 0$) is the optimal solution to \eqref{eq:calc:single}.
    If $ P(\mu^*) = \rho'$ then $b^{\mu^*}(c) = \frac{c}{\mu^*}$ is an optimal solution for \eqref{eq:calc:single}.
    We note that if the distribution on $(p, c)$ has no atoms, then $P(\mu)$ is continuous, implying that a $\mu^*$ with $P(\mu^*) = \rho'$ always exists.
    For the rest of the proof, we focus on the case when the distribution on $(p, c)$ has atoms and assume that $P(\mu^*) > \rho'$.

    Fix $\e > 0$.
    Note that $P(\mu^* + \e) < \rho'$.
    We notice that bidding $\frac{c}{\mu^* + \e}$ in the Lagrangian problem with multiplier $\mu^*$ is near optimal when $\e$ is close to $0$:
    \begin{alignat}{3} \label{eq:calc:1}
        \Line{
            \Ex{\Big.\qty\big( c - \mu^* p ) \One{c \ge p (\mu^* + \e)}}
        }{=}{
            \Ex{\Big.\qty\big( c - \mu^* p ) \One{c \ge p \mu^*}}
            -
            \Ex{\Big.\qty\big( c - \mu^* p ) \One{c \in [p \mu^*, p (\mu^* + \e) ) }}
        }{}
        \nonumber\\
        \Line{}{\ge}{
            \Ex{\Big.\qty\big( c - \mu^* p ) \One{c \ge p \mu^*}}
            -
            \e \Pr{c \in (p \mu^*, p (\mu^* + \e) ) }    
        }{}
    \end{alignat}

    Let $\bb^*$ be the distribution that bids $\frac{c}{\mu^*}$ with probability $q$ and $\frac{c}{\mu^* + \e}$ with probability $1-q$.
    Because it holds $P(\mu^* + \e) < \rho' < P(\mu^*)$, we pick $q$ such that $(1-q)P(\mu^* - \e) + q P(\mu^*) = \rho'$.
    This makes makes
    \begin{alignat}{3} \label{eq:calc:2}
        \Line{}{}{
            \Ex{\Big.\qty\big( c - \mu^* p ) \One{\bb^* \ge p}}
        }{}
        \nonumber\\
        \Line{}{=}{
            q \Ex{\Big.\qty\big( c - \mu^* p ) \One{c \ge \mu^* p}}
            +
            (1-q) \Ex{\Big.\qty\big( c - \mu^* p ) \One{c \ge (\mu^* + \e) p}}
        }{}
        \nonumber\\
        \Line{}{\ge}{
            \Ex{\Big.\qty\big( c - \mu^* p ) \One{c \ge p \mu^*}}
            -
            \e \Pr{c \in (p \mu^* , p (\mu^* + \e) ) }
        }{\text{using \eqref{eq:calc:1}}}
    \end{alignat}
    
    We proceed to prove that $\bb^*$ achieves conversion probability that approximates the value of \eqref{eq:calc:single} as $\e \to 0$.
    Fix any bidding distribution $\bb$ that is feasible for the constrained problem, i.e., $\Ex{p \One{\bb \ge p}} \le \rho'$.
    Given that $\frac{c}{\mu^*}$ is optimal for the Lagrangian problem with multiplier $\mu^*$, its objective value is better than $\bb$:
    \begin{alignat*}{3}
        \Line{
            \Ex{\Big.\qty\big( c - \mu^* p ) \One{c \ge \mu^* p}}
        }{\ge}{
            \Ex{\Big.\qty\big( c - \mu^* p ) \One{\bb \ge p}}
        }{}
        \\
        \Line{}{=}{
            \Ex{ c \One{\bb \ge p}}
            -
            \mu^* \Ex{ p \One{\bb \ge p}}
        }{}
        \\
        \Line{}{\ge}{
            \Ex{ c \One{\bb \ge p}}
            -
            \mu^* \rho'
        }{}
    \end{alignat*}

    Combining the above with \eqref{eq:calc:2} and the fact that $\Ex{p \One{\bb^* \ge p}} = \rho'$ we get that
    \begin{equation*}
            \Ex{\big. c \One{\bb^* \ge p}}
            \ge
            \Ex{ c \One{\bb \ge p}}
            -
            \e \Pr{c \in (p \mu^* , p (\mu^* + \e) ) }
    \end{equation*}

    This proves the third bullet of the lemma since as $\e \to 0$, the probability of conversion by $\bb^*$ becomes almost the one of $\bb$ for any feasible $\bb$.
    To prove the second bullet, we prove that $\bb^*$ is the optimal solution when the support is finite.
    We do so by taking $\e$ to be small enough (put positive) so that $\Pr{c \in (p \mu^* , p (\mu^* + \e) ) } = 0$.
    This is possible since otherwise, a $(p, c)$ value in the support needs to satisfy $c - p \mu^* \in (0, \e)$ for all $\e > 0$, which is impossible.
\end{proof}

\subsection{Deferred Proof of Lemma \ref{lem:calc:approximation}} \label{ssec:app:calc:approximation}

We first restate the lemma.

\LemCalcApproximation*

\begin{proof}[Proof of \cref{lem:calc:approximation}]
    We will prove the lemma assuming that $\e < \frac{1}{36 m^2}$; otherwise, the lemma is trivially true since $R(\vec\bmu),R'(\vec\bmu)$,$C(\vec\bmu),C'(\vec\bmu) \in [0, 1]$.
    We first show the following claim.
    We use the notation $\pm x$ to denote the interval $[-x , x]$ and define all operations on it, e.g., $\frac{x \pm y}{1 \pm z} = [\frac{x - y}{1 + z}, \frac{x + y}{1 - z}]$ as long as $x,y \ge 0$ and $0 \le z < 1$.

    \begin{proposition} \label{cl:calc:prod}
        Fix an integer $m$.
        Let $\e \in (0, 1/m)$ and $q_1, q_2, \ldots, q_m \in [0, 1]$.
        Then for all $\l \in [m]$ it holds
        \begin{equation*}
            \prod_{i = 1}^\l (q_i \pm \e)
            \sub
            \prod_{i = 1}^\l q_i \pm \qty(\l \e + \l^2 \e^2)
            \sub
            \prod_{i = 1}^\l q_i \pm 2 \l \e
        \end{equation*}
    \end{proposition}
    
    \begin{proof}
        We use induction on $\l$.
        The claim is trivially true for $\l = 1$.
        Then for any $\l + 1$ we have
        \begin{align*}
            \prod_{i = 1}^{\l+1} (q_i \pm \e)
            \sub
            (q_{\l+1} \pm \e) \qty( \prod_{i = 1}^\l q_i \pm (\l \e + \l^2 \e^2) )
            & \sub
            \prod_{i = 1}^{\l+1} q_i \pm \qty( \l \e + \l^2 \e^2 + \e + \l \e^2 + \l^2 \e^3)
            \\
            & \sub
            \prod_{i = 1}^{\l+1} q_i \pm \qty( (\l+1) \e + (\l+1)^2 \e^2 )
        \end{align*}
        where in the last part we used that $\l^2 \e^2 + \l \e^2 + \l^2 \e^3 \le \l^2 \e^2 + \l \e^2 + \l \e^2 \le (\l+1)^2 \e^2$. This completes the proof.
    \end{proof}
    
    We now show that the time between conversions of $\vec\bmu$ is similar for both $W$ and $W'$.
    Using \cref{def:bench:additional}:
    \begin{alignat*}{3}
        \Line{
            L'(\vec\bmu)
        }{=}{
            \sum_{\l = 1}^{m-1} \prod_{i=1}^{\l-1} \qty\big( 1 - W'(\bmu_i) )
            +
            \frac{1}{W'(\bmu_m)} \prod_{i=1}^{m-1} \qty\big( 1 - W'(\bmu_i) )
        }{}
        \\
        \Line{}{\sub}{
            \sum_{\l = 1}^{m-1} \qty( \prod_{i=1}^{\l-1} \qty\big( 1 - W(\bmu_i) ) \pm 2m\e)
            +
            \frac{1}{W(\bmu_m) \pm \e} \qty( \prod_{i=1}^{m-1} \qty\big( 1 - W(\bmu_i) ) \pm 2m\e)
        }{\text{\cref{cl:calc:prod}}\\W(\bmu_\l) \sub W'(\bmu_\l) \pm \e}
        \\\Line{}{\sub}{
            \sum_{\l = 1}^{m-1} \prod_{i=1}^{\l-1} \qty\big( 1 - W(\bmu_i) ) \pm 2m^2\e
            +
            \qty( \frac{1}{W(\bmu_m)} \pm \frac{2\e}{\bar c^2} )
            \qty( \prod_{i=1}^{m-1} \qty\big( 1 - W(\bmu_i) ) \pm 2m\e)
        }{W(\bmu_m) = \bar c\\\e\le\frac{\bar c}{2}}
        \\
        \Line{}{\sub}{
            L(\vec\bb)
            \pm
            \qty(
                2 m^2 \e
                +
                \frac{2 m}{\bar c} \e
                +
                \frac{2}{\bar c^2} \e
                +
                \frac{4 m}{\bar c^2} \e^2
            )
            \sub
            L(\vec\bb)
            \pm 8 m^2 \e
        }{\e \le \frac{1}{2m} \le \frac{\bar c}{2} }
    \end{alignat*}

    Using \cref{ssec:app:facts}, for any $\l < m$ we have
    \begin{alignat}{3} \label{eq:calc:pi_l}
        \Line{
            \pi_\l'(\vec\bmu)
        }{=}{
            \frac{\prod_{i = 1}^{\l - 1} \qty( 1 - W'(\bmu_i) )}{L'(\vec \bmu)} 
        }{}
        \nonumber\\
        \Line{}{\sub}{
            \frac{\prod_{i = 1}^{\l - 1} \qty( 1 - W'(\bmu_i) )}{L(\vec \bmu) \pm 8 m^2 \e}
        }{}
        \nonumber\\
        \Line{}{\sub}{
            \frac{\prod_{i = 1}^{\l - 1} \qty( 1 - W'(\bmu_i) )}{L(\vec \bmu)}
            \pm
            16 m^2 \e
        }{L(\vec\bmu) \ge 1 \ge 16 m^2 \e }
        \nonumber\\
        \Line{}{\sub}{
            \frac{\prod_{i = 1}^{\l - 1} \qty( 1 - W(\bmu_i) ) \pm 2 m \e}{L(\vec \bmu)}
            \pm
            16 m^2 \e
        }{\text{\cref{cl:calc:prod}}}
        \nonumber\\
        \Line{}{\sub}{
            \frac{\prod_{i = 1}^{\l - 1} \qty( 1 - W(\bmu_i) )}{L(\vec \bmu)}
            \pm
            18 m^2 \e
        }{L(\vec\bmu) \ge 1, m \ge 1}
    \end{alignat}
     
    The same calculation as above give us
    \begin{alignat*}{3}
        \Line{
            W'(b_m)\pi_m'(\vec\bmu)
        }{=}{
            \frac{\prod_{i = 1}^{m - 1} \qty( 1 - W'(\bmu_i) )}{L'(\vec\bmu)}
            \sub
            \frac{\prod_{i = 1}^{m - 1} \qty( 1 - W(\bmu_i) )}{L(\vec\bmu)} 
            \pm
            18 m^2 \e
        }{}
        \\
        \Line{}{=}{
            W(\bb_m)\pi_m(\vec\bmu)
            \pm
            18 m^2 \e
        }{}
    \end{alignat*}

    Now we examine the average-time reward of $\vec\bmu$
    \begin{alignat*}{3}
        \Line{
            R'(\vec\bmu)
        }{=}{
            \sum_{\l = 1}^{m-1} \qty\big( \rfunc(\l) - \rfunc(\l-1) ) \pi_\l'(\vec\bmu)
            +
            \qty(\rfunc(m) - \rfunc(m-1)) \pi_m'(\vec\bmu) W'(b_m)
        }{}
        \\
        \Line{}{\sub}{
            \sum_{\l = 1}^{m-1} \qty\big( \rfunc(\l) - \rfunc(\l-1) ) \qty( \pi_l(\vec\bmu) \pm 18 m^2 \e)
            +
            \qty\big(\rfunc(m) - \rfunc(m-1)) \qty( \pi_m(\vec\bmu) W(\bb_m) \pm 18 m^2 \e)
        }{}
        \\
        \Line{}{\sub}{
            \qty( \sum_{\l = 1}^{m-1} \qty\big(\rfunc(\l) - \rfunc(\l-1)) \pi_l(\vec\bmu) )
            \pm
            18 m^2 \e
            +
            \qty\big( \rfunc(m) - \rfunc(m-1) ) \pi_m(\vec\bmu) W(\bb_m)
            \pm
            18 m^2 \e
        }{}
        \\
        \Line{}{\sub}{
            R(\vec\bmu)
            \pm
            36 m^2 \e
        }{}
    \end{alignat*}
    where in the second step we used that $\sum_{\l = 1}^{m-1} \qty(\rfunc(\l) - \rfunc(\l-1)) \pi_\l(\vec\bmu) \le \max_\l \qty(\rfunc(\l) - \rfunc(\l-1)) \le 1$ which follows from $\sum_\l \pi_\l \le 1$.
    We now analyze the expected average payment
    \begin{align*}
        C'(\vec\bmu)
        =
        \sum_{\l = 1}^m P'(\vec\bmu) \pi_\l'(\vec\bmu)
        \sub
        \sum_{\l = 1}^m \qty( P(\vec\bmu) \pm \e) \pi_\l'(\vec\bmu)
        \sub
        \sum_{\l = 1}^m P(\vec\bmu) \pi_\l'(\vec\bmu) \pm \e
    \end{align*}

    While we know a bound for $\pi_\l'(\vec\bmu)$ for $\l < m$, we have not proven one for $\pi_m'(\vec\bmu)$.
    We have
    \begin{align*}
        \pi_m'(\vec\bmu)
        = &
        \frac{\pi_m'(\vec\bmu) W'(\bmu_m)}{W'(\bmu_m)}
        \sub
        \frac{\pi_m'(\vec\bmu) W'(\bmu_m)}{W(\bmu_m) \pm \e}
        \\
        \sub &
        \frac{\pi_m'(\vec\bmu) W'(\bmu_m)}{W(\bmu_m)} \pm \frac{2}{\bar c^2} \e
        \sub
        \frac{\pi_m(\vec\bmu) W(\bmu_m)}{W(\bmu_m)}
        \pm
        \qty( \frac{2}{\bar c^2} + 18\frac{m^2}{\bar c} ) \e
        \\
        \sub &
        \pi_m(\vec\bmu)
        \pm
        20\frac{m^2}{\bar c} \e
    \end{align*}
    where in the second $\sub$ relation we used that $\e \le \frac{\bar c}{2}$, in the next one that $\pi_m'(\vec\bmu) W'(\bmu_m) \sub \pi_m(\vec\bmu) W(\bmu_m) \pm 18 m^2 \e$ and that $W(\bmu_m) = \bar c$.
    In the final one, we use that $m \ge \frac{1}{\bar c}$.

    Using the bound for all $\pi_\l'$ we have that
    \begin{alignat*}{3}
        \Line{
            C'(\vec\bmu)
        }{\sub}{
            \sum_{\l = 1}^m P(\vec\bmu) \pi_\l'(\vec\bmu) \pm \e
        }{}
        \\
        \Line{}{\sub}{
            \sum_{\l = 1}^{m-1} P(\vec\bmu) \qty( \pi_\l(\vec\bmu) \pm 18 m^2 \e )
            +
            P(\bmu_m) \qty( \pi_m(\vec\bmu) \pm 20\frac{m^2}{\bar c} \e )
            \pm
            \e
        }{}
        \\
        \Line{}{\sub}{
            \sum_{\l = 1}^m P(\vec\bmu) \pi_\l(\vec\bmu)
            \pm
            \qty( 
                18 m^3
                +
                20\frac{m^2}{\bar c}
                +
                1
            ) \e
            \sub
            C(\vec\bmu)
            \pm
            39 m^3 \e
        }{P(\mu) \le 1}
    \end{alignat*}

    This completes the proof.
\end{proof}

\subsection{Deferred Proof Lemma \ref{lem:calc:sample_error}} \label{ssec:app:calc:sample_error}

We first restate the lemma.

\LemCalcSampleError*

\begin{proof}[Proof of \cref{lem:calc:sample_error}]
    The empirical estimates of the two functions are
    \begin{align*}
        W_n(\mu)
        =
        \frac{1}{n} \sum_{i = 1}^n c_i \One{c_i \ge \mu p_i}
        \quad\text{ and }\quad
        P_n(\mu)
        =
        \frac{1}{n} \sum_{i = 1}^n p_i \One{c_i \ge \mu p_i}
    \end{align*}
    
    We first prove the claim about the $W(\cdot)$ function.
    First, for $i\in [n]$ we define $X_i = (p_i, c_i)$.
    Define the error of $W_n$ from $W$:
    \begin{equation*}
        f(X_1, \ldots, X_n)
        =
        \sup_{\mu \ge 0} \abs{ W_n(\mu) - W(\mu) }
        =
        \norm{ W_n - W }_\infty
        =
        \norm{ W_n - \Ex{W_n} }_\infty
    \end{equation*}
    where the expectation in the right most part is taken over the pairs $(p_i, c_i)$.
    The claim we want to make now is that $f(X_1, \ldots, X_n)$ is small with high probability.
    We first bound its expectation, using \cite[Theorem 23]{DBLP:journals/corr/KleinbergLST23}.
    Because for each $i$, the function $c_i \One{c_i \ge \mu p_i}$ is non-increasing in $\mu$ and takes values in $[0, 1]$ we get\footnote{\cite[Theorem 23]{DBLP:journals/corr/KleinbergLST23} requires that the functions are non-decreasing which we can get by taking the sum over over $i$ of $1 - c_i \One{c_i \ge \mu p_i}$.} that $\Ex{f(X_1, \ldots, X_n)} \le \order*{\nicefrac{1}{\sqrt n}}$.
    In addition, we notice that we can use McDiarmid's inequality on $f$ since $f$ has the bounded differences property: for any $x_1, \ldots, x_n = (p_1, c_1), \ldots, (p_n, c_n)$, any $i$, and any $x_i' = (p_i', c_i')$ it holds
    \begin{alignat*}{3}
        \Line{}{}{
            \hspace{-20pt}
            \abs{ f(x_1, \ldots, x_i, \ldots, x_n) - f(x_1, \ldots, x_i', \ldots, x_n) }
        }{}
        \\
        \Line{}{=}{
            \left|
                \sup_{\mu \ge 0}
                \abs{ \frac{1}{n} \sum_{j \in [n]} c_j \One{c_j \ge \mu p_j} - W(\mu) }
                -
            \right.
        }{}
        \\
        \Line{}{}{
            \quad\left.
                \sup_{\mu \ge 0}
                \abs{ \frac{1}{n} c_i \One{c_i \ge \mu p_i} + \frac{1}{n} \sum_{j \in [n]\setminus\{i\}} c_j \One{c_j \ge \mu p_j} - W(\mu) }
            \right|
        }{}
        \\
        \Line{}{\le}{
            \sup_{\mu \ge 0}
            \abs{
                \frac{1}{n} c_i \One{c_i \ge \mu p_i}
                -
                \frac{1}{n} c_i' \One{c_i' \ge \mu p_i'}
            }
            \le
            1
        }{}
    \end{alignat*}

    Using McDiarmid's inequality we get that for all $\d > 0$, with probability at least $1 - \d$ it holds
    \begin{equation*}
        f(X_1, \ldots, X_n)
        \le
        \Ex{f(X_1, \ldots, X_n)}
        +
        \sqrt{\frac{n}{2} \log\frac{1}{\d}}
        \le
        \order{\frac{1}{\sqrt n}}
        +
        \sqrt{\frac{1}{2 n} \log\frac{1}{\d}}
    \end{equation*}
    where the second inequality holds by the bound on $\Ex{f(X_1, \ldots, X_n)}$.
    This proves the bound for $W_n(\cdot)$.
    We can prove a similar bound for the approximation $P_n(\cdot)$ using the exact same steps.
    Using the union bound on these two bounds yields the lemma.
\end{proof}

\subsection{Linear Program for \texorpdfstring{$\optInf_m$}{the Infinite Horizon problem}} \label{ssec:app:calc:lp}

In this section, we re-formulate the optimization problem \eqref{eq:app:opt_inf_stat} into a linear program.
We do so only when the distribution of the prices and contexts has finite support.
Specifically, we consider that there are at most $n$ price/conversion rate pairs in the support: $\{ (p_i, c_i) \}_{i \in [n]}$.
As we showed in \cref{lem:calc:simple_single}, the optimal solution of maximizing the probability of winning with a conversion subject to an expected budget constraint is bidding $\frac{c}{\mu}$ for some $\mu \ge 0$.
For such a $\mu$, we overload our previous notation and define the probability of getting a conversion using $\mu$ as
\begin{equation*}
    W(\mu)
    =
    \Ex[p, c]{c \One{c \ge \mu p}}
    =
    \sum_{i=1}^n \qty( c_i \One{c_i \ge \mu p_i} \Pr[p, c]{(p, c) = (p_i, c_i)} )
\end{equation*}
and similarly we define the expected payment $P(\mu)$.
This naturally defines $\mu_i = \frac{c_i}{p_i}$ for every $(p_i, c_i)$ in the support since any other $\mu$ has the same effect as some $\mu_i$.
The only exception is bidding $0$ which might not be covered by some $\mu_i$; thus we defined $\mu_0 = \infty$, which corresponds to bidding $\frac{c}{\mu_0} = 0$.
We will use $\{\mu_i\}_{i \in [n]\cup \{0\}}$ as the ``actions'' of each state.

We use $\{ q_{\l i} \}_{\l \in [m], i\in[n]\cup\{0\}}$ as our decision variables.
For a state $\l$ and action $i$, $q_{\l i}$ stands for the occupancy measure of that state/action pair, i.e., is the probability of being in state $\l$ and using $\mu_i$.
This makes the stationary probability of state $\l$ to be $\pi_\l = \sum_i q_{\l i}$.
Substituting this in \eqref{eq:app:opt_inf_stat} we get
\begin{equation} \label{eq:calc:LP}
\begin{aligned}
    \optInf_m = \;
    & \max_{ q_{\l,i} \ge 0 } &&
    \sum_{\l=1}^m \rfunc(\l) \sum_{i=0}^n W(\mu_i) q_{\l i}
    &&
    \\
    & \textrm{such that} &&
    \sum_{\l=1}^m \sum_{i=0}^n P(\mu_i) q_{\l i} \le \rho
    &&
    \\
    & \textrm{where} &&
    \sum_{i=0}^n q_{1 i} = \sum_{\l = 1}^m \sum_{i=1}^n W(\mu_i) q_{\l i}
    &&
    \\
    & &&
    \sum_{i=0}^n q_{\l i} = \sum_{i=0}^n \qty\big( 1 - W(\mu_i) ) q_{\l-1, i}
    && \forall \l = 2, 3, \ldots, m-1
    \\
    & &&
    \sum_{i=0}^n q_{m i} = \sum_{\l=m-1}^m \sum_{0=1}^n ( 1 - W(\mu_i) ) q_{\l i}
    && 
    \\
    & &&
    \sum_{\l = 1}^m \sum_{0=1}^n q_{\l i} = 1
    &&
\end{aligned}
\end{equation}

We notice that the above is linear in the variables $q_{\l i}$.
In addition, there are $(n+1) m$ variables and $m + 2$ constraints.
We also note that the constraint of bidding $1$ in state $m$ can be encoded by adding an action that corresponds to such a bid, $\mu_{n+1} = 0$, and adding the constraint $\sum_{i=0}^n q_{m i} = 0$.
This proves that we can execute \cref{algo:online} in polynomial time since in every round $t$, the empirical distribution has support at most $t$.

\subsection{Deferred Proof of Theorem \ref{thm:calc}} \label{ssec:app:calc:theorem}

We now finally prove \cref{thm:calc}, which follows as a combination of \cref{lem:calc:approximation,lem:calc:sample_error} and Linear Program \ref{eq:calc:LP}.
We first restate the theorem.

\ThmCalc*

\begin{proof}[Proof of \cref{thm:calc}]
    We use the empirical estimates of \cref{lem:calc:sample_error}.
    Let $R'(\cdot)$ and $C'(\cdot)$ denote the resulting expected average reward and payment of these distributions.
    To keep consistent with the notation in the statement of \cref{thm:calc} we use vectors of mappings from contexts to bids as the proposed solution, $\vec\bb$, instead of vectors of multipliers, $\vec\bmu$ (which was the notation in \cref{lem:calc:approximation,lem:calc:sample_error}).
    We can calculate the optimal solution using the Linear Program of \eqref{eq:calc:LP}, since the empirical distributions have finite support.
    Let $\vec\bb^t$ be the resulting distribution.

    Assume that the error bounds of \cref{lem:calc:sample_error} hold, which happens with probability at least $1 - \d$.
    This means that
    \begin{equation*}
        C(\vec\bb^t)
        \le
        C'(\vec\bb^t)
        +
        \order{m^3 \sqrt{\frac{1}{t} \log\frac{1}{\d}}}
        \le
        \rho
        +
        \order{m^3 \sqrt{\frac{1}{t} \log\frac{1}{\d}}}
    \end{equation*}
    where the first inequality holds by \cref{lem:calc:sample_error,lem:calc:approximation} and the second hold by the fact that $\vec\bb^t$ is feasible for the cost $C'(\cdot)$.

    The calculation for the reward is a bit more complicated.
    Fix any $\vec\bb$ such that $C(\vec\bb) \le \rho$, implying also that $C'(\vec\bb) \le \rho + \e$ where $\e = \order{m^3 \sqrt{\frac{1}{t} \log\frac{1}{\d}}}$ (similar to the calculation above).
    This means that $\vec\bb$ is not necessarily feasible when $R'(\cdot)$ and $C'(\cdot)$ are the reward and payment.
    For this reason, we define $\vec\bb'$ as the combination of $\bb$ and the bidding that always bids $0$.
    Specifically, we define $\vec\bb'$ by taking a convex combination of the occupancy measures or $\vec\bb$ and the always bidding $0$ solution.
    We take $1 - \frac{\e}{\rho}$ times the occupancy measure of $\vec\bb$ and $\frac{\e}{\rho}$ the occupancy measure of always bidding $0$.
    This results in the following average payment of $\vec\bb'$:
    \begin{equation*}
        C'(\vec\bb')
        =
        \qty( 1 - \frac{\e}{\rho} ) C'(\vec\bb)
        \le
        \qty( 1 - \frac{\e}{\rho} ) \qty( \rho + \e )
        \le
        \rho
    \end{equation*}

    The above makes $\vec\bb'$ feasible.
    Because $\vec\bb^t$ is constrained to bid $1$ at state $m$, we cannot directly compare it to $\vec\bb'$.
    For this reason we need \cref{thm:bench:small_m}.
    This makes
    \begin{alignat*}{3}
        \Line{
            R(\vec\bb^t)
        }{\ge}{
            R'(\vec\bb^t) - \order{m^2 \sqrt{\frac{1}{t} \log\frac{1}{\d}}}
            \ge
            R'(\vec\bb') - \order{m^2 \sqrt{\frac{1}{t} \log\frac{1}{\d}}} - \frac{m}{T}
        }{\text{\cref{thm:bench:small_m} and}\\m \ge \frac{2}{\bar c \rho}\log T, C'(\vec\bb') \le \rho }
        \\
        \Line{}{\ge}{
            R(\vec\bb') - \order{m^2 \sqrt{\frac{1}{t} \log\frac{1}{\d}}}
        }{\text{\cref{lem:calc:approximation,lem:calc:sample_error}}, \frac{m}{T} \le \frac{m^2}{\sqrt t} }
        \\\Line{}{\ge}{
            \qty(1 - \frac{\e}{\rho}) R(\vec\bb) - \order{m^2 \sqrt{\frac{1}{t} \log\frac{1}{\d}}}
            \ge
            R(\vec\bb) - \order{\frac{m^3}{\rho} \sqrt{\frac{1}{t} \log\frac{1}{\d}}}
        }{\text{definition of } \vec\bb'}
    \end{alignat*}
    where the last inequality holds by definition of $\e = \order{m^3 \sqrt{\frac{1}{t} \log\frac{1}{\d}}}$.
    The above $R(\cdot)$ functions referred to when the per-round reward function was $\rfunc_m$.
    We get the result for the case when the per-round reward function is $\rfunc_M$ my using \cref{thm:bench:small_m} once more.
\end{proof}

\end{document}